\documentclass[runningheads]{llncs}

\usepackage{graphicx,amssymb,amsmath}
\usepackage{graphicx} 
\usepackage{enumitem}
\setlist[itemize]{noitemsep}
\usepackage{algorithm}
\usepackage{algpseudocode}
\usepackage{float}
\usepackage{subcaption}
\makeatletter
\let\llncs@subparagraph\subparagraph
\renewcommand\subparagraph{\@startsection{subparagraph}{5}{\parindent}%
  {3.25ex \@plus1ex \@minus.2ex}{-1em}{\normalfont\normalsize\bfseries}}
\makeatother

\usepackage[compact]{titlesec} % reduced space between headings and body text

\makeatletter
\let\subparagraph\llncs@subparagraph % put llncs's "don't use this" guard back
\makeatother
\usepackage{microtype}

\usepackage[title]{appendix}

\usepackage{booktabs}
\usepackage{multirow}
\usepackage[table]{xcolor} % For professional row shading
\usepackage{longtable} % Required for multi-page tables
\definecolor{lightgray}{gray}{0.95}

\begin{document}
\title{Efficient $K$-Visibility Query in Polygons}

%\author{Yeganeh Bahoo, Roni Sherman}
% \date{December 2025}

\author{Yeganeh Bahoo\inst{1}\orcidID{0000-0001-5349-4946} \and
Roni Sherman\inst{1}\orcidID{0009-0004-0542-3480}}
\index{Bahoo, Yeganeh}
\index{Sherman, Roni}
\authorrunning{Y. Bahoo et al.}
\institute{Toronto Metropolitan University, Toronto, ON M5B~2K3, Canada\\
\email{\{bahoo,roni.sherman\}@torontomu.ca} 
}

\maketitle

{\def\thefootnote{} \footnotetext{We acknowledge the support of the Natural Sciences and Engineering Research Council of Canada (NSERC).}}
\setcounter{footnote}{0}

\begin{abstract}
This paper investigates $k$-visibility, where a line of sight can penetrate up to $k$ obstacles. While computing the $k$-visibility polygon from a single query point is well-studied, existing spatial preprocessing approaches rely on full $O(n^2)$ line arrangements through all vertex pairs without characterizing the minimal set of topological boundaries. We present a refined cell decomposition framework that isolates the exact geometric events governing $k$-visibility: primary vertex horizon lines and secondary mutually critical hinge lines. We prove that this minimal set of partition lines yields a spatial decomposition of $\Theta(n^4)$ cells within which the combinatorial structure of the $k$-visibility polygon remains strictly invariant. By leveraging a combinatorial $\delta$-compression scheme across cell boundaries, we achieve an overall storage complexity of $\mathcal{O}(n^4)$ while supporting optimal $\mathcal{O}(\log n + m)$ query time to reconstruct explicit $k$-visibility polygons of size $m$. Our framework naturally extends to polygons containing holes.

\end{abstract}
% There are a number of algorithms for calculating the $k$-visibility polygon of a point. One algorithm was introduced by Martins et al.~\cite{Martins2009} with $O(n^2)$ complexity. Bahoo et al.~\cite{bahoo2020computing} presented a faster algorithm with $O(nlog(n))$ time complexity and an even faster one~\cite{yeganeh2019computational} with $O(kn)$ time complexity. Although these algorithms are efficient, no research has been done yet into finding query algorithms for the $k$-visibility polygon of a point, which would use more memory but will have faster query time.  

% Bose~\cite{bose2002efficient} introduced a $O(log(n) + h)$  algorithm for querying the $0$-visibility polygon of a point. This paper generalizes the query algorithm for $k$-visibility.  

\section{Introduction}

Given a simple polygon $P$ with $n$ vertices and an integer $k \ge 0$, two points $p, q \in P$ are said to be $k$-visible to each other if the open line segment $\overline{pq}$ intersects the boundary $\partial P$ at most $k$ times. The integer $k$ represents the \emph{visibility depth} (or the number of obstacle walls penetrated). The \emph{$k$-visibility polygon} of a query point $q$, denoted $\text{Vis}_k(q, P)$, is the set of all points in $P$ that are $k$-visible from $q$. A central goal in computational geometry is to preprocess $P$ into a spatial partition called a \emph{cell decomposition}—a division of $P$ into open, connected regions (cells) such that as a query point $q$ moves within a single cell, the combinatorial structure of its visibility polygon $\text{Vis}_k(q, P)$ remains invariant. Such decompositions allow efficient point-location query data structures to report $\text{Vis}_k(q, P)$ for any query point $q$.

Computing the visibility region of a point within a geometric domain is a classic problem in computational geometry, with foundational linear-time algorithms established for standard $0$-visibility in simple polygons with $n$ vertices~\cite{joe1987corrections,guibas1987linear}. Several algorithms have been proposed to calculate the $k$-visibility polygon of a point---the region visible when allowing up to $k$ obstacle intersections---from scratch, including an initial $O(n^2)$ time framework by Martins et al.~\cite{Martins2009}, an $O(n \log n)$ approach, and an asymptotically faster parameterized algorithm by Bahoo et al~\cite{bahoo2020computing}. Complementary research investigated time-space trade-offs under severe memory constraints~\cite{bahooTradeoff}. Efficient query data structures for $0$-visibility have also seen significant development. Early methods achieved $O(\log n + m)$ query time---where $m$ is the number of vertices in the reported visibility region---with $O(n^3)$ space~\cite{bose2002efficient}, while recent advancements utilize $O(n^{2+\epsilon})$ space to optimize visibility queries~\cite{bhore2026visibility}. In dynamic settings, extensive research has targeted kinetic visibility algorithms~\cite{aronov2002visible}. 
Concurrently, generalized $k$-cell decompositions have been studied to track shadow transitions and prevent the split or merge of unseen regions in pursuit-evasion games~\cite{generalizedKCell,bahoo2026exactgeneralizedkcelldecomposition}. While these structures share the goal of partitioning the domain into invariant visibility regions, key technical differences exist: the framework in~\cite{generalizedKCell} introduces redundant partition lines, whereas the exact decomposition in~\cite{bahoo2026exactgeneralizedkcelldecomposition} focuses specifically on shadow-level events for shadows defined strictly within $P$. In contrast, our new decomposition explicitly tracks boundaries where the generating tuples of the visibility windows change and accounts for shadows defined  outside of $P$ as well, yielding a tight, non-redundant cell arrangement tailored for exact point-location queries. This decomposition also takes into account windows of the $k$ and $k-2$ visibility polygons of vertices.

Beyond theoretical exploration, $k$-visibility structures have gained significant traction in modern applications, most notably in solving the wireless $k$-modem or illumination problem, where guards act as transmitters capable of penetrating up to $k$ walls~\cite{aicholzerModem}. Furthermore, recent developments have leveraged inverse $k$-visibility for real-time RSSI-based indoor geometric mapping and autonomous path-planning~\cite{kim2025inverse,sfw}. $k$-visibility has also been recently studied in relation to $M$-guarding~\cite{bahooMGuarding}.

While $k$-visibility query structures exist, constructing a tight, minimal cell decomposition without redundant partition boundaries remains unresolved. In their work exploring static structures, Bahoo et al.~\cite{bahoo2020computing} proposed a data structure requiring $O(n^5)$ space to support $k$-visibility queries. However, their model relies on a global, unrefined arrangement formed by extending lines through all $\binom{n}{2}$ pairs of vertices. While this guarantees topological completeness, it creates an oversized $O(n^4)$-cell arrangement populated by redundant, spurious boundaries where no actual visibility changes occur. Post-merging these cells is computationally expensive and fails to reveal the underlying structural geometry. Fully resolving this challenge requires a tight classification of the \emph{minimal} spatial arrangement that triggers discrete topological mutations---abrupt structural changes in the combinatorial state of the $k$-visibility polygon that a point sees. 

In this work, we generalize the cell-decomposition query framework of Bose et al.~\cite{bose2002efficient} and Guibas et al.~\cite{guibas1997visibility} to $k$-visibility environments. We establish the minimal set of primary and secondary horizon lines required for structural invariance, directly constructing a tight arrangement of $\Theta(n^4)$ cells without redundant lines. Furthermore, by introducing a delta-compressed dual arrangement tree that stores only $O(1)$ atomic mutations along tree edges, our scheme reduces total space complexity from $O(n^5)$ to $O(n^4)$ while supporting optimal $O(\log n + m)$ query reconstruction time.
\section{Problem Statement}
Let $P$ be a simple polygon with $n$ vertices, and let $k \ge 0$ be a fixed, non-negative integer parameter. The $k$-visibility query problem is defined as follows:

\begin{description}
    \item[Preprocessing Phase:] Given a polygon $P$ and the parameter $k$, construct a data structure $\mathcal{D}(P, k)$ that partitions $P$ into a cell-decomposition arrangement such that the combinatorial representation of the $k$-visible region of any point in a cell of the decomposition remains invariant.
    \item[Query Phase:] Given a query point $q \in P$, utilize the preprocessed data structure $\mathcal{D}(P, k)$ to efficiently locate the cell containing $q$ and retrieve the $k$-visibility polygon of $q$, denoted as $\text{Vis}_k(q)$.
\end{description}

We assume general position and that no three vertices can be colinear.

\section{The Geometric Structure of the $k$-Visibility Polygon}

For a query point $q \in P$, the $k$-visibility polygon $\text{Vis}_k(q)$ is a subset of $P$ whose boundary comprises two distinct types of geometric features: $k$-visible boundary chains and visibility windows. The $k$-visible boundary chains consist of continuous sequences of vertices and edges belonging to $\partial P$---the polygon boundary---or the bounding box (the box outside $P$ that bounds it) that possess a visibility depth of $k$ from $q$. These disconnected visible chains are topologically connected by visibility windows—line segments acting as artificial horizons that bound the occluded regions of the polygon. An example of a $2$-visibility polygon is illustrated in Figure~\ref{fig:2-vis-poly}, where $\text{Vis}_2(q)$ is represented by the shaded blue region.

 \begin{figure}[H]
\centering
\includegraphics[width=0.8\linewidth]{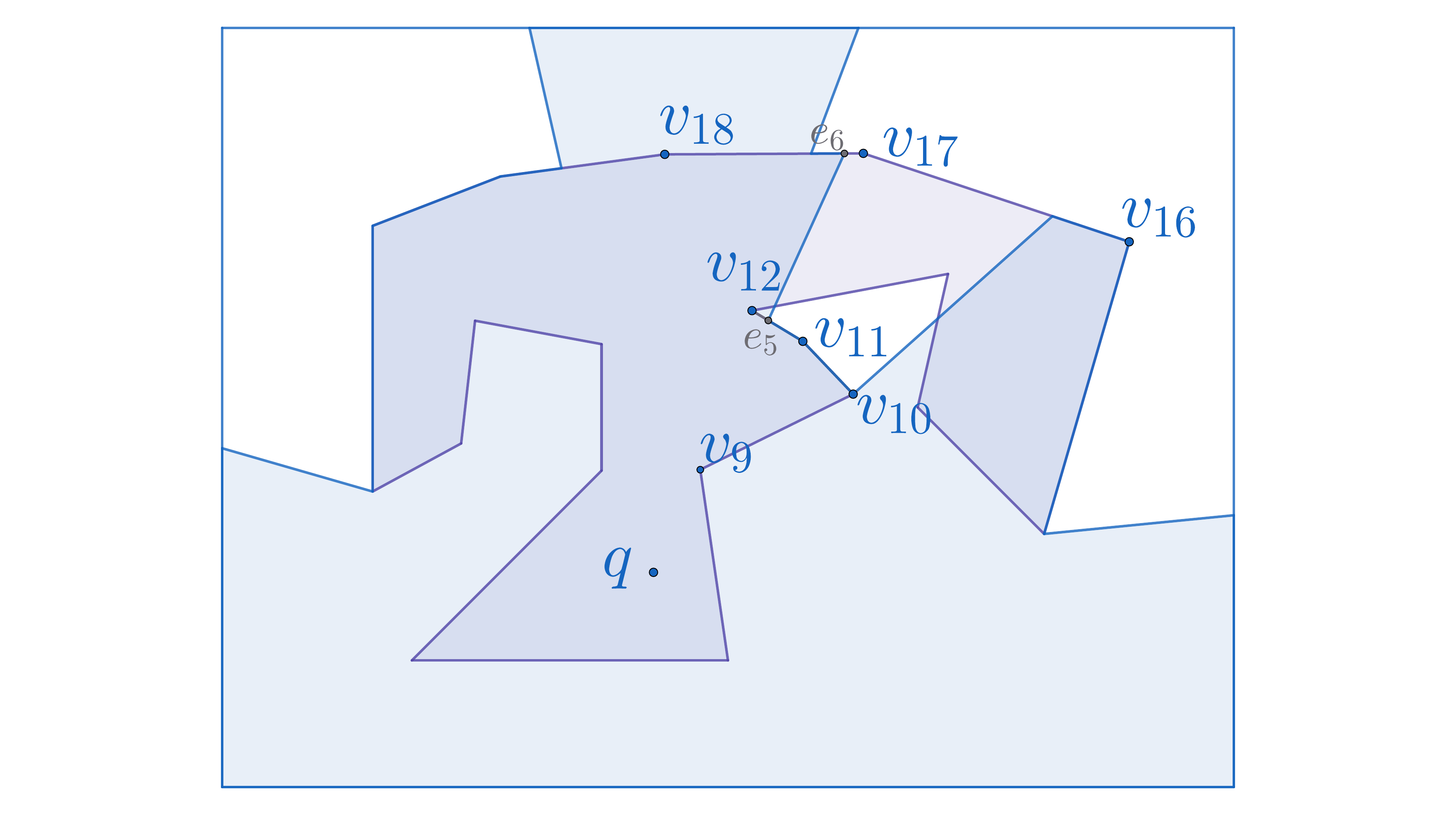}
\caption{$2$-visibility polygon (blue region) of point $q$}
\label{fig:2-vis-poly}
\end{figure}

\begin{definition}[Critical Vertex] Let $\ell_{g}$ be the line passing through a vertex $a$ and a point $b$. The vertex $a$ is critical to $b$ if both edges incident to $a$ lie strictly on the same side of $\ell_{g}$.  For example, $v_{9}$ in Figure~\ref{fig:2-vis-poly} is critical to $q$. 
Geometrically, a critical vertex $a$ acts as an obstacle reflex corner relative to $b$, causing the ray emanating from $b$ through $a$ to extend into the interior of $P$ beyond $a$. This extension forms a local topological boundary (a visibility window) that separates an \textbf{occluded shadow region}---a subregion of $P$ immediately behind vertex $a$ where points lose direct visibility from $b$---from the region that is visible from $b$. 
\end{definition}

% \begin{definition}
%     (Partition Line) The combinatorial structure of the visibility polygon of a query point $a$ changes along certain event lines, called partition lines. 
% \end{definition}

\begin{definition}[$k$-Visibility]
Let $P$ be a simple polygon and $\partial P$ denote its boundary. Two points $p, q \in P$ are $k$-visible to each other if the open line segment $\overline{pq}$ intersects $\partial P$ in at most $k$ distinct boundary crossings.
\end{definition}
\begin{definition}[$k$-Visibility Polygon]Given a point $q \in P$ and an integer $k \ge 0$, the \textbf{$k$-visibility polygon} of $q$ in $P$, denoted $Vis_k(q, P)$, is the set of all points $p \in P$ that are $k$-visible from $q$:$$Vis_k(q, P) = \{ p \in P \mid \vert{}\overline{pq} \cap \partial P\vert{} \le k \}$$\end{definition}

A non-polygonal edge of the $k$-visibility polygon of $q$ is called a \textit{window}. Combinatorially, every visibility window $W$ is uniquely generated by a single critical vertex $c \in P$, which acts as a geometric hinge that terminates the direct line of sight from the point $q$. A window $W$ is completely characterized by a three-tuple consisting of the following geometric elements:
\begin{enumerate}
    \item The \textbf{critical vertex} $c$ responsible for inducing the shadow or occlusion boundary (e.g., $v_{9}$ for window $e_5e_6$).
    \item The \textbf{proximal edge} $e_{\text{prox}}$ of the window, which corresponds to the boundary edge closest to the critical vertex $c$ where the obstruction originates (e.g., the edge $v_{11}v_{12}$ for the window $e_5e_6$ in Figure~\ref{fig:2-vis-poly}).
    \item The \textbf{distal edge} $e_{\text{dist}}$ of the window, which represents the background boundary edge of $P$ where the extended ray from $q$ through $c$ ends (e.g., the edge $v_{16}v_{17}$ for the window $e_5e_6$ in Figure~\ref{fig:2-vis-poly}).
\end{enumerate}

\begin{definition}
A visibility window $W$ within the $k$-visibility polygon $\text{Vis}_k(q)$ is a  line segment collinear to the query point and a critical vertex $c \in P$, uniquely characterized by a $c$, a proximal (start) edge, and a distal (end) edge. The window $W$ acts as a local topological boundary separating a $k$-visible subregion from an occluded shadow region. Formally, a neighborhood intersected by $W$ contains an interior subset belonging strictly to $\text{Vis}_k(q)$ on one side, and an exterior subset belonging to $P - \text{Vis}_k(q)$ on the opposite side.
\end{definition}

% \begin{definition}
% A window of a visibility polygon of a query point is defined by a vertex that is critical to the query point, a start edge and an end edge. To one side of the window is the visibility polygon, while to the other is a shadow. 
% \end{definition}

\section{Combinatorial Cell Data and Storage Representation}

To facilitate rapid query execution, the polygon $P$ is decomposed into a finite spatial arrangement of cells. Within each cell, the topological structure (the combinatorial visibility sequence) of the $k$-visibility polygon remains strictly invariant. However, while the ordered sequence of visible elements is fixed throughout a cell, the precise geometric coordinates of the visibility window endpoints depend on the exact position of the query point $q$. Thus, to output the exact spatial boundary of $\text{Vis}_k(q)$ at query time, a dynamic geometric reconstruction is required: the cell's stored combinatorial sequence is traversed, and the actual geometric positions of window endpoints are computed dynamically by shooting rays from $q$ through each critical vertex $c$ onto its corresponding proximal edge $e_{\text{prox}}$ and distal edge $e_{\text{dist}}$.

To support this dynamic query reconstruction, each cell explicitly stores its invariant combinatorial sequence, ordered counter-clockwise along $\partial P$. This sequence comprises an alternating chain of two distinct entry types:

% To facilitate rapid query execution, the polygon $P$ is decomposed into a finite spatial arrangement of cells. Within each cell, the topological structure of the $k$-visibility polygon remains strictly invariant. To reconstruct $\text{Vis}_k(q)$ dynamically for any query point $q$ situated within a given cell, the cell must explicitly store its underlying combinatorial sequence. This sequence is ordered counter-clockwise along $\partial P$ and comprises an alternating chain of two distinct entry types:

\begin{itemize}
    \item \textbf{A $k$-visible vertex:} A standalone vertex $v \in P$ whose visibility depth from any point within the cell is $k$.
    \item \textbf{A window generator tuple:} A structural descriptor representing a visibility window,  characterized as  $C = \langle c, e_{\text{prox}}, e_{\text{dist}} \rangle$ or $C = \langle c, e_{\text{dist}}, e_{\text{prox}} \rangle$.
\end{itemize}

The bounding box environment may also contribute to the sequence: a vertex $v$ of the bounding box is included if it is visible, and an edge of the bounding box may serve as the distal edge $e_{\text{dist}}$ within a window generator tuple if the occlusion ray terminates on it.  We assume that any visible parts of the bounding box are $k$-visible.

During the offline preprocessing phase, this structural sequence is computed for each cell by executing the topological labeling algorithm introduced by Martins et al.~\cite{Martins2009}. By selecting and traversing only the edge and vertex labels whose computed visibility depths satisfy the threshold equal to $k$ and connecting them via windows, the precise sequence of visible vertices and window parameters can be extracted and mapped directly to the cell.

\section{Spatial Partition Lines: Classification and Correctness}

The construction of the cell decomposition relies on identifying the exact geometric boundaries where the combinatorial structure of the $k$-visibility polygon undergoes a topological mutation. To guarantee that the stored structural sequence $\mathcal{S}(q)$ remains invariant within each cell, we overlay a set of critical curves classified into two distinct categories of visibility event lines:
\begin{enumerate}
    \item \textbf{Primary Vertex Horizon Lines (Type 1a):} Window edges within $P$ belonging to the individual $k$-visibility polygon of each vertex, governing the entry and exit of vertices from the shadow into the visibility sequence.
    \item \textbf{Primary Vertex Horizon Lines (Type 1b):} Window edges within $P$ belonging to the individual $k-2$-visibility polygon of each vertex, governing the entry and exit of vertices from the visibility polygon into the visibility sequence.

    \item \textbf{Secondary Hinge Transition Lines (Type 2):} Collinear extensions passing through pairs of mutually critical vertices, governing a combinatorial change in window generators.
\end{enumerate}

Below, we formally prove that these two categories of boundaries are both necessary and sufficient to isolate distinct combinatorial states.

\begin{theorem}
The cell decomposition strictly requires primary vertex horizon lines (Type 1a) derived from the window extensions of $\text{Vis}_k(v)$ for every vertex $v \in P$.
\end{theorem}

\begin{proof}
Let $v \in P$ be an arbitrary vertex, and let $W$ be a visibility window of $\text{Vis}_k(v)$ anchored by a critical vertex $c$. We define the corresponding horizon line, $H_{l1}$, as the active segment of $W$ restricted to $P$ ($H_{l1} = W \cap P$). Let $q_1, q_2 \in P$ be two query points in arbitrarily small open neighborhoods situated on opposite sides of $H_{l1}$, such that the segment $s = \overline{q_1 q_2}$ intersects $H_v$ transversally at a point $p_0$. By definition of $H_{l1}$, crossing this horizon boundary alters the obstacle penetration depth of the line segment connecting the query point to $v$:$$\vert{}\overline{q_1 v} \cap \partial P\vert{} = k + 2 \quad \text{and} \quad \vert{}\overline{q_2 v} \cap \partial P\vert{} = k$$(or vice versa). Consequently, 
% the indicator predicate for the membership of $v$ in the cell's combinatorial visibility sequence $\mathcal{S}(q)$ yields:
$$v \notin \mathcal{S}(q_1) \quad \text{and} \quad v \in \mathcal{S}(q_2)$$Since $\mathcal{S}(q_1) \neq \mathcal{S}(q_2)$, the sequence undergoes a discrete topological mutation at $p_0$. Therefore, the interior segment $H_{l1} = W \cap P$ acts as a mandatory cell partition boundary.
\end{proof}

  \begin{figure}[H]
\centering
\makebox[\textwidth][c]{%
  \scalebox{1.18}{%
    \begin{minipage}{\linewidth}
    \centering
    \begin{subfigure}[b]{.49\linewidth}
    \includegraphics[width=\linewidth]{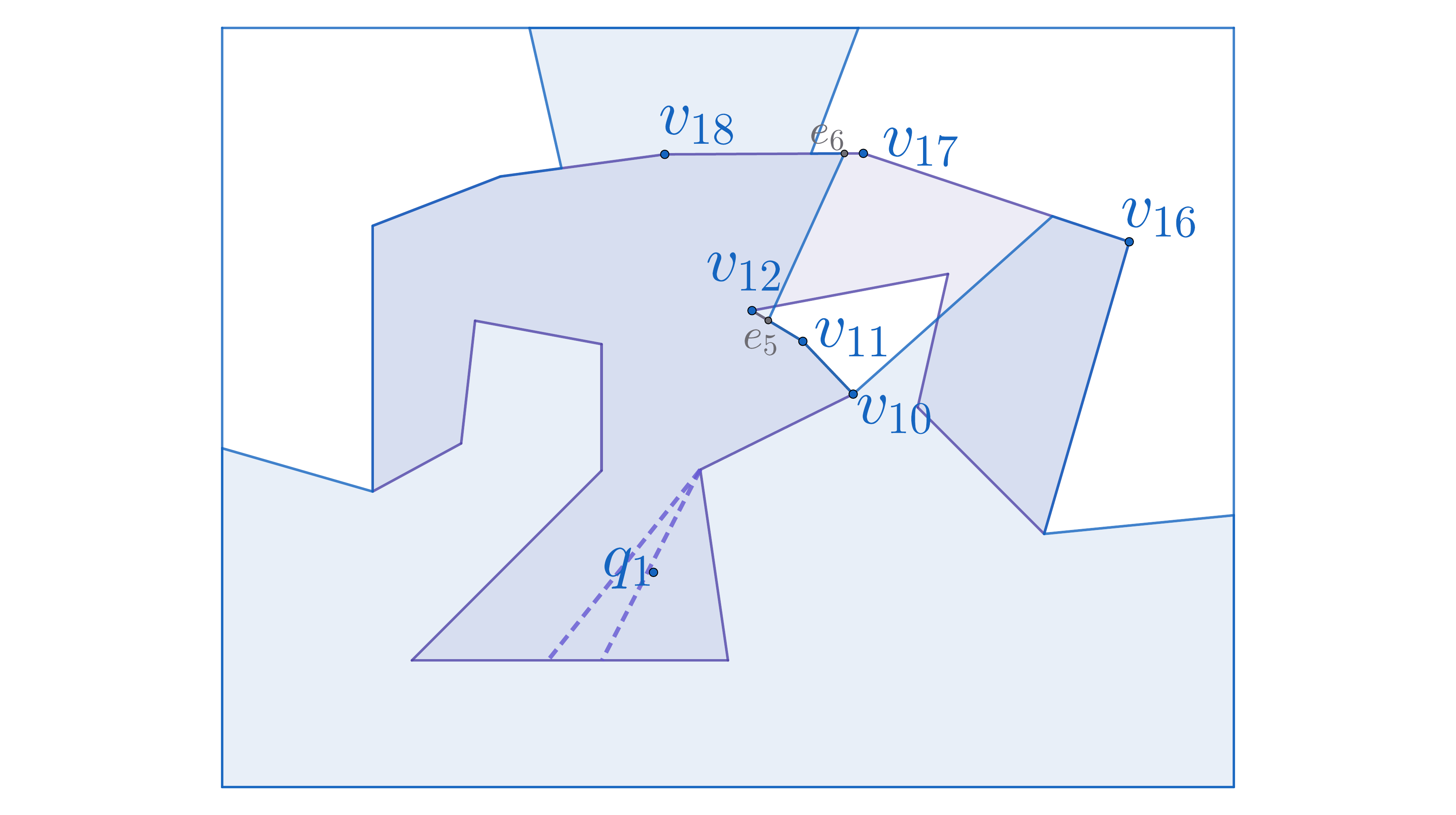}
    \caption{Right of $H_{l1}$}\label{fig:vertex-horizon-line1}
    \end{subfigure}
    \begin{subfigure}[b]{.49\linewidth}
    \includegraphics[width=\linewidth]{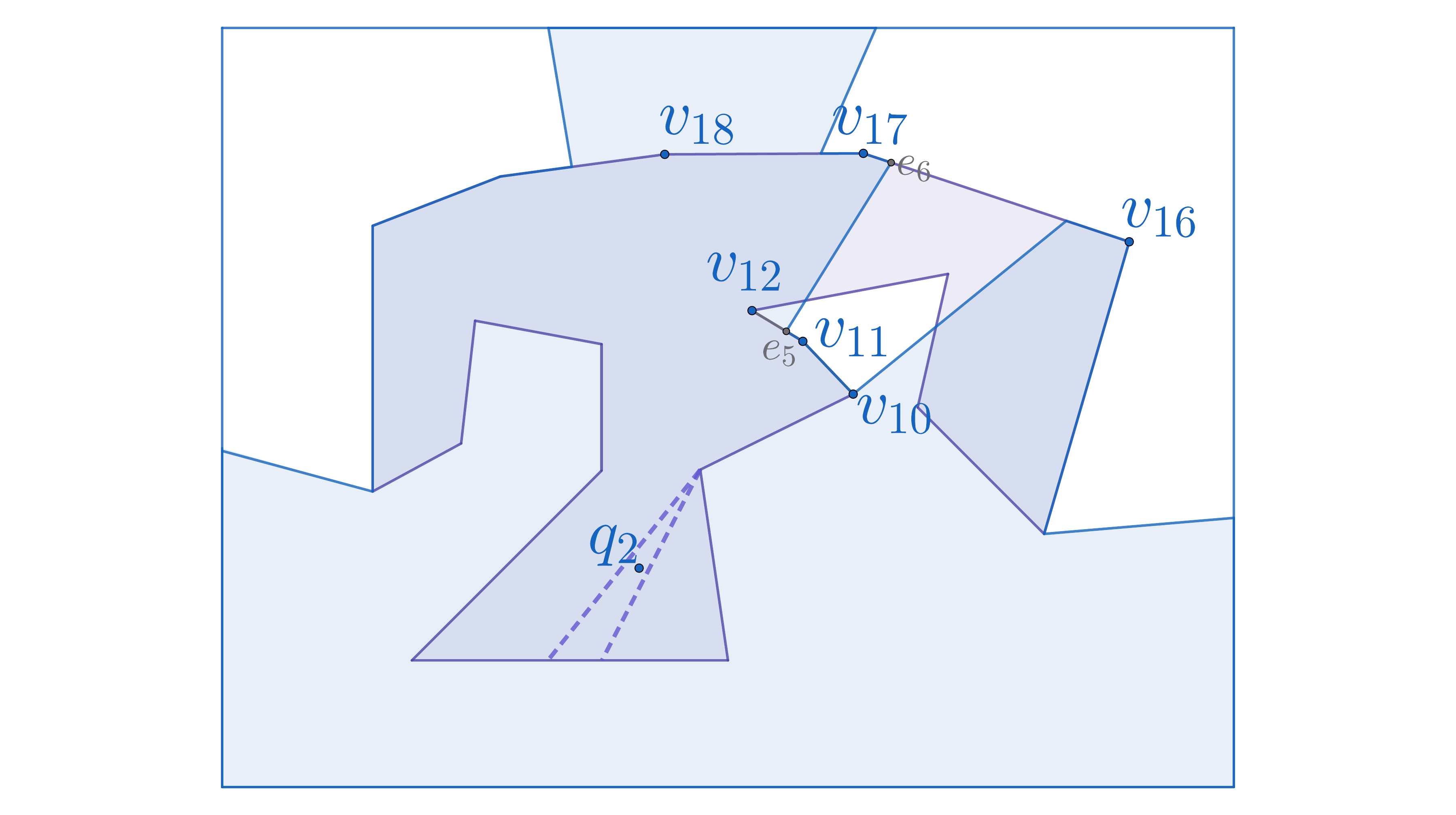}
    \caption{Left of $H_{l1}$}\label{fig:vertex-horizon-line2}
    \end{subfigure}

    \begin{subfigure}[b]{.49\linewidth}
    \includegraphics[width=\linewidth]{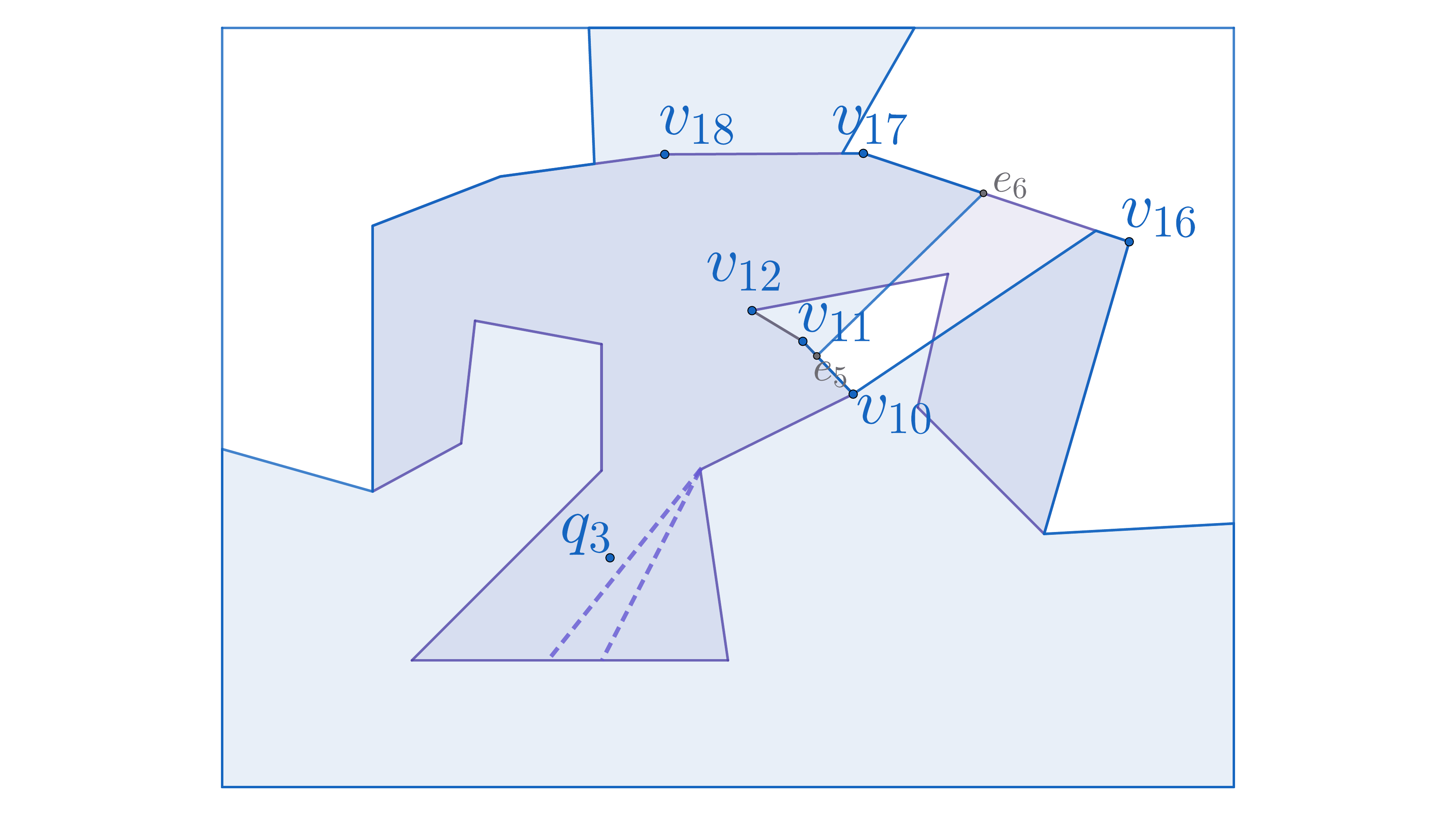}
    \caption{Left of $H_{l2}$}\label{fig:vertex-horizon-line3}
    \end{subfigure}
    
    \end{minipage}%
  }%
}
\caption{Vertex horizon lines of $\text{Vis}_{k}(v_{17})$ and $\text{Vis}_{k-2}(v_{11})$}
\label{fig:Vertex_horizon_line}
\end{figure}

 Consider the vertex $v_{17}$ and its Type 1a partition line, $H_{l1}$ (the dashed purple line closest to $q_{1}$). Crossing $H_{l1}$ from $q_1$ to $q_2$ (Figures~\ref{fig:vertex-horizon-line1}–\ref{fig:vertex-horizon-line2}) brings $v_{17}$ out of shadow into $2$-visibility, adding $v_{17}$ to $\mathcal{S}(q_2)$ and updating the distal edge of window $e_5 e_6$ from $v_{17}v_{18}$ to $v_{16}v_{17}$. Because crossing $H_{l1}$ alters the combinatorial sequence, Type 1a lines are strictly required as cell partition boundaries.

% Consider Figure~\ref{fig:vertex-horizon-line1} and Figure~\ref{fig:vertex-horizon-line2}. At $q_{1}$, the line of sight to $v_{17}$ intersects more than $2$ blocking edges, meaning $v_{17}$ is $2$-invisible ($v_{17} \notin \mathcal{S}(q_{1})$), and it is  within the shadow of $q_{1}$. Upon crossing $H_{l1}$ (dashed purple partition line closest to $q_1$) to position $q_{2}$, $v_{17}$ leaves the shadow and transitions into view, forcing it to enter the cell's stored structural sequence ($v_{17} \in \mathcal{S}(q_{2})$). This shift also causes the projected window ray ($e_{5}e_{6}$) to change its distal intersection from edge $v_{17}v_{18}$ to edge $v_{16}v_{17}$ (a window generator tuple changes). Because crossing $H_{l1}$ induces a discrete structural change in the visibility sequence, Type 1a lines must explicitly act as cell partition boundaries to preserve intra-cell sequence invariance.
\qed

\begin{theorem}
The cell decomposition strictly requires primary vertex horizon lines (Type 1b) derived from the window extensions of $\text{Vis}_{k-2}(v)$ for every vertex $v \in P$.
\end{theorem}

\begin{proof}
By definition of a visibility window of $\text{Vis}_{k-2}(v)$, crossing the horizon segment $H_{l2}$ (part of the window that is within $P$) alters the obstacle penetration depth of $v$ relative to the query point:
\begin{enumerate}
    \item At $q_2$,  $v$ lies on the boundary $\partial \text{Vis}_k(q_2)$ (with $v \notin \text{int}(\text{Vis}_k(q_2))$). Here, a vertex $c$  acts as a window generator anchoring a visibility window in $\text{Vis}_k(q_2)$ with proximal edge $e_{\text{prox}} = v_i v_{i+1}$. Thus, the structural sequence contains a generator tuple $C = \langle c, v_i v_{i+1}, e_{\text{dist}} \rangle \in \mathcal{S}(q_2)$.
    \item Upon crossing $H_{l2}$ to $q_3$, the obstacle depth drops to $k-2$ points, causing $v$ to transition directly into the interior of the visible region ($v \in \text{int}(\text{Vis}_k(q_3))$)). The adjacent proximal edge descriptor shifts from $v_i v_{i+1}$ to $v_{i-1}v_i$.
\end{enumerate}

Because $v \notin \text{int}(\text{Vis}_k(q_2))$ while $v \in \text{int}(\text{Vis}_k(q_3))$, and the recorded proximal edge descriptor updates from $v_i v_{i+1}$ to $v_{i-1}v_i$, the combinatorial visibility sequences satisfy $\mathcal{S}(q_2) \neq \mathcal{S}(q_3)$. Thus, $H_{l2}$ forms a necessary cell partition boundary.

    Consider the vertex $v_{11}$ and its Type 1b partition line, $H_{l2}$ (the dashed purple line closest to $q_{3}$). 
     Crossing $H_{l2}$ to $q_3$ (Figures~\ref{fig:vertex-horizon-line2}–\ref{fig:vertex-horizon-line3}) moves $v_{11}$ into the interior of $\text{Vis}_k(q_3)$, removing it from the boundary sequence $\mathcal{S}(q)$ and shifting the proximal edge descriptor from $v_{11}v_{12}$ to $v_{10}v_{11}$. This combinatorial shift necessitates Type 1b lines as partition boundaries.

\end{proof}
\begin{theorem}
    The cell decomposition does not require horizon lines from the windows of $\text{Vis}_{x}(v)$ where $x \leq k-4$ or $x \geq k + 2$ (when $v$ is not critical to the vertex responsible for the window).
\end{theorem}
The proof is in Appendix~\ref{appendix:unneededwindows}.

\begin{theorem}
The cell decomposition strictly requires the secondary hinge transition lines (Type 2) generated by mutually critical vertex configurations.
\end{theorem}

\begin{proof}
Let $v_1, v_2 \in P$ be two mutually critical vertices, and let $\ell_g = \text{line}(v_1, v_2)$ be the line passing through them. Let $H_{l3} = \overline{x_1 x_2} \subset \ell_g \cap P$ denote the active hinge transition line segment defined between boundary points $x_1$ and $x_2$. Let $a, b \in \text{int}(P)$ be two interior query points situated in small open neighborhoods on opposite sides of $H_{l3}$, such that the open line segment $\overline{ab}$ intersects $H_{l3}$ transversally at a single point.

Depending on the specific obstacle depth configuration ($Z$ and $W$ relative to $k$ as defined in section~\ref{section:additionalLines}), crossing $H_{l3}$ induces a discrete combinatorial mutation in the window generator tuples $C_{i} \in \mathcal{S}(q)$ associated with the critical pair. Without loss of generality, consider the configuration where a visibility window is projected beyond $v_1$ and $v_2$ (corresponding to the representative subcase in Figure~\ref{fig:CCS-generic-4} for W, i.e where $x_3x_4$ is):

\begin{enumerate}
    \item When $q = a$ (situated above $x_1 x_2$), the line-of-sight orientation establishes $v_2$ as the anchor for the first window generator tuple (corresponding to the dashed red line) $C_{1,a} = \langle v_2, e_{\text{prox}1}, e_{\text{dist}1} \rangle$, while $v_1$ anchors the second window generator tuple $C_{2,a} = \langle v_1, e_{\text{prox}2}, e_{\text{dist}2} \rangle$ (here $e_{\text{prox}2} = e_{\text{dist}1})$.
    \item Upon crossing $H_{l3}$ to $q = b$ (situated below $x_1 x_2$), the relative line-of-sight orientation flips. This structural inversion swaps the primary anchors across both windows: $C_{1,b} = \langle v_1, e_{\text{prox}1}, e_{\text{dist}1} \rangle$ and $C_{2,b} = \langle v_2, e_{\text{prox}2}, e_{\text{dist}2} \rangle$ (as illustrated in the representative subcase of Figure~\ref{fig:CCS-generic-4}).
\end{enumerate}

An analogous dual-tuple mutation occurs across $H_{l3}$ for all other valid subcases of mutually critical pairs (detailed in Appendix~\ref{appendix:mc-cases}). Because dynamic window reconstruction at query time evaluates the ray passing from $q$ through the specific anchor vertex recorded in $C_1$ and $C_2$, the stored generator tuples satisfy $\{C_{1,a}, C_{2,a}\} \neq \{C_{1,b}, C_{2,b}\}$, which implies $\mathcal{S}(a) \neq \mathcal{S}(b)$. Therefore, $a$ and $b$ cannot share the same topological cell, establishing that $H_{l3}$ is a mandatory cell partition boundary.

% Consider Figure~\ref{fig:CCS-generic-4}. Let $x_1x_2$ ($H_{l3}$) be a hinge partition line defined by the collinear alignment of two mutually critical vertices $v_1$ and $v_2$. As the query point crosses $H_{l3}$ from position $a$ to position $b$, the critical vertex generating the window changes, and so the window generator tuple changes from $(v_{2}, e_{prox}, e_{dist})$ to $(v_{1}, e_{prox}, e_{dist})$. Because $O(1)$ dynamic window reconstruction at query time relies on a critical vertex, the identity of the anchor ($v_1$ or $v_2$) must be invariant within any single cell. Thus, the collinear extensions where window generator tuples change are structurally mandatory. (Consider another type of change in Figure~\ref{fig:CCO-generic-4} in Appendix~\ref{appendix:mc-cases})

\qed
\end{proof}

\section{Hinge Transition Lines}
\label{section:additionalLines}

To determine where hinge transition lines are needed, we look at what happens to the visibility polygon in relation to two vertices that are mutually critical, similarly to~\cite{bahoo2026exactgeneralizedkcelldecomposition}, with the difference being that in that work, the shadow is defined only within the polygon, whereas here the shadow is also defined outside. To do this, we define the following: let $\ell{g}$ be the line through two critical vertices $v_1$ and $v_2$ (from left to right). Let $x_1$ be a point to the left of $v_1$ on $\ell_{g}$ and $x_2$ be the point on $\ell_{g}$ immediately subsequent to $x_1$ also on the polygon. Let $x_4$ be a point after $v_2$ on $\ell_{g}$ that touches the polygon and $x_3$ the point immediately preceding it which also the touches the polygon. Let the query point be denoted $a$ when it is above $x_1x_2$ and $b$ when it is below. Let $\ell_{b}$ be the line through the query point through $v_1$. Let $\ell_{r}$ be the line from the query point through $v_2$. To see where the combinatorial structure of the $k$-visibility polygon changes due to the two critical vertices, we vary values of $k$ and obstruction count. For this, we use two variables, the first one is $Z$, which is the number of edges between $x_1$ and $v_2$ in terms of $k$ not including the edges incident to $v_1$ and $v_2$, and $W$, which is the number of edges between $x_1$ and $x_4$ in terms of $k$ not including the edges incident to $v_1$ and $v_2$. There are eight basic cases of how two critical vertices may be arranged:

\begin{itemize}
\renewcommand{\labelitemi}{$\bullet$}
\setlength{\itemsep}{0pt}
    \setlength{\parskip}{0pt}
    \item CCS (Convex - Convex - Same)
    \item CCO (Convex - Convex - Opposite)
    \item CRS (Convex - Reflex - Same)
    \item CRO (Convex - Reflex - Opposite)
    \item RCS (Reflect - Convex - Same)
    \item RCO (Reflex - Convex - Opposite)
    \item RRS (Reflex - Reflex - Same)
    \item RRO (Reflex - Reflex - Opposite)
\end{itemize}

There are also four more special cases  (CC, CR,
RC, RR). These occur when $v_1$ and $v_2$ are adjacent to each other. These cases are necessary because the visibility polygon behaves differently compared to the first eight basic cases.

\paragraph{Completeness of Case Analysis.}
To establish that the 8 generic subcases (along with the 4 adjacent-vertex special cases) exhaustively capture all secondary topological events, we observe the following structural properties:
\begin{enumerate}
    \item \textbf{General Position Assumption:} Assuming no three vertices of $P$ are collinear guarantees that any partition line lies on a line $\ell_g$ that passes through at most two vertices ($v_1, v_2$).
    \item \textbf{Decoupling Primary vs. Secondary Events:} Single-vertex critical events—such as changes to a window's proximal ($e_{\text{prox}}$) or distal ($e_{\text{dist}}$) edge —are already governed by primary horizon lines generated by the $k$-visibility and $(k-2)$-visibility polygons of individual vertices. If a line passes through a vertex and another vertex that is not critical to it, no  mutation occurs upon crossing $\ell_g$.
    \item \textbf{Exhaustive Parameter Bounding for Mutually Critical Pairs:} For mutually critical pairs, the obstruction variables $Z$ and $W$ represent boundary crossings based on a constant $k$ that change in increments of $2$. The parameter space is bounded at the lower extreme when $v_2$ and $x_4$ are fully visible and at the upper extreme when they are permanently occluded from both query positions $a$ and $b$. Because all intermediate state transitions are either analyzed or ruled out by parity constraints, no topological mutations are omitted.
\item \textbf{Geometric Coincidence and Line Sufficiency:} Certain secondary transition lines coincide geometrically with primary horizon lines (e.g., when a vertex comes in and out of a shadow region, and both window generator tuple also change) or with other secondary cases (e.g., when one case shows a mutation in $e_{\text{prox}}$ while a complementary case shows a mutation in $e_{\text{dist}}$ along the same window). While analyzed separately for combinatorial completeness, identifying either condition is sufficient to instantiate the partition line in the spatial arrangement, ensuring no line is omitted or spuriously duplicated.
\end{enumerate}

Here we analyze only the case CCS, and the rest of the cases can be found in the appendix, summarized in   Table~\ref{tab:cases-comprehensive}. The subcases for for each case are different from those in~\cite{bahoo2026exactgeneralizedkcelldecomposition}.

\subsection{CCS}
\label{section:CCS} 

\begin{lemma}
\label{lemma:ccs}
    When $Z = k - 1$, $W = k$, $Z = k - 3$, $W = k - 2$ or $W = k - 4$, the combinatorial structure of the visibility polygon changes as the query point moves from $a$ to $b$, so there has to be a partition line at $x_1x_2$. 
\end{lemma}

\begin{proof}
Consider the following subcases, which covers all possible values of $Z$ and $W$. Since any continuous ray traversing an obstacle edge must both enter and exit the obstacle interior to return to the polygon interior, $Z$ and $W$ change in increments of 2. Consequently, intermediate values such as $Z = k + 2$ are structurally impossible.
\begin{itemize}
    \item $Z \geq k + 3$, $W \geq k + 4$: In this case, the vertex $v_{2}$ and its local neighborhood are completely occluded from the query region, possessing a visibility depth strictly greater than $k$. The same occlusion applies to the segment $x_{3}x_{4}$. Consequently, crossing the line $\ell{g}$ between $x_1$ and $x_2$ induces no topological change to the boundaries of the $k$-visibility polygon within these regions, rendering a partition line unnecessary.
    
    \item $Z = k + 1$, $W = k + 2$: See Figure~\ref{fig:CCS-generic-0}. Here, both $v_{2}$ and $x_{3}x_{4}$ reside entirely within an occluded shadow region relative to the threshold $k$, both when the query point is at $a$ and $b$, so no structural mutation occurs within the $k$-visibility polygon.
\end{itemize}

     \begin{figure}[H]
\centering
\makebox[\textwidth][c]{%
  \scalebox{0.8}{%
    \begin{minipage}{\linewidth}
    \centering
    \begin{subfigure}[b]{.49\linewidth}
    \includegraphics[width=\linewidth]{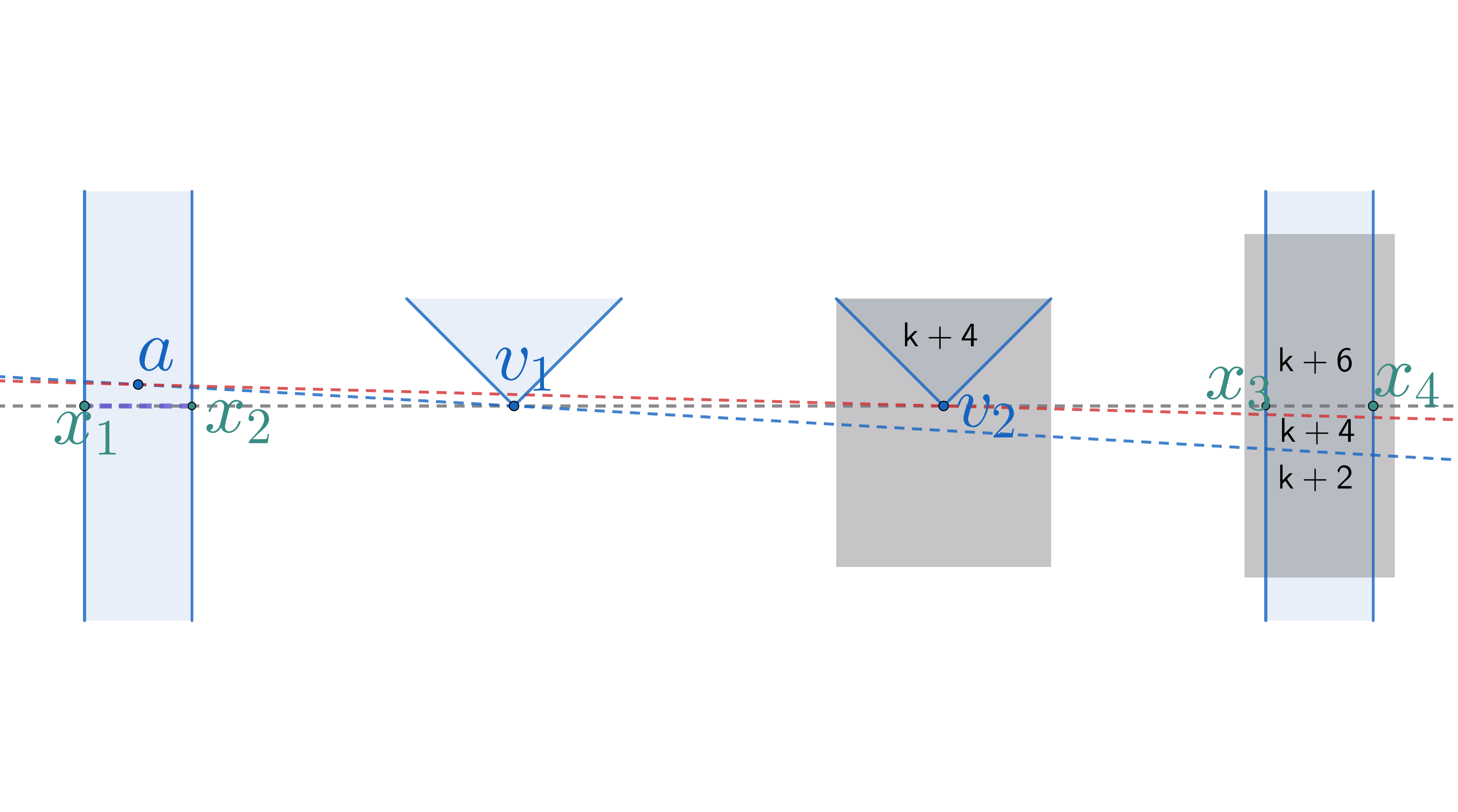}
    \caption{Above $l_{g}$}\label{fig:CCS-genericA0}
    \end{subfigure}
    \begin{subfigure}[b]{.49\linewidth}
    \includegraphics[width=\linewidth]{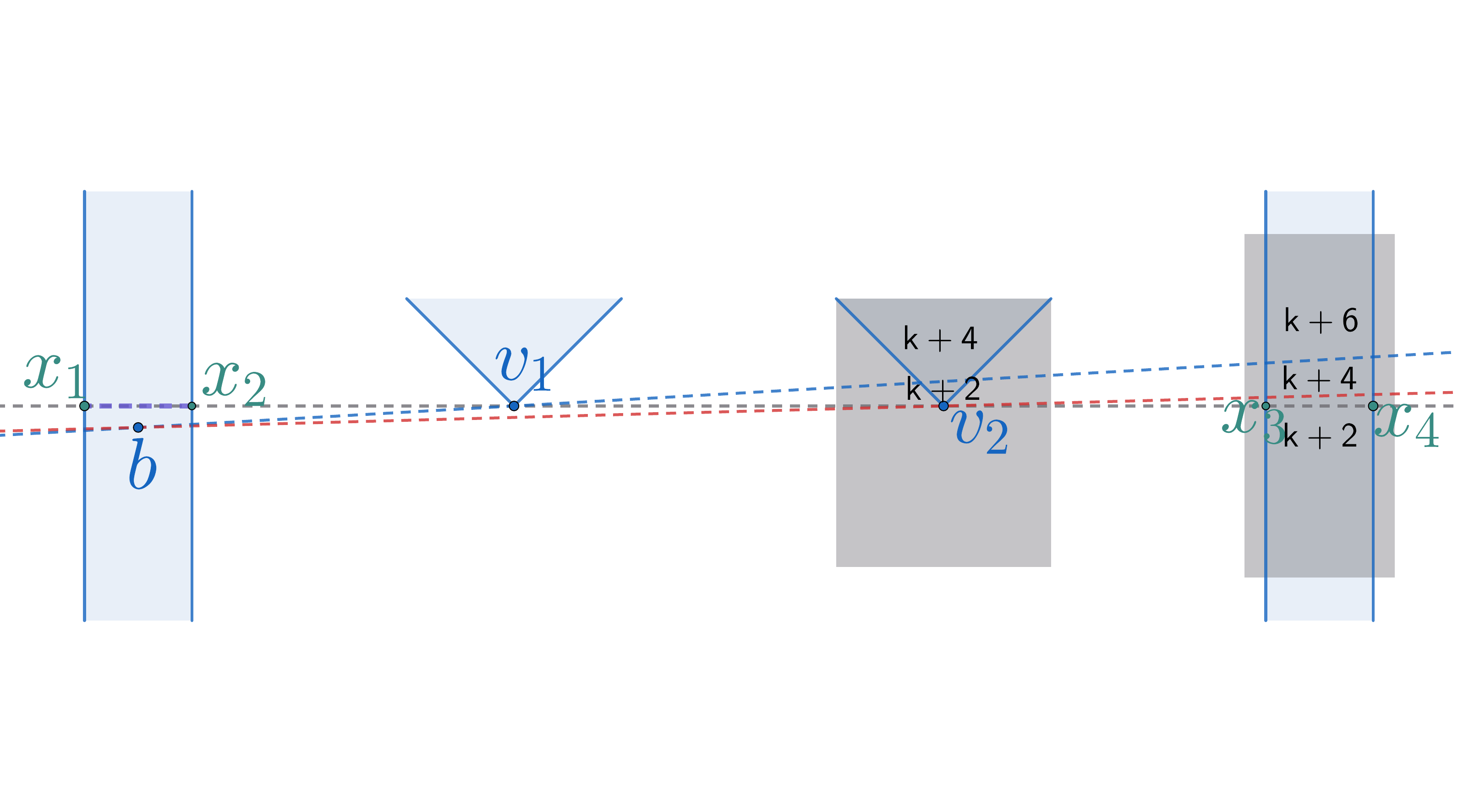}
    \caption{Below $l_{g}$}\label{fig:CCS-genericB0}
    \end{subfigure}
    \end{minipage}%
  }%
}
\caption{CCS; $Z = k + 1$, $W = k + 2$}
\label{fig:CCS-generic-0}
\end{figure}
\begin{itemize}
   
    \item $Z = k - 1$: See Figure~\ref{fig:CCS-generic-2}. When the query point is positioned at $a$ (above $\ell_g$), a visibility window is projected along the ray $l_{b}$. When the query point is at $b$ (below $\ell{g}$), the window along $l_b$ terminates on the edge incident to  $v_{2}$, while a new visibility window originates from $v_{2}$ along the ray $l_{r}$. This represents a structural mutation of the $k$-visibility polygon, dictating that the segment $x_{1}x_{2}$ must be explicitly integrated as a partition line.
    
    \item $W = k$: See Figure~\ref{fig:CCS-generic-2}. At point $a$, the local visibility boundary contains a window generated by the critical vertex $v_{1}$ that terminates on the distal edge near $x_{4}$. Upon crossing to point $b$, the generator of this window transitions to the critical vertex $v_{2}$ while still terminating near $x_{4}$. Because the stored window generator tuple changes, $x_{1}x_{2}$ acts as a necessary hinge transition boundary.

    \item $Z = k - 3$. See Figure~\ref{fig:CCS-generic-4}. When at $a$, the window near $v_{2}$ is defined by $v_{2}$. At $b$, the window near $v_{2}$ is defined by $v_{1}$. Therefore, there needs to be a partition line on $x_{1}x_{2}$.
\item $W = k - 2$. See Figure~\ref{fig:CCS-generic-4}. At $a$, the window near $x_{4}$ is defined by $v_{2}$. At $b$, the window is defined by $v_{1}$. Therefore, there needs to be a partition line on $x_{1}x_{2}$.

 \item $Z = k - 5$. See Figure~\ref{fig:CCS-generic-6}. The neighbourhood of $v_2$ is entirely visible at both $a$ and $b$, so no partition line is needed.
\item $W = k - 4$. See Figure~\ref{fig:CCS-generic-6}. At $a$, the window near $x_{4}$ is defined by $v_{2}$. At $b$, the window is defined by $v_{1}$. Therefore, there needs to be a partition line on $x_{1}x_{2}$.

\item $Z \leq k-7$. The entire region near $v_{2}$ is visible, so nothing happens to the visibility polygon there. 
\item $W \leq k - 6$. The entire region near $x_{3}x_{4}$ is visible, so nothing happens to the visibility polygon there. 
\end{itemize}

 \begin{figure}[H]
\centering
\begin{subfigure}[b]{.49\linewidth}
\includegraphics[width=0.9\linewidth]{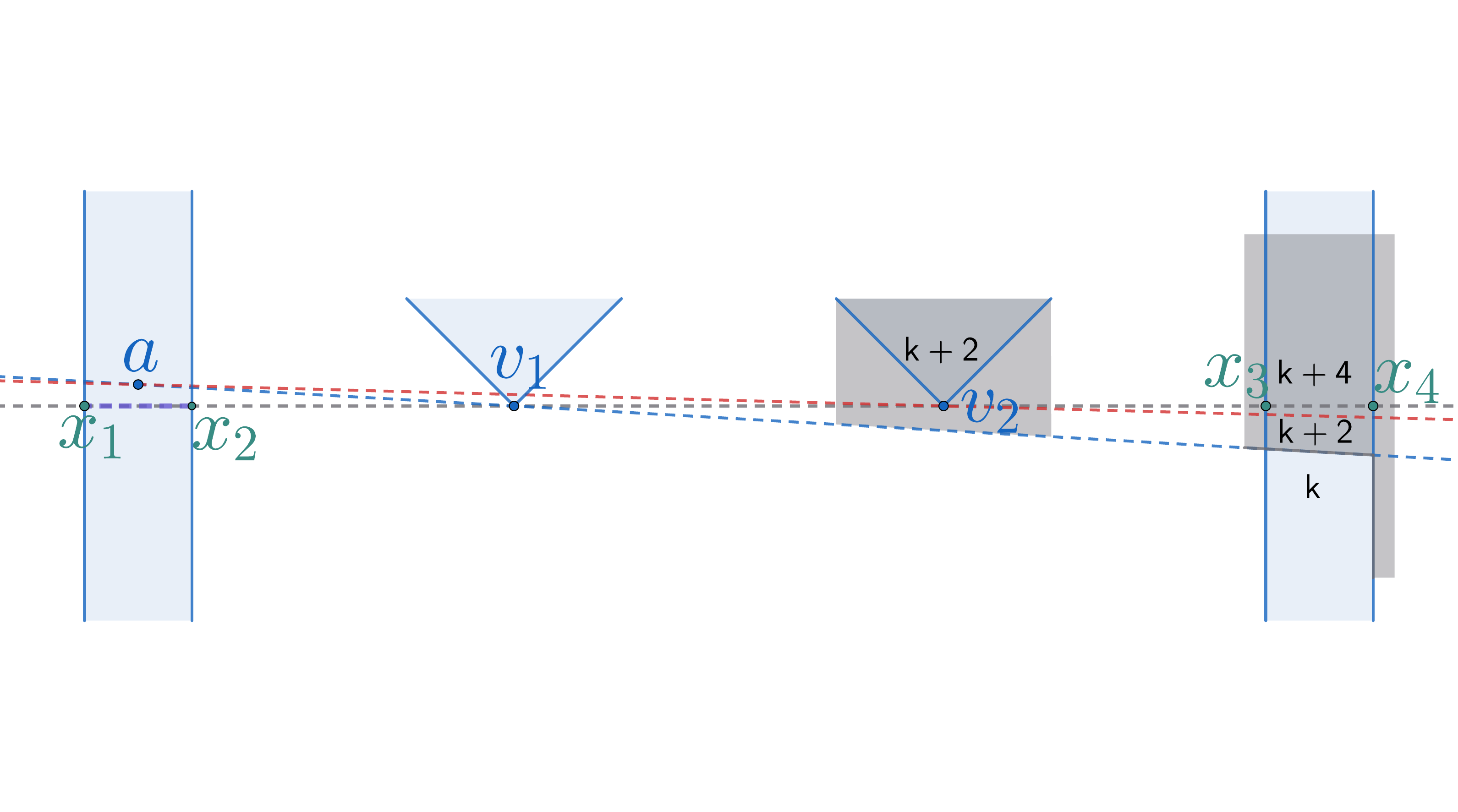}
\caption{Above $l_{g}$}\label{fig:CCS-genericA2}
\end{subfigure}
\begin{subfigure}[b]{.49\linewidth}

\includegraphics[width=0.85\linewidth]{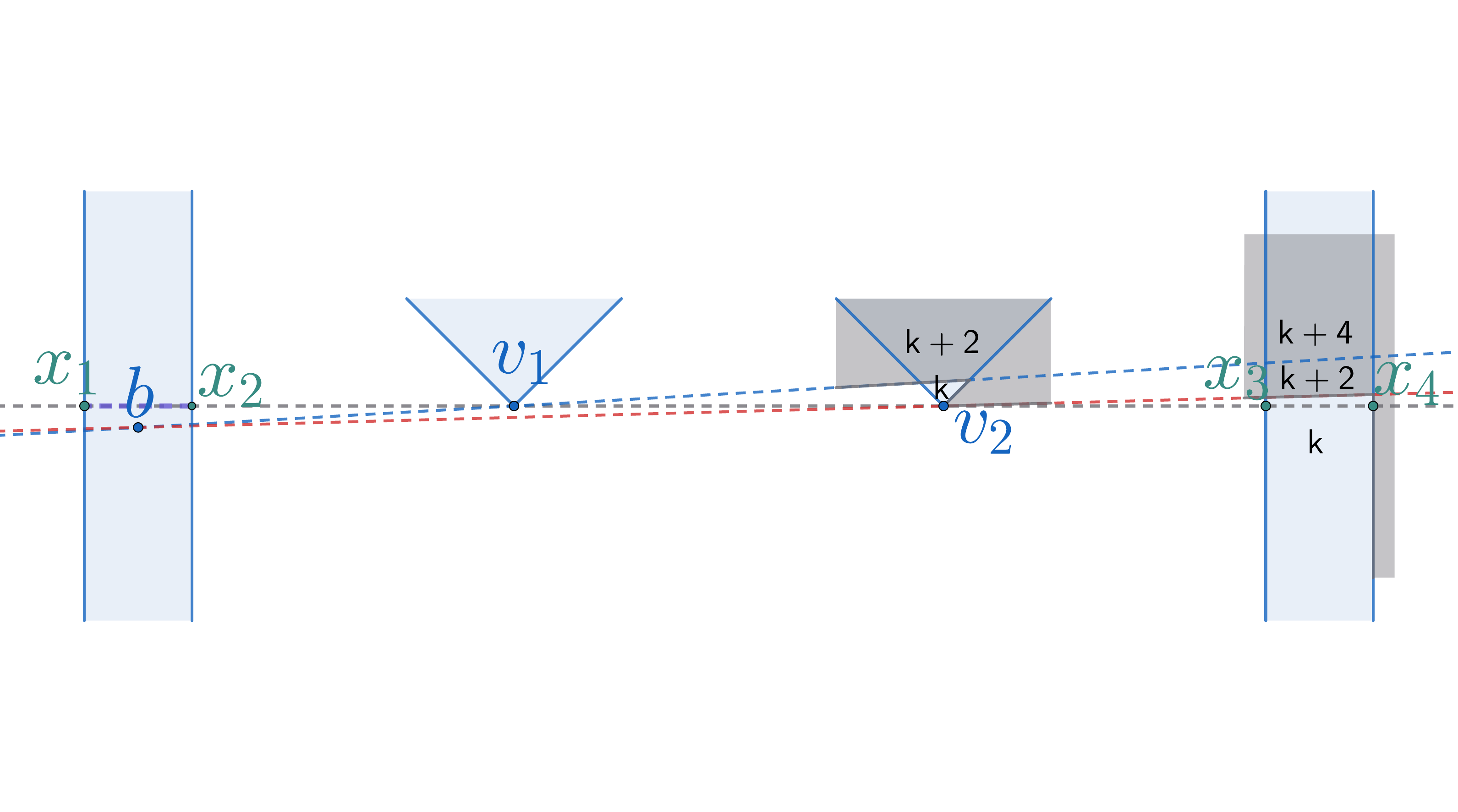}
\caption{Below $l_{g}$}\label{fig:CCS-genericB2}
\end{subfigure}

\caption{CCS; $Z = k - 1$, $W = k$}
\label{fig:CCS-generic-2}
\end{figure}

 \begin{figure}[H]
\centering
\begin{subfigure}[b]{.49\linewidth}
\includegraphics[width=0.85\linewidth]{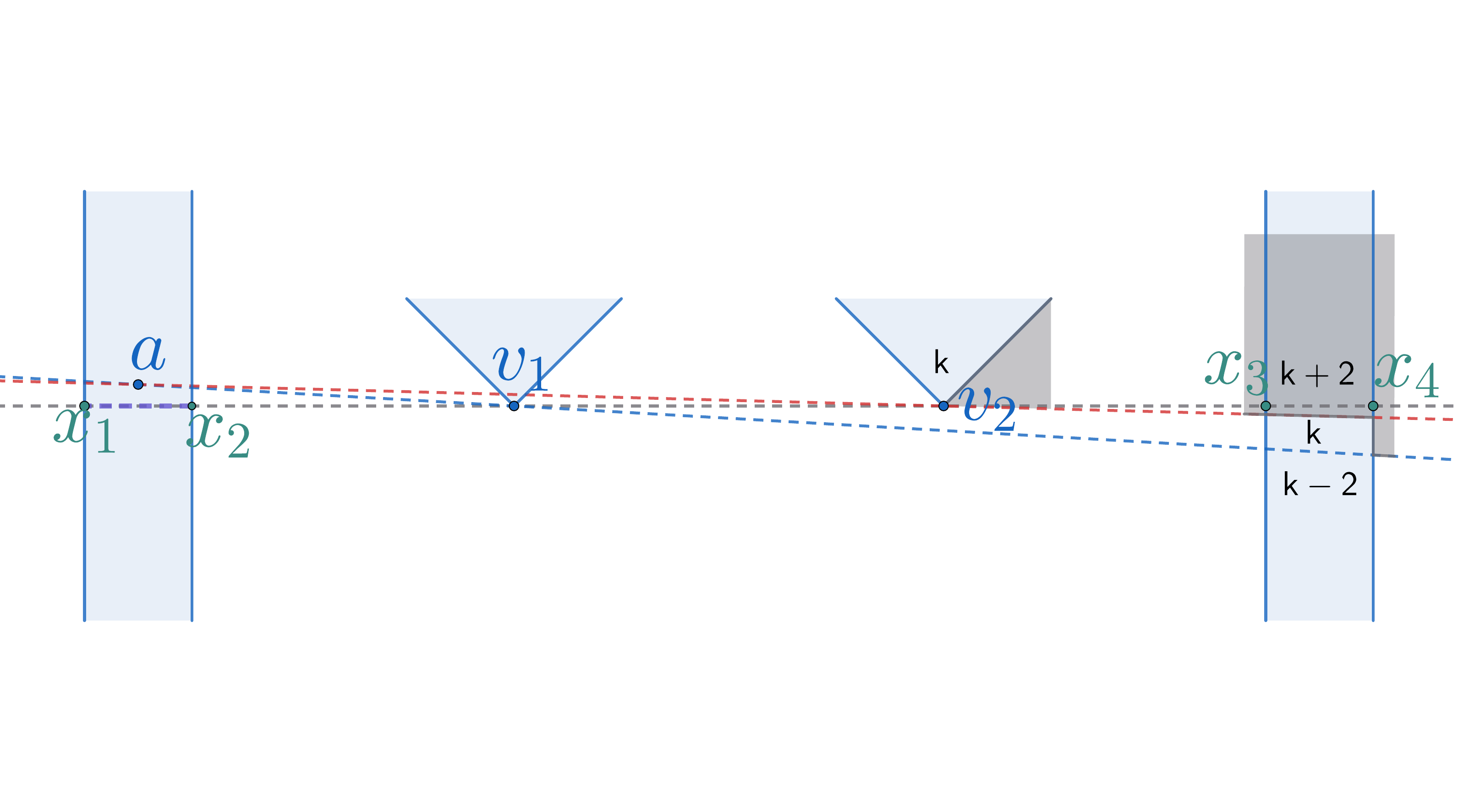}
\caption{Above $l_{g}$}\label{fig:CCS-genericA4}
\end{subfigure}
\begin{subfigure}[b]{.49\linewidth}

\includegraphics[width=0.85\linewidth]{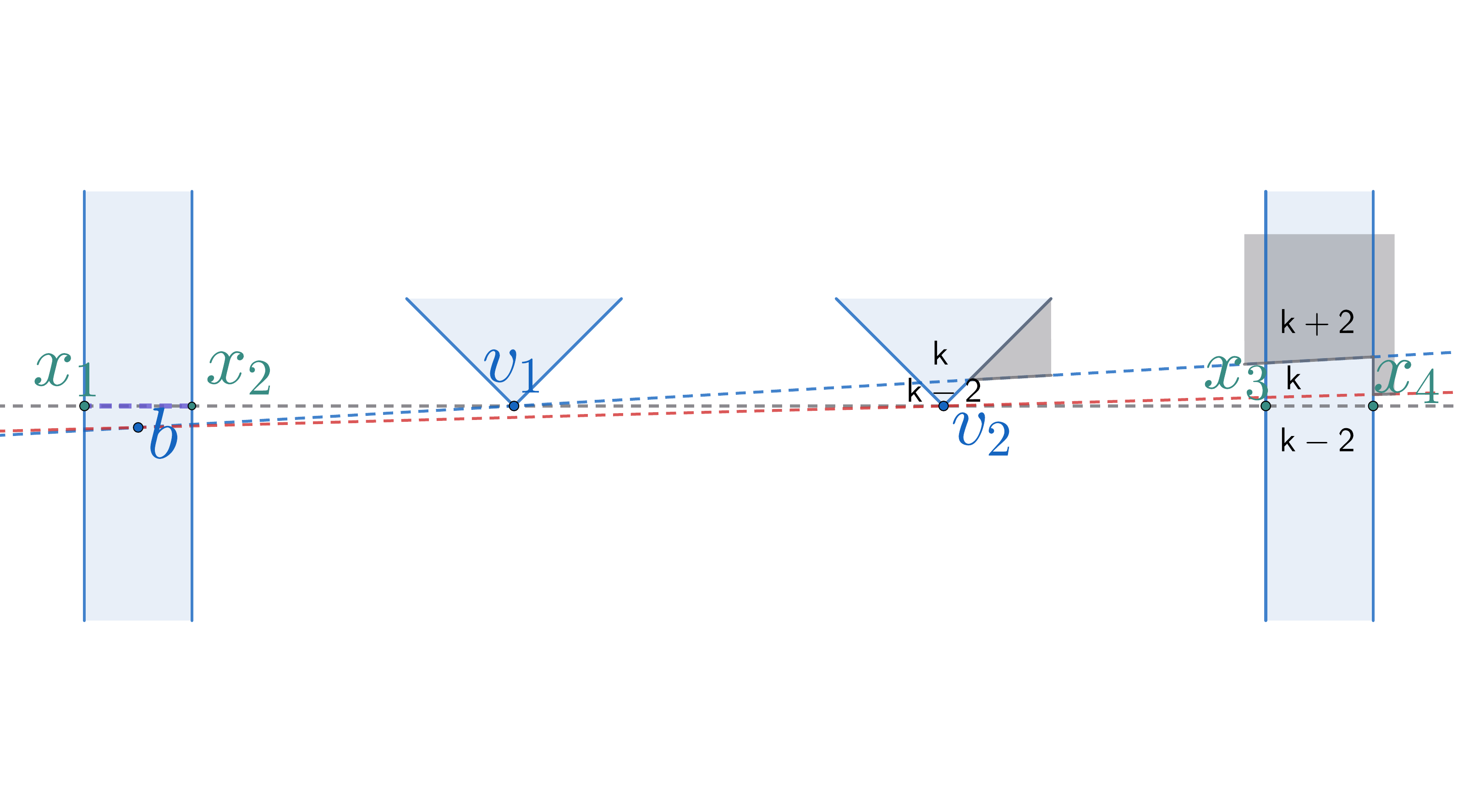}
\caption{Below $l_{g}$}\label{fig:CCS-genericB4}
\end{subfigure}

\caption{CCS; $Z = k - 3$, $W = k-2$}
\label{fig:CCS-generic-4}
\end{figure}

 \begin{figure}[H]
\centering
\begin{subfigure}[b]{.49\linewidth}
\includegraphics[width=0.85\linewidth]{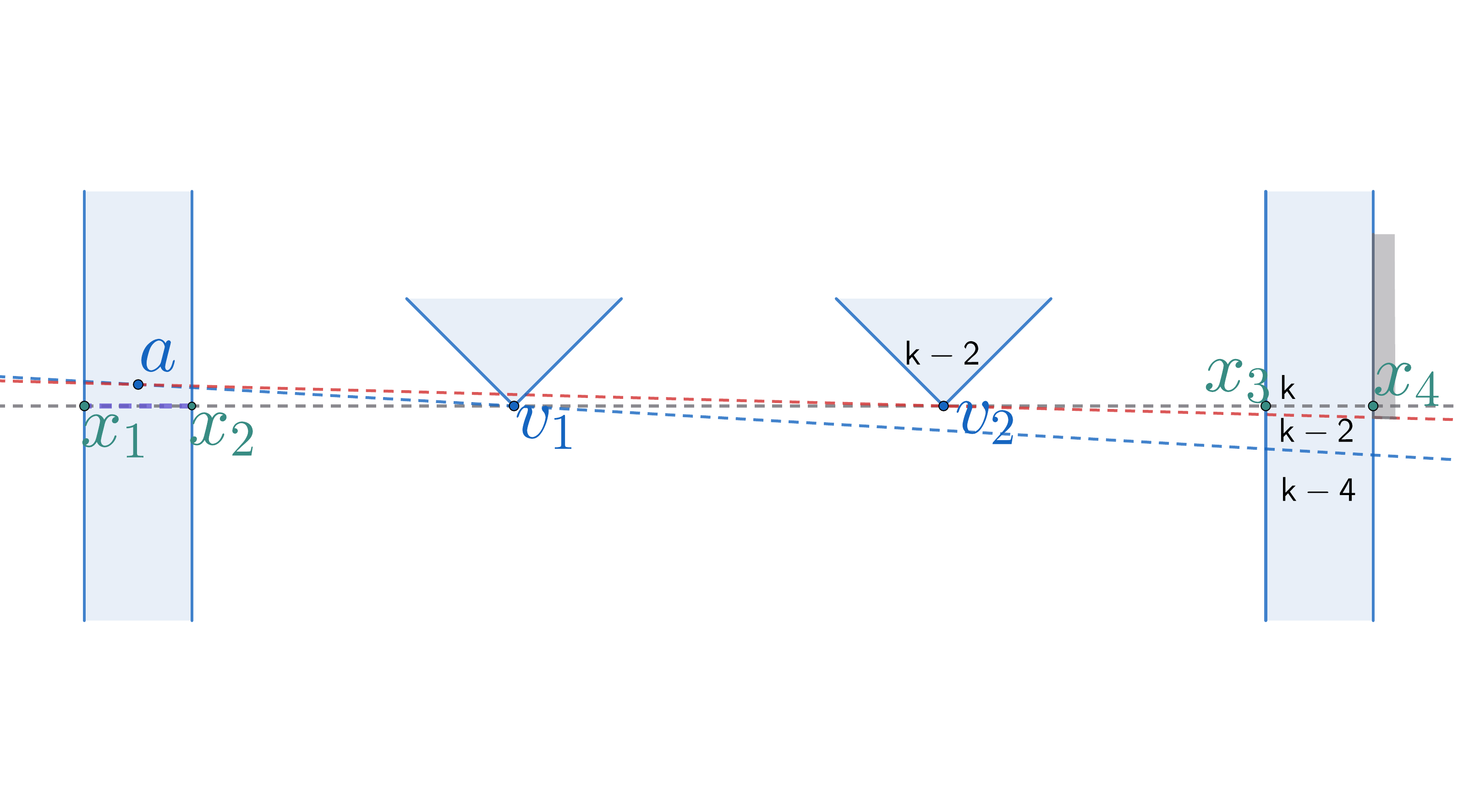}
\caption{Above $l_{g}$}\label{fig:CCS-genericA4}
\end{subfigure}
\begin{subfigure}[b]{.49\linewidth}

\includegraphics[width=0.85\linewidth]{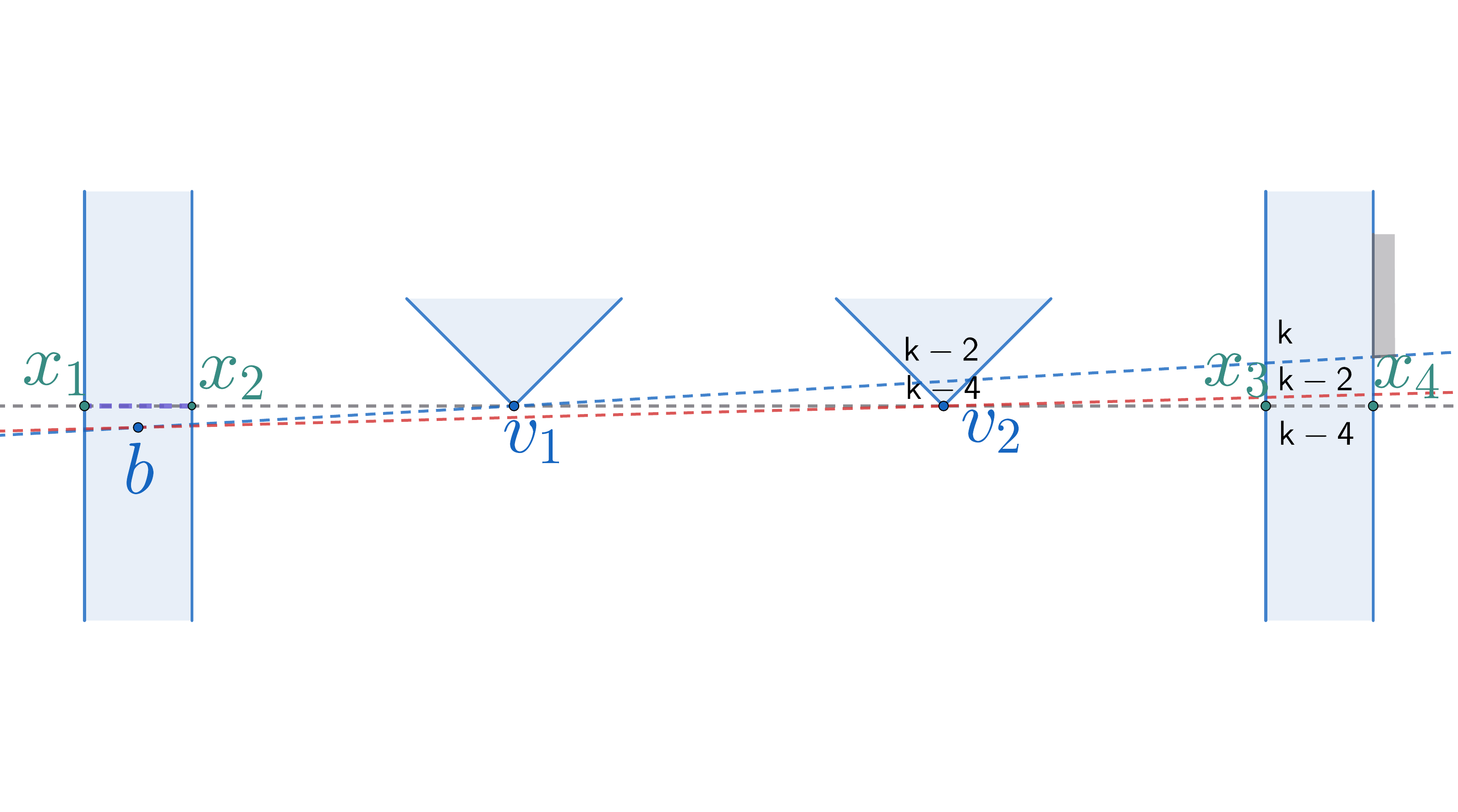}
\caption{Below $l_{g}$}\label{fig:CCS-genericB4}
\end{subfigure}

\caption{CCS; $Z = k - 5$, $W = k-4$}
\label{fig:CCS-generic-6}
\end{figure}

\end{proof}

The rest of the cases are in Appendix~\ref{appendix:mc-cases}.

\section{Algorithm: Constructing the Cell Decomposition}
The preprocessing phase constructs the cell decomposition arrangement and populates each cell with its invariant visibility sequence, as detailed in Algorithm~\ref{alg:cell_decomposition}.

\begin{algorithm}[H]
\caption{Cell Decomposition Preprocessing}
\label{alg:cell_decomposition}
\begin{algorithmic}
\Require{Simple polygon $P$ with $n$ vertices; threshold parameter $k \ge 0$.}
\Ensure{Planar arrangement $\mathcal{A}$ where each cell $C_i$ stores a counter-clockwise sequence $\mathcal{S}_i$.}

\Statex \textbf{Phase 1: Arrangement Geometry Construction}
\State Collect primary horizon lines $\mathcal{L}_1 \gets \bigcup_{v \in P} \{\text{window segments of } \text{Vis}_k(v)\  \text{and Vis}_{k-2}(v) \text{ in } P\}$.
\State Collect hinge transition lines $\mathcal{L}_2 \gets \{\text{lines from mutually critical vertex pairs (Sec. \ref{section:additionalLines})}\}$.
Remove duplicate lines
\State Construct the planar arrangement $\mathcal{A} \gets \text{Arrangement}(\mathcal{L}_1 \cup \mathcal{L}_2)$.
\State Split each non-convex cell into a convex one by extending the edges of reflex vertices.

\Statex \textbf{Phase 2: Combinatorial Sequence Propagation}
\For{each face (cell) $C_i \in \mathcal{A}$}
    \State Sample an arbitrary representative point $q \in C_i$.
    \State Label the boundary $\partial P$ and identify visibility windows as in~\cite{Martins2009}.
\State Initialize cell sequence $\mathcal{S}_i \gets \emptyset$.
\newline
\For{each individual vertex, edge, or connecting window within the $k$-labeled chains of $\partial P$ encountered during the counter-clockwise traversal defined in~\cite{Martins2009}}
    \If{a $k$-visible vertex $v$ is encountered}
        \State $\mathcal{S}_i \gets \mathcal{S}_i \cup \{v\}$.
    \ElsIf{a visibility window ray $r$ is encountered}
        \State Identify its critical generator $c$, proximal edge $e_\text{prox}$, and distal edge $e_\text{dist}$ from $\partial P$.
        \State Set $(e_1, e_2) \gets \text{encountered } e_\text{prox} \text{ before } e_\text{dist} \ ? \ (e_\text{prox}, e_\text{dist}) : (e_\text{dist}, e_\text{prox})$.
        \State $\mathcal{S}_i \gets \mathcal{S}_i \cup \{\langle c, e_1, e_2 \rangle\}$.
    \EndIf
\EndFor
    \State Assign permanent descriptor $\mathcal{S}_i$ to cell $C_i$.
\EndFor
\end{algorithmic}
\end{algorithm}

\section{Extension to Polygons with Holes}
Efficiently querying the $k$-visibility polygon of a polygon with holes requires first computing it from a single point. While the algorithm in~\cite{Martins2009} targets simple polygons, it extends to those with holes: Shoot rays from the query point through critical vertices. Label the resulting polygonal chains with their $k$-visibility values. Traverse the $k$-visible chains and connect them with windows.

\section{Query Processing and Visibility Reconstruction}

To dynamically reconstruct the $k$-visibility polygon $\text{Vis}_k(q)$ for a query point $q \in P$, the retrieval procedure executes in two sequential phases: an $O(\log n)$ point-location phase followed by a linear-time boundary construction phase.

First, an optimal planar point-location query is performed on the arrangement $\mathcal{A}$ to identify the unique cell $C_i$ containing $q$. Once the target cell is located, the algorithm retrieves its precalculated invariant visibility sequence $\mathcal{S}_i$. The geometric boundary of $\text{Vis}_k(q)$ is then reconstructed by performing a single linear-time traversal over $\mathcal{S}_i$ and mapping its entries to physical coordinates:

\begin{itemize}
    \item \textbf{Vertex Entries:} If an entry is a standalone $k$-visible vertex $v \in \mathcal{S}_i$, it is connected directly to the evolving boundary chain of the visibility polygon.
   \item \textbf{Window Generator Tuple:} If an entry is a tuple of the form $C = \langle c, e_1, e_2 \rangle$, it encodes both boundary edges intersected by the window's ray in the order they are encountered during the traversal. The physical endpoints of the visibility window are determined by computing ray-boundary intersections. Specifically, the algorithm calculates the intersection points between the directed line $l(q, c)$—passing through the query point $q$ and the critical vertex $c$—and the specified boundary edges $e_1$ and $e_2$. Because this reduces to a standard line-segment intersection operation, both endpoints of the window can be evaluated in $O(1)$ constant time.
\end{itemize}

By connecting these calculated vertices and boundary segments in the order prescribed by $\mathcal{S}_i$, the explicit geometry of $\text{Vis}_k(q)$ is successfully outputted in $O(|\mathcal{S}_i|) = O(n)$ time.

\section{Complexity of the Cell Decomposition}

The geometric boundaries of the cell decomposition are formed by the intersections of $O(n^2)$ partition lines generated by primary vertex visibility windows and secondary hinge transition lines. Based on the arrangement properties of $k$-star-shaped polygons established by Evans and Sember~\cite{k-starshaped}, the total number of distinct spatial cells in this planar arrangement is tightly bounded at $\Theta(n^4)$. There is no dependence on $k$.

While the spatial partition itself contains $\Theta(n^4)$ cells, each individual cell $C_i$ must explicitly store its invariant visibility sequence $\mathcal{S}_i$. In the worst case, a query point inside a cell can see a linear fraction of the polygon, meaning each visibility sequence can contain up to $\Theta(n)$ vertices and window generator tuples. Consequently, naively populating every cell with its corresponding sequence yields a total storage footprint of $\Theta(n^4) \times O(n) = O(n^5)$ space for the preprocessed data structure.

\section{Complexity of the Query Phase}

The runtime complexity of answering a dynamic $k$-visibility query is divided into two distinct steps: structural identification and geometric reconstruction.

\begin{enumerate}
    \item \textbf{Structural Sequence Retrieval:} Given a dynamic query point $a \in P$, locating the unique cell $C_i \in \mathcal{A}$ containing $a$ requires an optimal planar point-location query. Using standard hierarchical or trapezoidal decomposition methods, this location step is completed in $O(\log n)$ time. Because the combinatorial sequence $\mathcal{S}_i$ is directly linked to the cell, identifying the underlying visibility structure adds no extra overhead once the cell is found.
    \item \textbf{Geometric Polygon Reconstruction:} To reconstruct and output the explicit physical geometry of $\text{Vis}_k(q)$, the algorithm must iterate through the retrieved sequence $\mathcal{S}_i$. In the worst case, the $k$-visibility polygon can consist of a linear fraction of the entire polygon, meaning the sequence contains $|\mathcal{S}_i| = \Theta(n)$ vertices and window generator tuples. Since each entry is processed and mapped to physical coordinates via constant-time geometric operations ($O(1)$ per entry), traversing and constructing the final visibility polygon takes $O(n)$ time.
\end{enumerate}

Consequently, the total query time to fully reconstruct the $k$-visibility polygon is $O(\log n + m)$, where $m = |\mathcal{S}_i|$ is the size of the output sequence, bounded in the worst case by $O(n)$.

\section{Delta-Compressed Dual Arrangement Tree}
To eliminate redundant sequence storage across adjacent arrangement cells, we construct a spanning tree $\mathcal{T}$ of the dual cell graph via Breadth-First Search (BFS) starting from a central root cell $C_{\text{root}}$:
\begin{enumerate}
    \item \textbf{Root Storage:} The root cell $C_{\text{root}}$ explicitly stores its full, uncompressed visibility sequence $\mathcal{S}(C_{\text{root}})$.
    \item \textbf{Edge Deltas:} For each directed tree edge $(C_{\text{parent}} \to C_{\text{child}})$, we store only the combinatorial mutation $\Delta(C_{\text{parent}}, C_{\text{child}})$ required to transform $\mathcal{S}(C_{\text{parent}})$ into $\mathcal{S}(C_{\text{child}})$.
    \item \textbf{Compound Delta Encoding:} Across any single cell boundary, a structural transition may induce multiple simultaneous local changes (e.g., an vertex insertion coupled with window generator replacements). We encode each edge delta as an ordered sequence of $m$ atomic edit operations:
    \[
    \Delta(C_{\text{parent}}, C_{\text{child}}) = \left\langle \text{op}_1, \text{op}_2, \dots, \text{op}_m \right\rangle
    \]
    where $m = O(1)$ in general position, and each atomic operation takes one of three forms:
    \begin{itemize}
        \item $\text{REPLACE}(\text{idx}, \text{descriptor})$: Overwrites an existing window generator tuple or vertex descriptor at index $\text{idx}$ (e.g., swapping anchor identities or updating proximal/distal edges).
        \item $\text{INSERT}(\text{idx}, \text{descriptor})$: Inserts a new $k$-visible vertex or newly formed window generator into the sequence at index $\text{idx}$.
        \item $\text{DELETE}(\text{idx})$: Removes a vertex or window generator.
    \end{itemize}
\end{enumerate}

At query time, point-location identifies the target cell $C_{\text{target}}$. The uncompressed visibility sequence $\mathcal{S}(C_{\text{target}})$ is dynamically reconstructed in-memory by walking down the unique tree path from $C_{\text{root}}$ to $C_{\text{target}}$ in $\mathcal{T}$ and sequentially applying the sequence of $O(d)$ compound deltas, where $d$ is the depth of $C_{\text{target}}$ in $\mathcal{T}$.

While the naive representation requires $O(n^5)$ space ($O(n^4)$ cells, each storing a sequence of size $O(n)$), the delta-compressed tree reduces total space to $O(n^4)$, as the dual tree contains $O(n^4)$ edges each storing a compound delta of size $O(1)$. Reconstructing the visibility sequence at query time requires $O(d)$ edit applications, where $d$ is the depth of $C_{\text{target}}$ in $\mathcal{T}$.

\section{Conclusion and Future Work}In this paper, we introduced the first complete framework for answering $k$-visibility queries within polygons, successfully generalizing the classic 0-visibility cell decomposition model. By identifying, analyzing, and exhaustively classifying the exact geometric conditions that govern topological mutations—specifically primary vertex horizons and secondary hinge transition lines arising from pairs of mutually critical vertices—we established a tight spatial partition of $\Theta(n^4)$ cells. Within each cell, the combinatorial sequence of the $k$-visibility polygon remains strictly invariant. Furthermore, by exploiting the local combinatorial changes across adjacent cell boundaries, we introduced a $\delta$-compression scheme that reduces the overall storage requirement to $\mathcal{O}(n^4)$, while maintaining optimal $\mathcal{O}(\log n + m)$ query time to reconstruct explicit $k$-visibility polygons. Crucially, our approach naturally accommodates polygons with holes, expanding its applicability to complex environments.
Several intriguing avenues remain open for future investigation. Extending this cell decomposition paradigm to 3D polyhedral environments or investigating dynamic settings where polygon vertices are allowed to move present compelling challenges for future research in higher-order visibility.

\bibliographystyle{abbrv}

\bibliography{bibliography}

% - need exact definition of window (each window is between e and e' with critical vertex c)

\newpage
\appendix
\section*{Appendix} % Creates a big, unnumbered "Appendix" title

% \section{Proof for Type 1b lines}
% \label{appendix:type1b}

% \begin{figure}[H]
     
% \centering
% \begin{subfigure}[b]{.49\linewidth}
% \includegraphics[width=\linewidth]{Figures/query-2-vis-simple-v2_partitionline.pdf}
% \caption{$v_{11}v_{12}$ as $e_{prox}$}\label{fig:eprox1}
% \end{subfigure}
% \begin{subfigure}[b]{.49\linewidth}

% \includegraphics[width=\linewidth]{Figures/query-2-vis-simple-v3-partitionline.pdf}
% \caption{$v_{10}v_{11}$ as $e_{prox}$}\label{fig:eprox2}
% \end{subfigure}

% \caption{Addition of a vertex ($v_{11}$) to the visibility sequence and change of proximal edge based on a partition line made from the window of the $k-2$-visibility polygon of a vertex ($v_{11}$)}
% \label{fig:eprox}
% \end{figure}

% \begin{proof}
%     The windows of the $k-2$-visibility polygon of each vertex are needed because to one side of the window (dashed purple line in Figure~\ref{fig:eprox} for $v_{11}$), the vertex is within the $k$-visibility polygon of the query point (Figure~\ref{fig:eprox2}), while to the other side, it becomes $k$-visible (Figure~\ref{fig:eprox1}), so it becomes part of the combinatorial sequence defining the visibility polygon. Furthermore, it may also correspond to change in the $e_{prox}$ of a window generator tuple (e.g., $v_{11}v_{12}$ in Figure~\ref{fig:eprox1} and $v_{10}v_{11}$ in Figure~\ref{fig:eprox2}).
% \end{proof}\qed

\section{Redundancy of Windows of $\text{Vis}_{x}{v}$ for $x \le k-4$ and $x \ge k+2$ (when $v$ is not critical to the vertex responsible for the window)}
\label{appendix:unneededwindows}

\begin{proof}
    Since we assume general position and that no three vertices are on the same line, then there can be three scenarios: 1) two collinear mutually critical vertices, 2) one vertex that is critical to another but not vice versa, and 3) two vertices that are not critical to each other. 

\begin{itemize}
    \item 
    Case 1 is handled by the hinge transition lines in Section~\ref{section:additionalLines}.
    \item 
    Case 3 Draw a line $\ell_{g}$ through the two vertices $v_{1}$ and $v_{2}$. Crossing $\ell_{g}$ does not cause a change in whether $v_{2}$ is visible/invisible.
    \item Case 2. Call the non-critical vertex $v_2$ and the vertex critical to it $v_1$.  Assume there is a horizon line made by the window of $\text{Vis}_x{v}$ where $x \geq k + 2$. When the query point is to one side of the line, $v$ is in the shadow of the query point, and when the query point is to the other side of the line. $v$ is still in the shadow of the query point. Therefore, $v$ does not leave or enter the visibility sequence that defines the visibility polygon as the query point crosses this horizon line, so this horizon line is redundant. To see an example of this, consider Figure~\ref{fig:visk+2} for the case of $2$-visibility. If there is a horizon line $H$ (dashed purple line) of the $\text{Vis}_{4}(v_{2})$ that acts as a partition line, crossing it does not make $v_{2}$ part of the visibility sequence, as it remains entirely in shadow (both at $a$ and $b$).  \\

     \begin{figure}[H]
\centering
\begin{subfigure}[b]{.49\linewidth}
\includegraphics[width=0.9\linewidth]{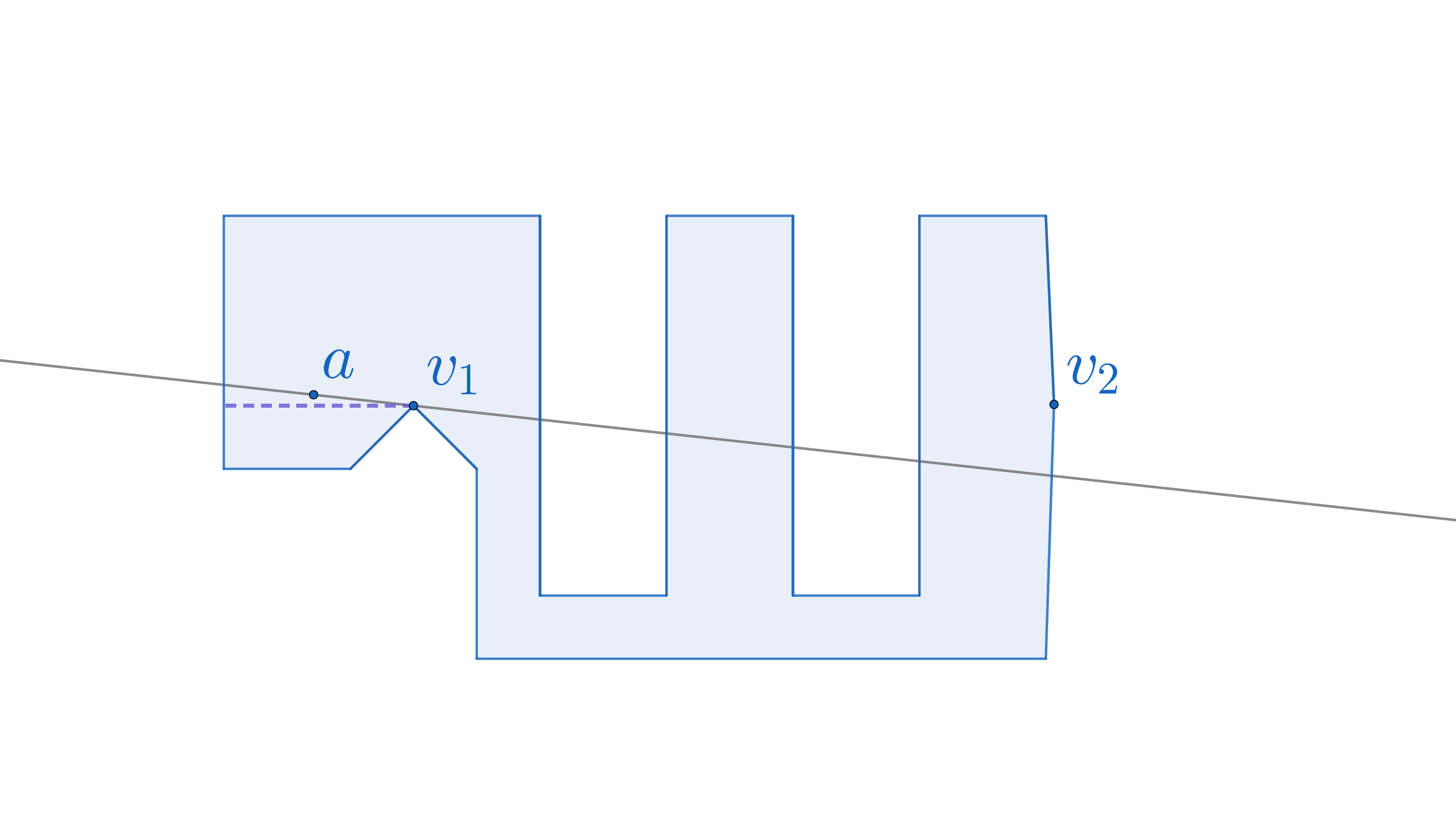}
\caption{Above $H$}\label{fig:CCS-genericA4}
\end{subfigure}
\begin{subfigure}[b]{.49\linewidth}

\includegraphics[width=0.9\linewidth]{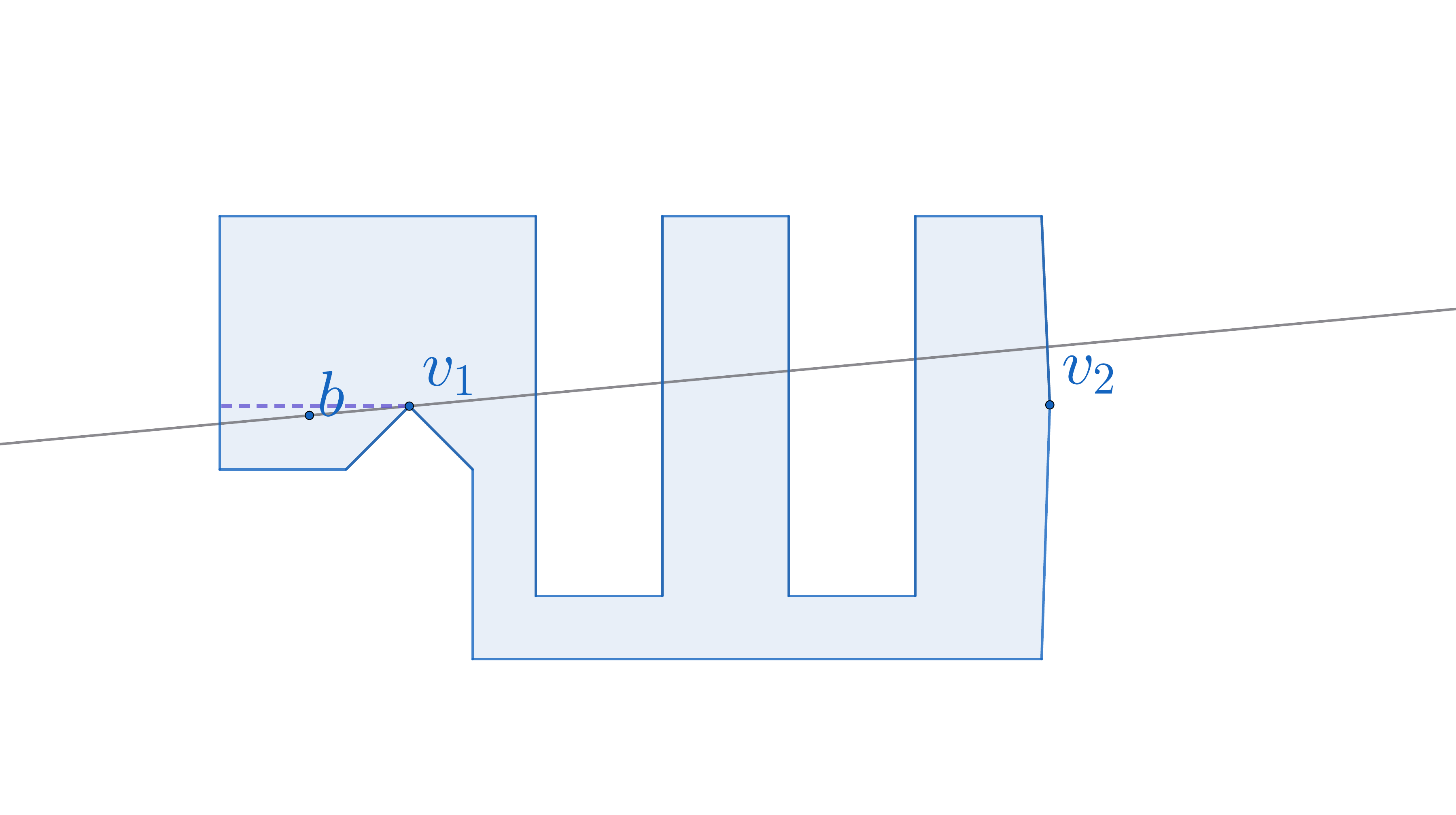}
\caption{Below $H$}\label{fig:CCS-genericB4}
\end{subfigure}

\caption{$\text{Vis}_{4}(v_{2})$}
\label{fig:visk+2}
\end{figure}
Assume there is a horizon line made by the window of $\text{Vis}_x{v}$ where $x \leq k - 4$. When the query point is to one side of the line, $v$ is completely within the visibility polygon of the query point. When the query point is to the other side, $v$ is still completely within the visibility polygon of the query point. Therefore, $v$ does not leave or enter the visibility sequence that defines the visibility polygon, so this horizon line is redundant. The same example in Figure~\ref{fig:visk+2} can be used with $8$-visibility, where $H$ (dashed purple line) is a horizon line of $\text{Vis}_{4}(v_{2})$. $v_{2}$ remains within the visibility polygon at both $a$ and $b$.

\end{itemize}
\end{proof}
\section{Cases for Two Mutually Critical Vertices}
\label{appendix:mc-cases}
\subsection{CCS}
\label{ccs}
See Section~\ref{section:CCS}

% CCS(2) at $v_{2}$ (nothing; b/w $l_{b}$ and $l_{r}$) and $x_{3}x_{4}$ (below $l_{b}$; below $l_{r}$)\\
% CCS(4) at $x_{3}x_{4}$ (below $l_{r}$; below $l_{b}$)
\subsection{CCO}
\label{section:CCO}
\begin{lemma}
\label{lemma:CCO}
A partition line is needed in the following cases:
\renewcommand{\labelitemi}{$\bullet$}
\begin{itemize}
\setlength{\itemsep}{0em}
    \item  $Z = k - 1$ (Figure~\ref{fig:CCO-generic-2})
    \item $W = k$ (Figure~\ref{fig:CCO-generic-2})
    \item $Z = k - 3$ (Figure~\ref{fig:CCO-generic-4})
    \item $W = k - 2$
    (Figure~\ref{fig:CCO-generic-4})
    \item $W = k - 4$
    (Figure~\ref{fig:CCO-generic-6})
    \end{itemize}
\end{lemma}
\begin {proof} Consider Figures~\ref{fig:CCO-genericA2} - ~\ref{fig:CCO-genericB6}. 

Note: for $Z \geq k + 1$, all of $v_{2}$ and its immediate neighbourhood would be in shadow. For $W \geq k + 2$, $x_{3}x_{4}$ and its immediate neighbourhood is entirely in shadow. 

     \begin{figure}[H]
     
\centering
\begin{subfigure}[b]{.49\linewidth}
\includegraphics[width=\linewidth]{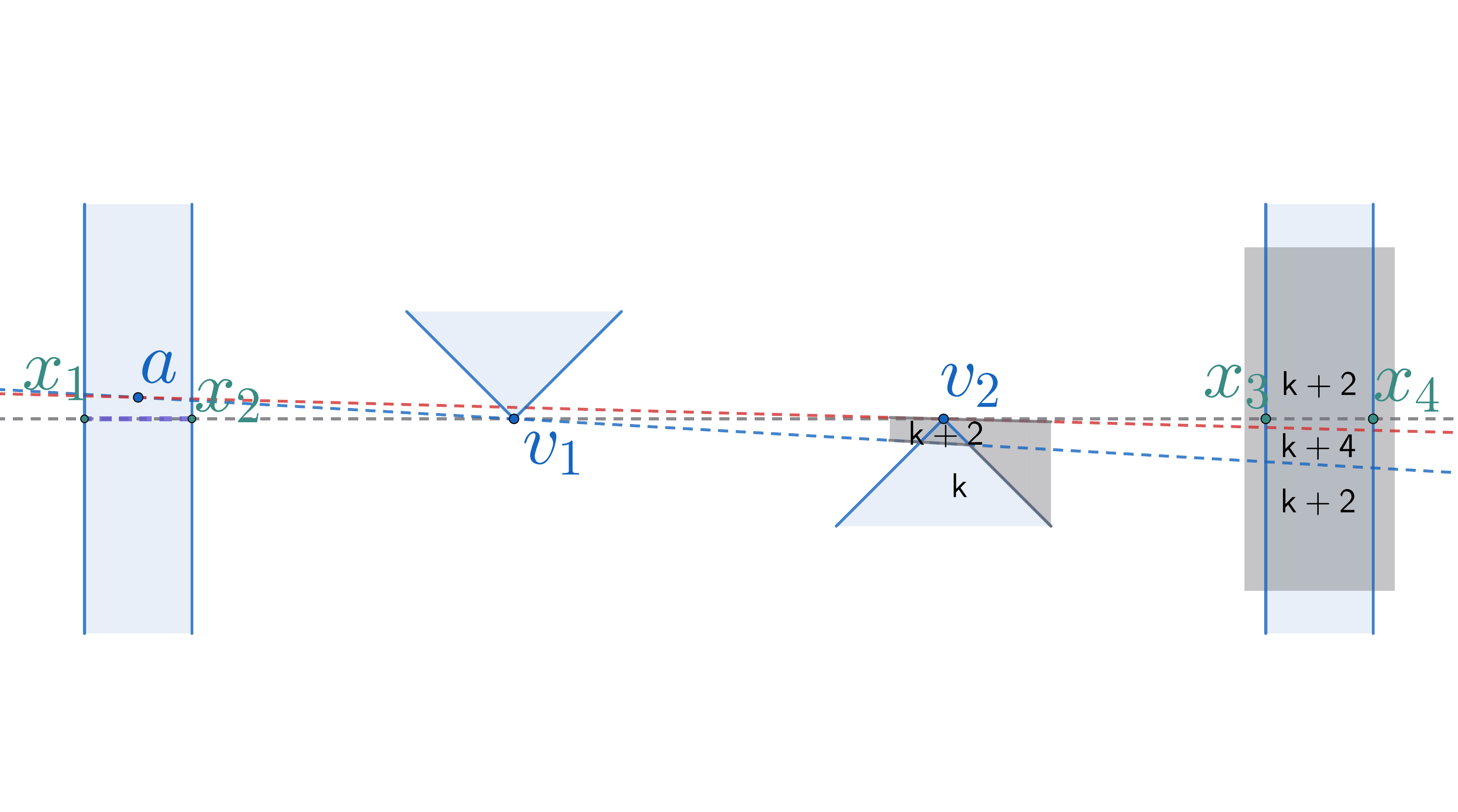}
\caption{Above $l_{g}$}\label{fig:CCO-genericA2}
\end{subfigure}
\begin{subfigure}[b]{.49\linewidth}

\includegraphics[width=\linewidth]{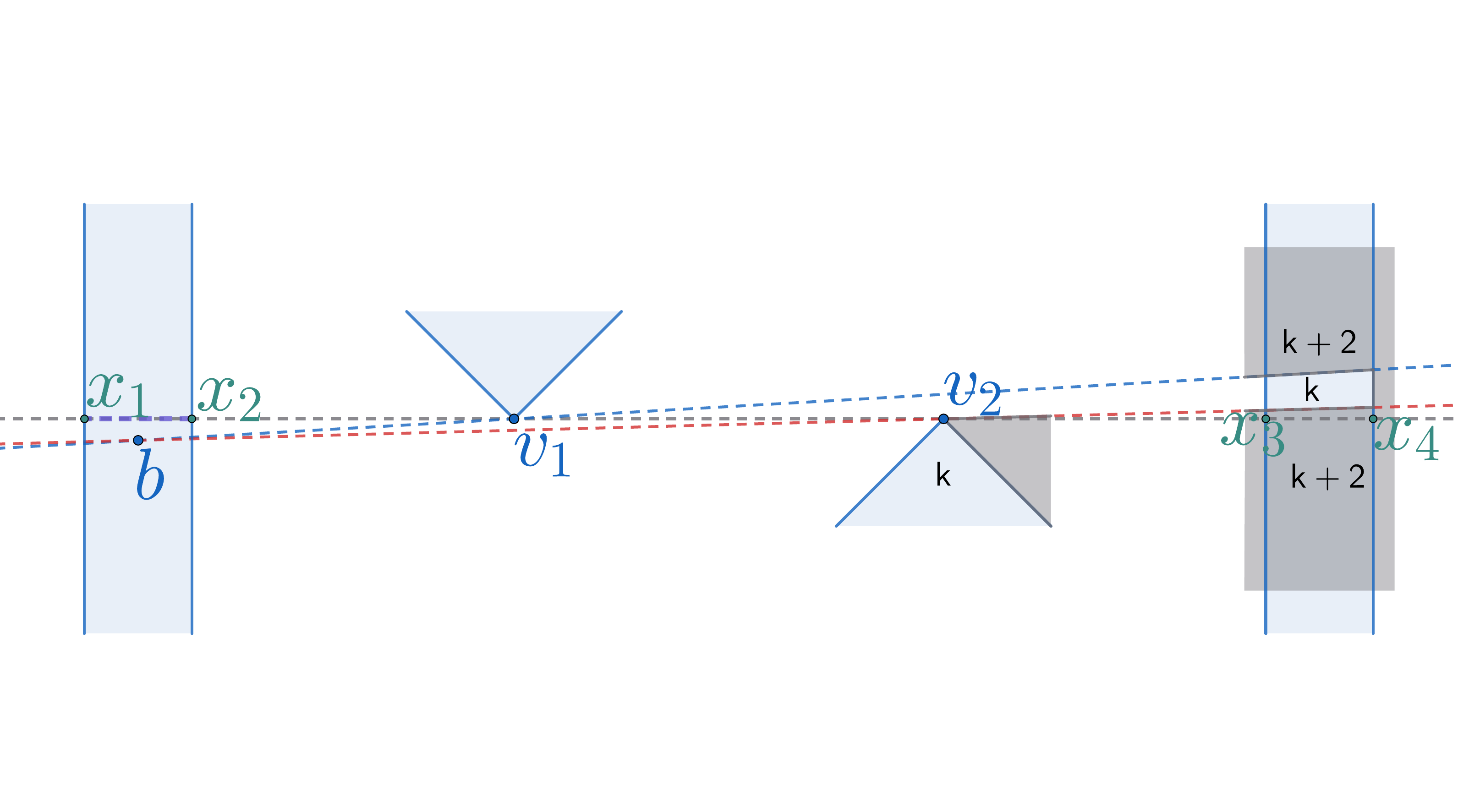}
\caption{Below $l_{g}$}\label{fig:CCO-genericB2}
\end{subfigure}

\caption{CCO (2); $Z = k - 1$, $W = k$}
\label{fig:CCO-generic-2}
\end{figure}
 \begin{figure}[H]
\centering
\begin{subfigure}[b]{.49\linewidth}
\includegraphics[width=\linewidth]{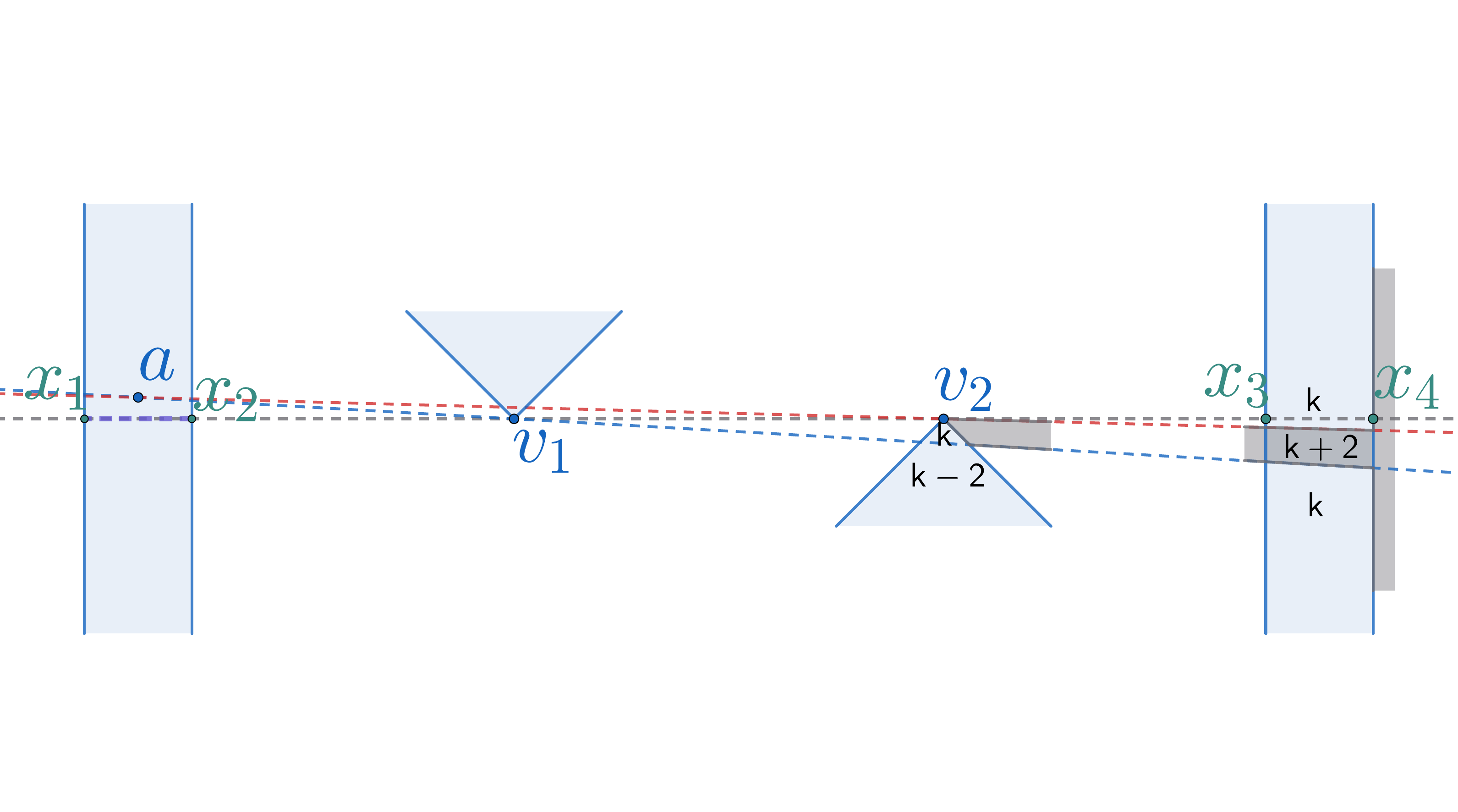}
\caption{Above $l_{g}$}\label{fig:CCO-genericA4}
\end{subfigure}
\begin{subfigure}[b]{.49\linewidth}

\includegraphics[width=\linewidth]{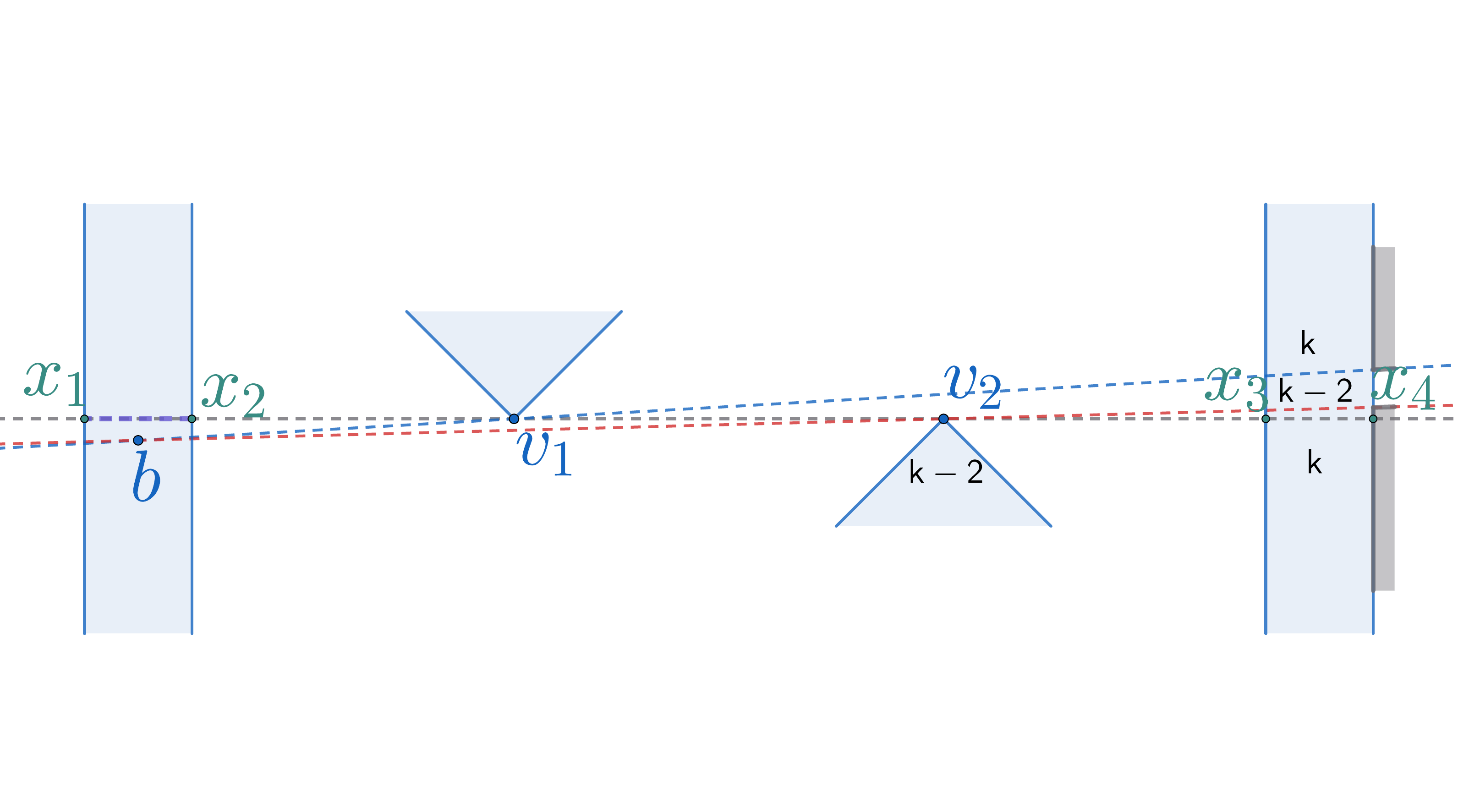}
\caption{Below $l_{g}$}\label{fig:CCO-genericB4}
\end{subfigure}

\caption{CCO (4); $Z = k - 3$, $W = k - 2$}
\label{fig:CCO-generic-4}
\end{figure}

 \begin{figure}[H]
\centering
\begin{subfigure}[b]{.49\linewidth}
\includegraphics[width=\linewidth]{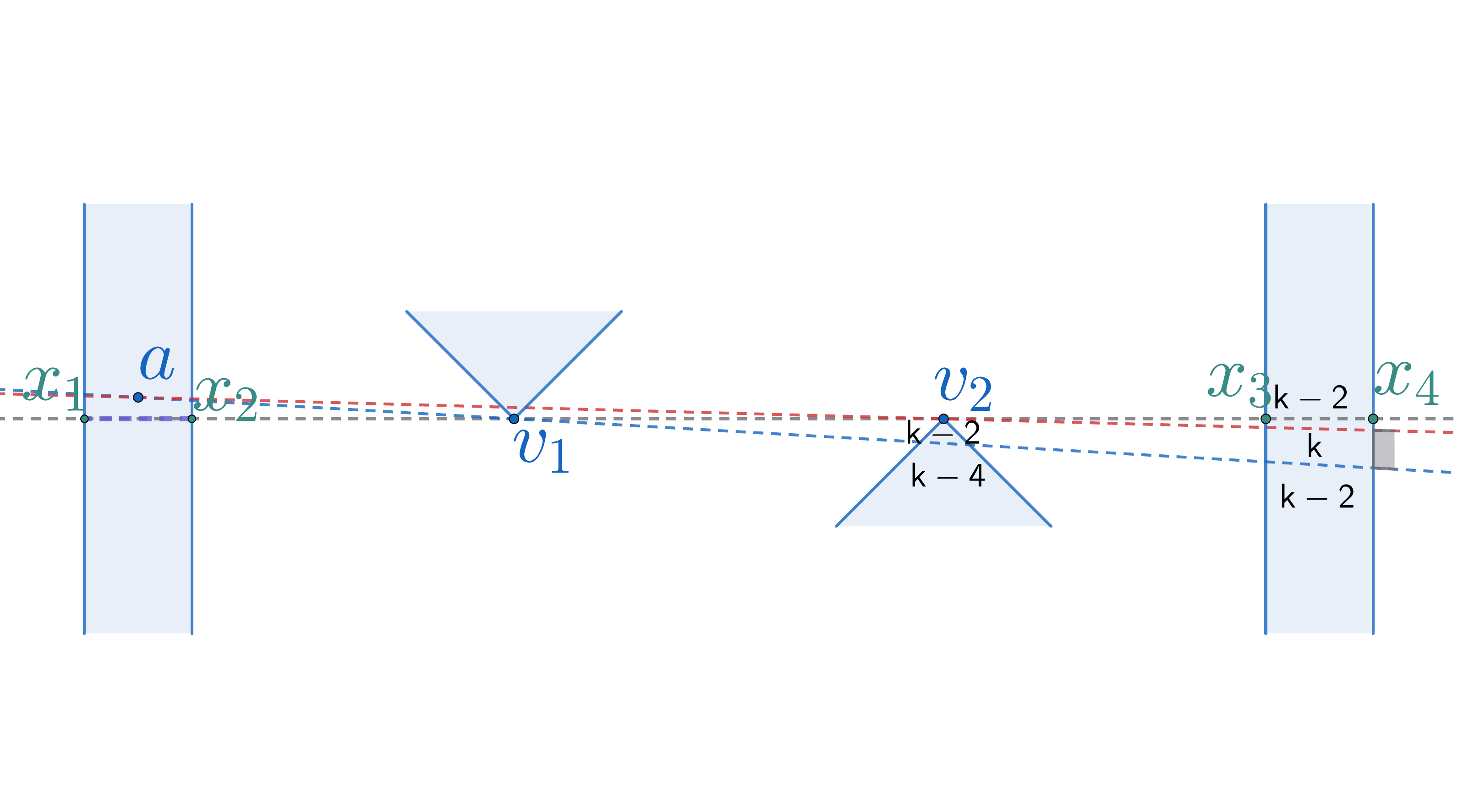}
\caption{Above $l_{g}$}\label{fig:CCO-genericA6}
\end{subfigure}
\begin{subfigure}[b]{.49\linewidth}

\includegraphics[width=\linewidth]{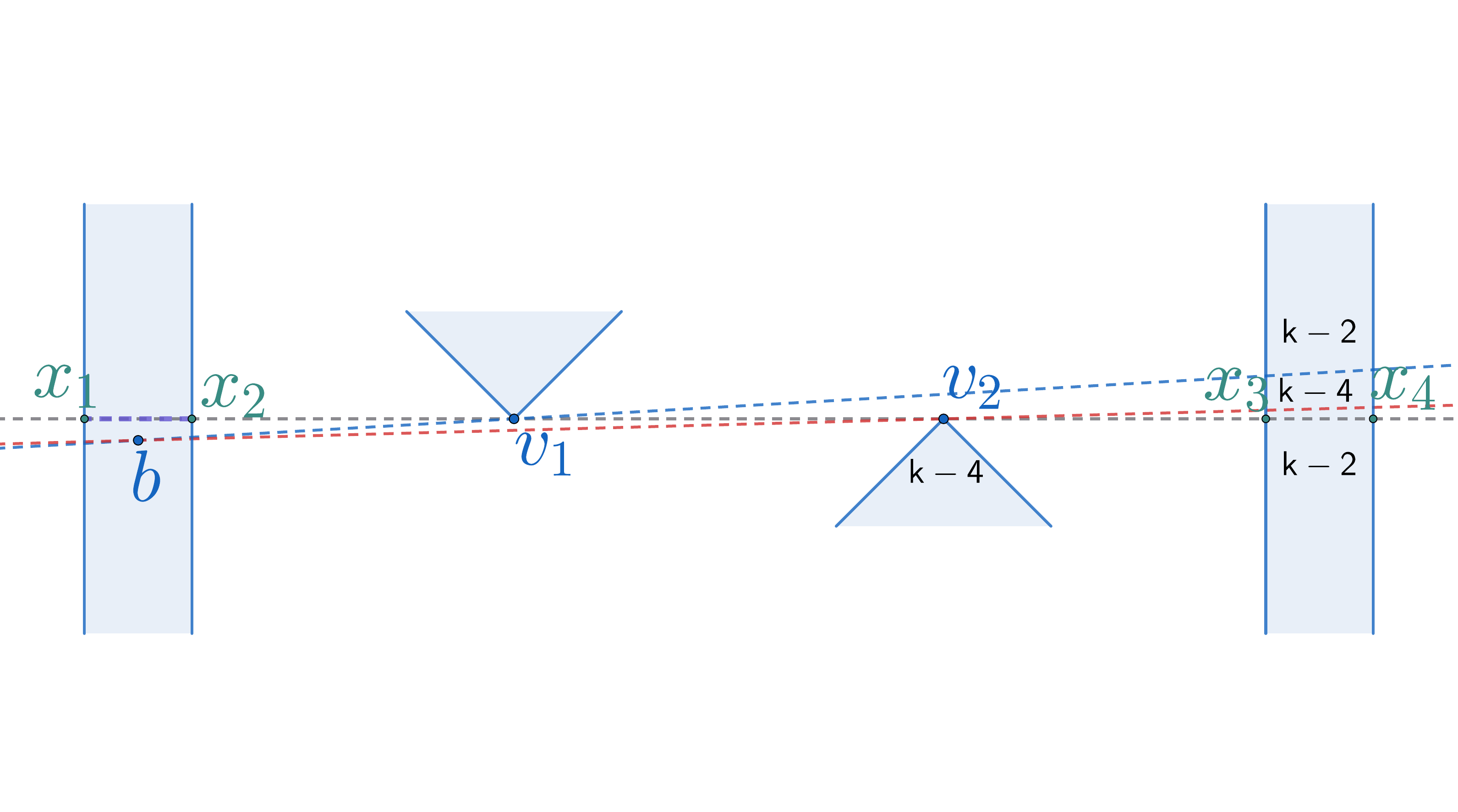}
\caption{Below $l_{g}$}\label{fig:CCO-genericB6}
\end{subfigure}

\caption{CCO; $Z = k - 5$, $W = k - 4$}
\label{fig:CCO-generic-6}
\end{figure}
for $Z \leq k - 7$,  $v_{2}$ and its surrounding is entirely visible, for $W \leq k - 6$,  $x_{3}x_{4}$ are entirely visible.

\end{proof}
    
% CCO(2) at $v_{2}$ (below $l_{b}$; nothing?) and $x_{3}x_{4}$ (nothing; between $l_{b}$ and $l_{r}$)\\
% CCO(4) at $x_{3}x_{4}$ (below $l_{b}$, above $l_{g}$; c.vis)
\section{CRS}
\begin{lemma}
\label{lemma:CRS}
A partition line is needed in the following cases:
\renewcommand{\labelitemi}{$\bullet$}
\begin{itemize}

    \item $Z = k$ (Figure~\ref{fig:CRO-generic2})
    \item $Z = k - 2$ (Figure~\ref{fig:CRO-generic4}
    \item $W = k$ (Figure~\ref{fig:CRO-generic4})
    \item $Z = k -4$ (Figure~\ref{fig:CRO-generic6})
    \item $W = k - 2$ (Figure~\ref{fig:CRO-generic6})
    \item $W = k - 4$ (Figure~\ref{fig:CRO-generic8})
\end{itemize}
\end{lemma}

\begin{proof}
 See Figures~\ref{fig:CRO-generic2} - ~\ref{fig:CRO-generic8}.

for $Z \geq k + 2$, all of $v_{2}$ and its immediate neighbourhood are entirely in shadow. For $W \geq k + 4$, $x_{3}x_{4}$ and its immediate neighbourhood is entirely in shadow. 
\begin{figure}[H]
\centering
\begin{subfigure}[b]{.49\linewidth}
\includegraphics[width=\linewidth]{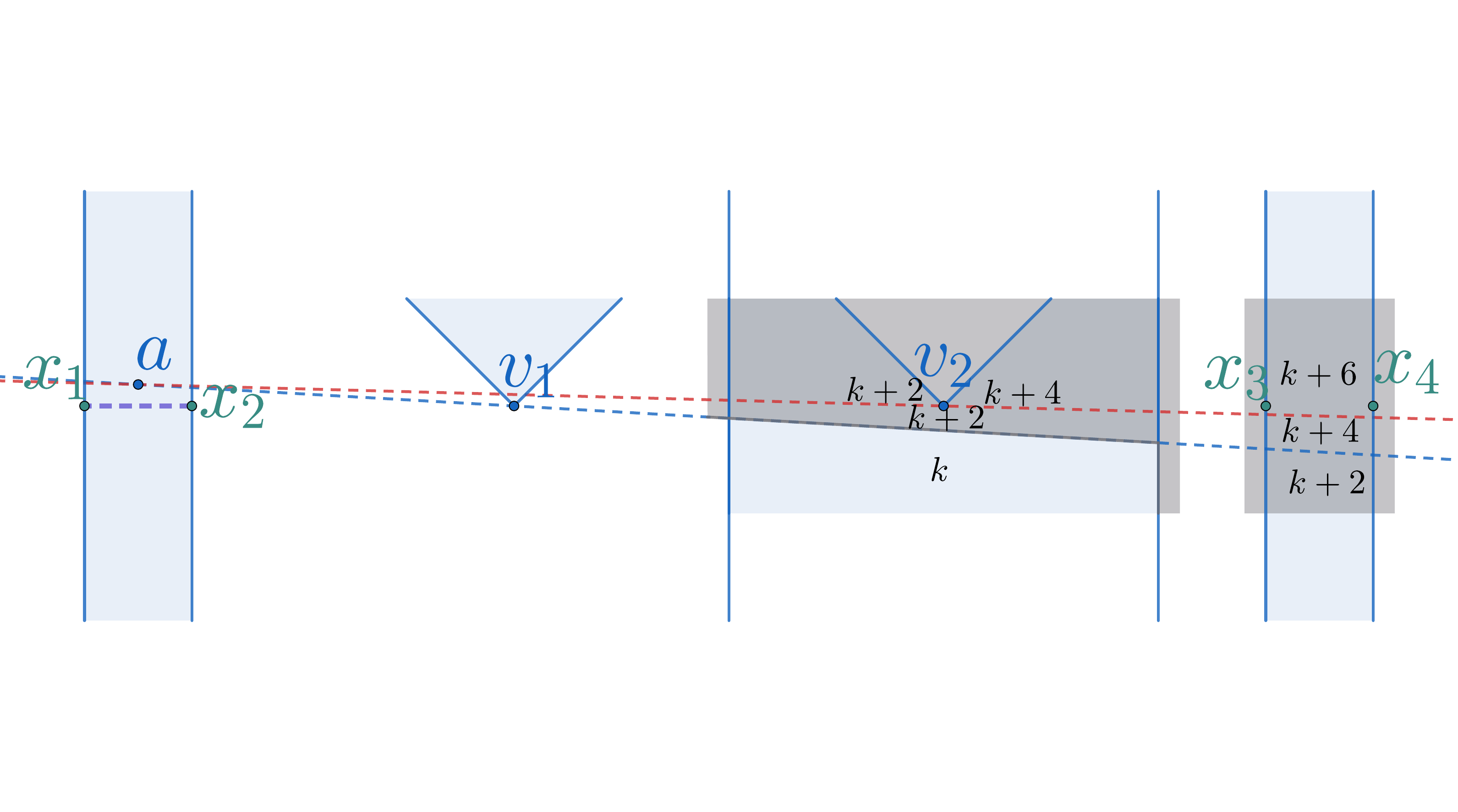}
\caption{Above $l_{g}$}\label{fig:CRO-genericA2}
\end{subfigure}
\begin{subfigure}[b]{.49\linewidth}
\includegraphics[width=\linewidth]{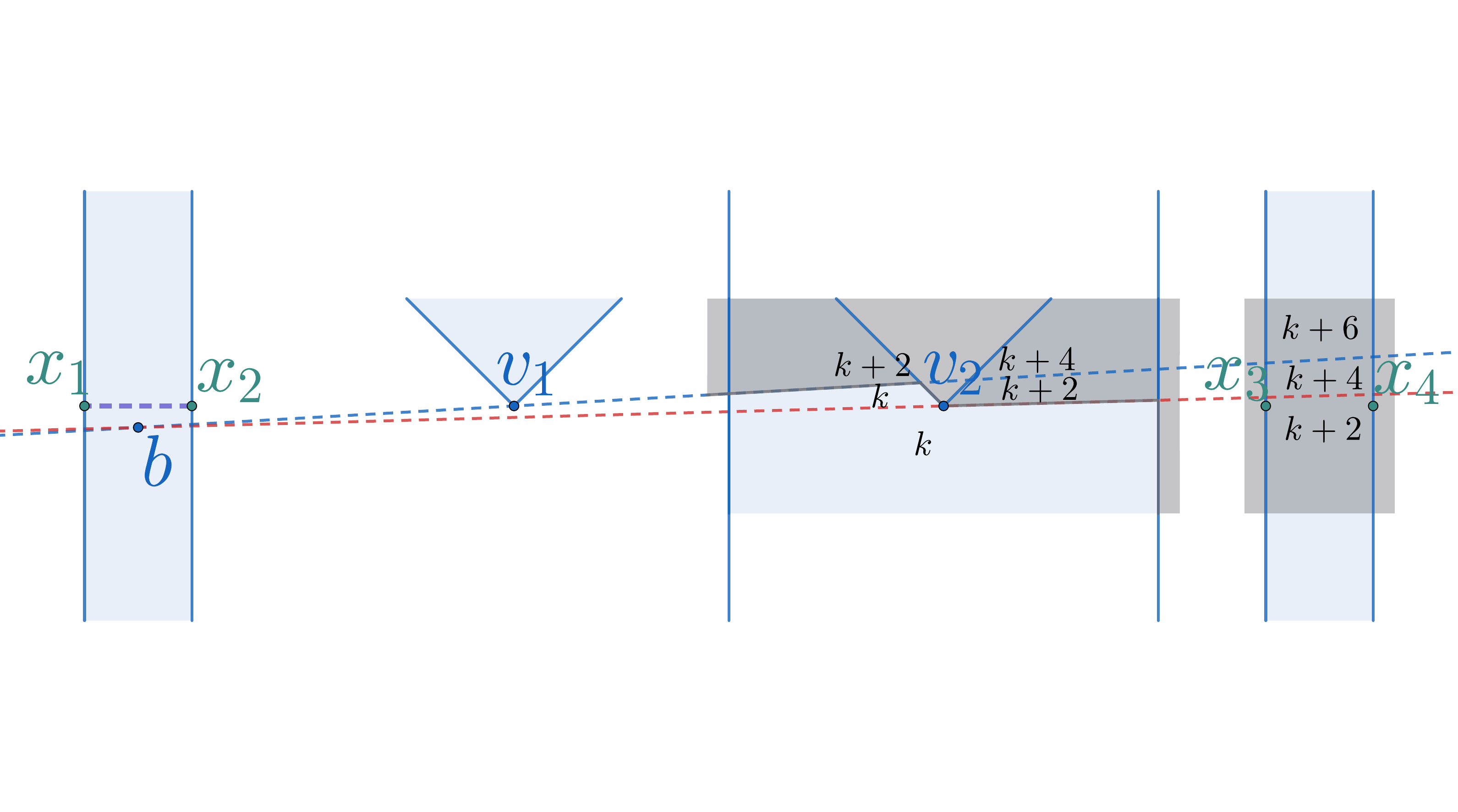}
\caption{Below $l_{g}$}\label{fig:CRO-genericB2}
\end{subfigure}

\caption{CRS; $Z=k$; $W=k+2$ }
\label{fig:CRO-generic2}
\end{figure}

 \begin{figure}[H]
\centering
\begin{subfigure}[b]{.49\linewidth}
\includegraphics[width=\linewidth]{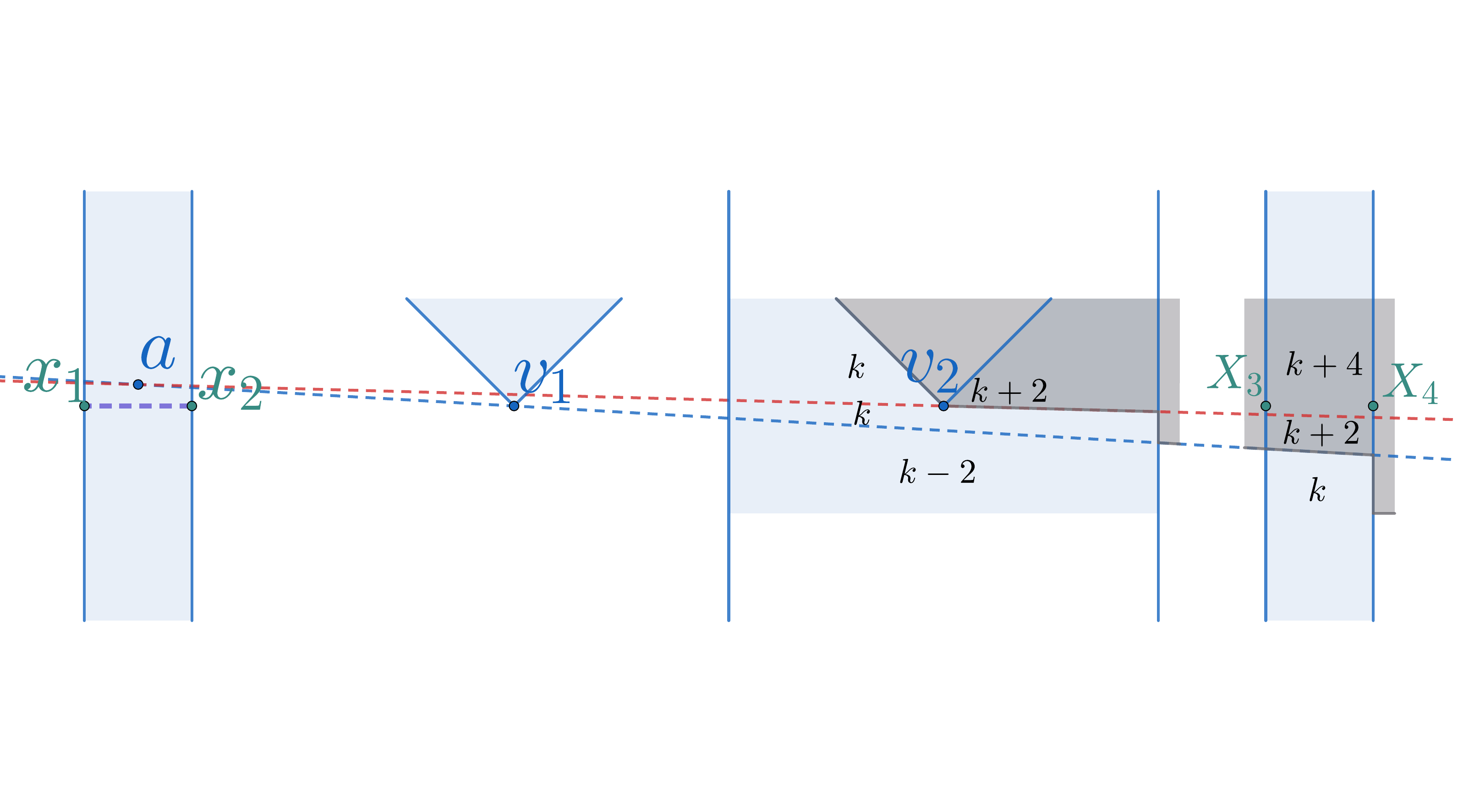}
\caption{Above $l_{g}$}\label{fig:CRO-genericA4}
\end{subfigure}
\begin{subfigure}[b]{.49\linewidth}
\includegraphics[width=\linewidth]{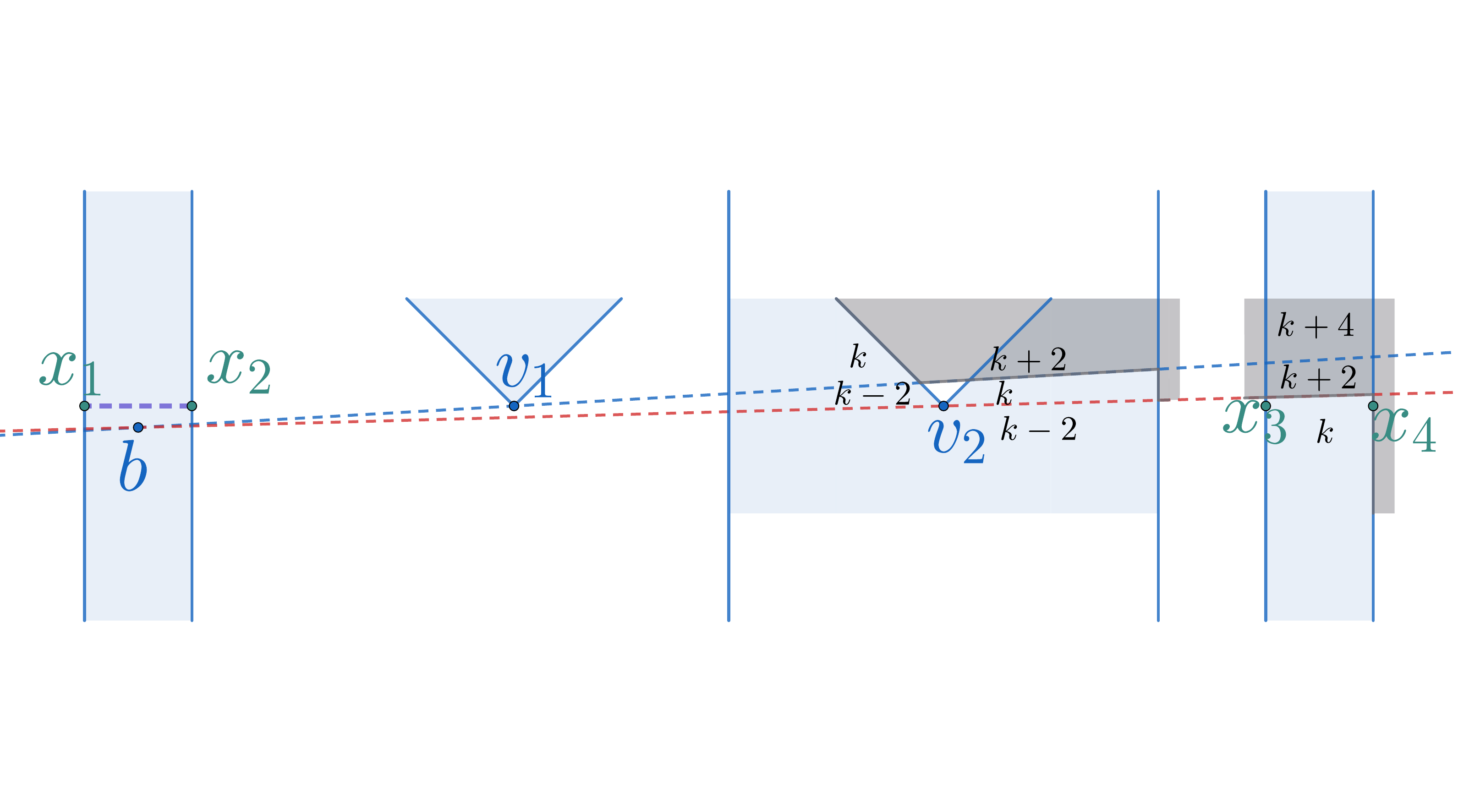}
\caption{Below $l_{g}$}\label{fig:CRO-genericB4}
\end{subfigure}

\caption{CRS; $Z = k - 2$; $W = k$}
\label{fig:CRO-generic4}
\end{figure}

 \begin{figure}[H]
\centering
\begin{subfigure}[b]{.49\linewidth}
\includegraphics[width=\linewidth]{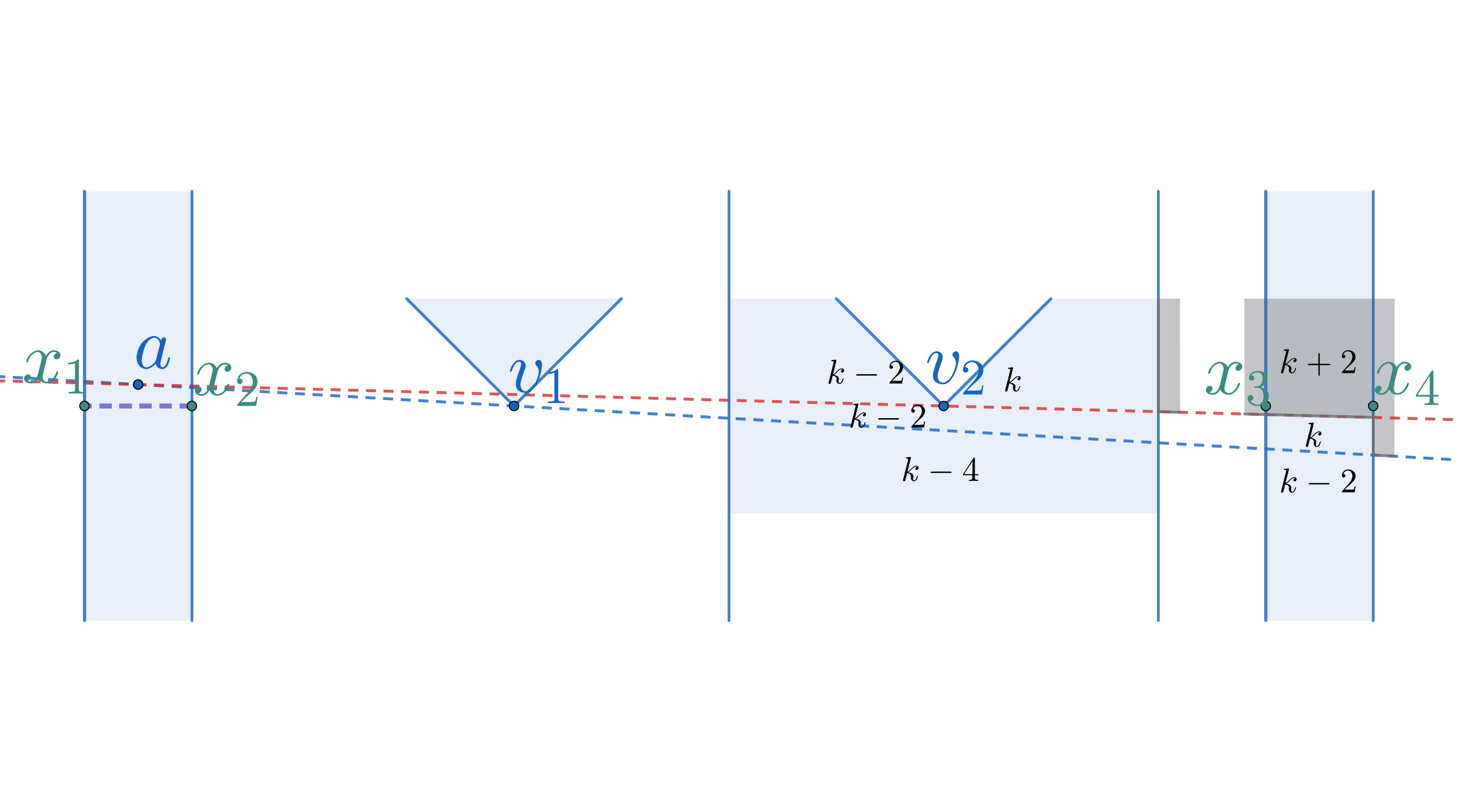}
\caption{Above $l_{g}$ }\label{fig:CRO-genericA6}
\end{subfigure}
\begin{subfigure}[b]{.49\linewidth}
\includegraphics[width=\linewidth]{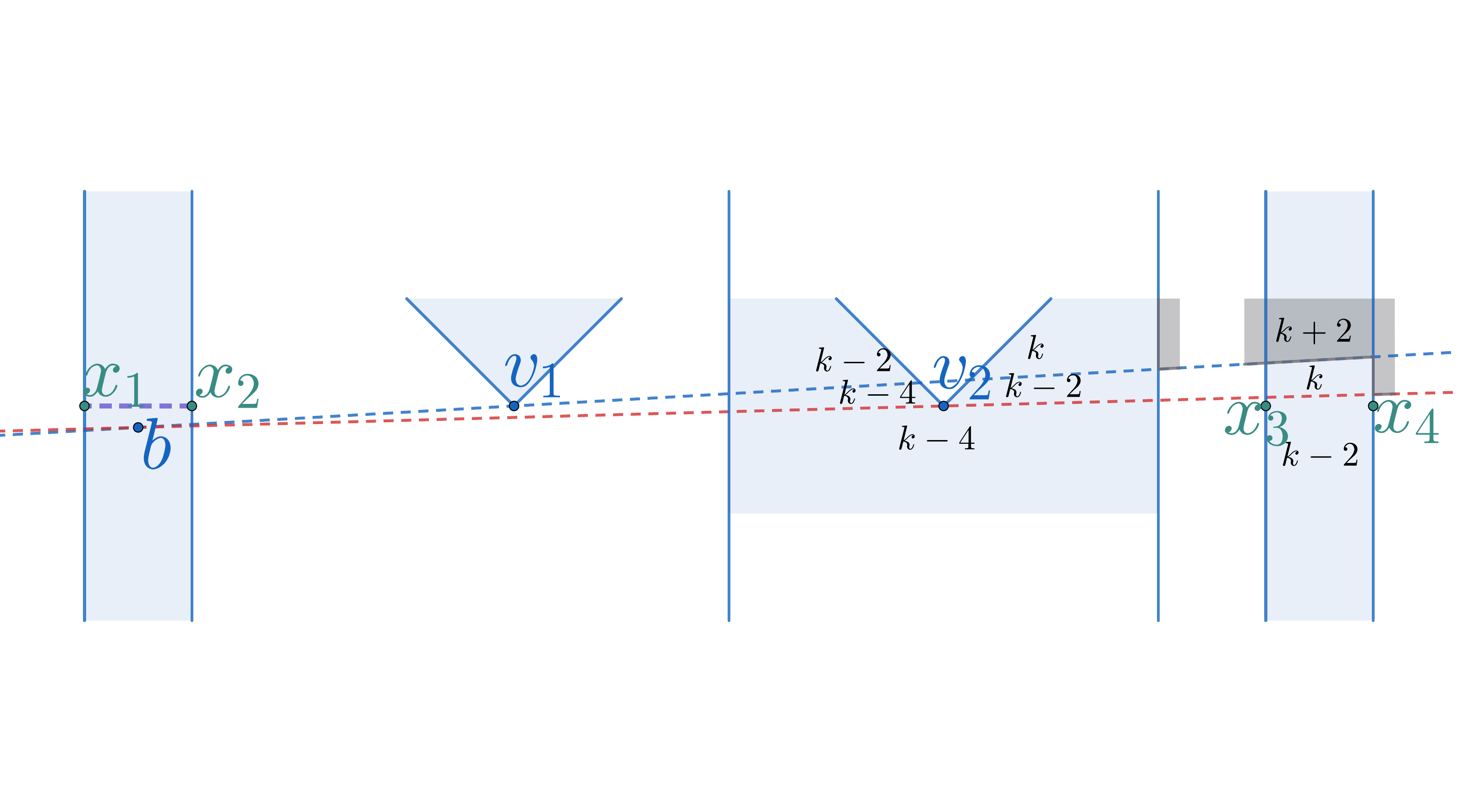}
\caption{Below $l_{g}$ }\label{fig:CRO-genericB6}
\end{subfigure}

\caption{CRS; $Z = k - 4$, $W = k - 2$}
\label{fig:CRO-generic6}
\end{figure}

 \begin{figure}[H]
\centering
\begin{subfigure}[b]{.49\linewidth}
\includegraphics[width=\linewidth]{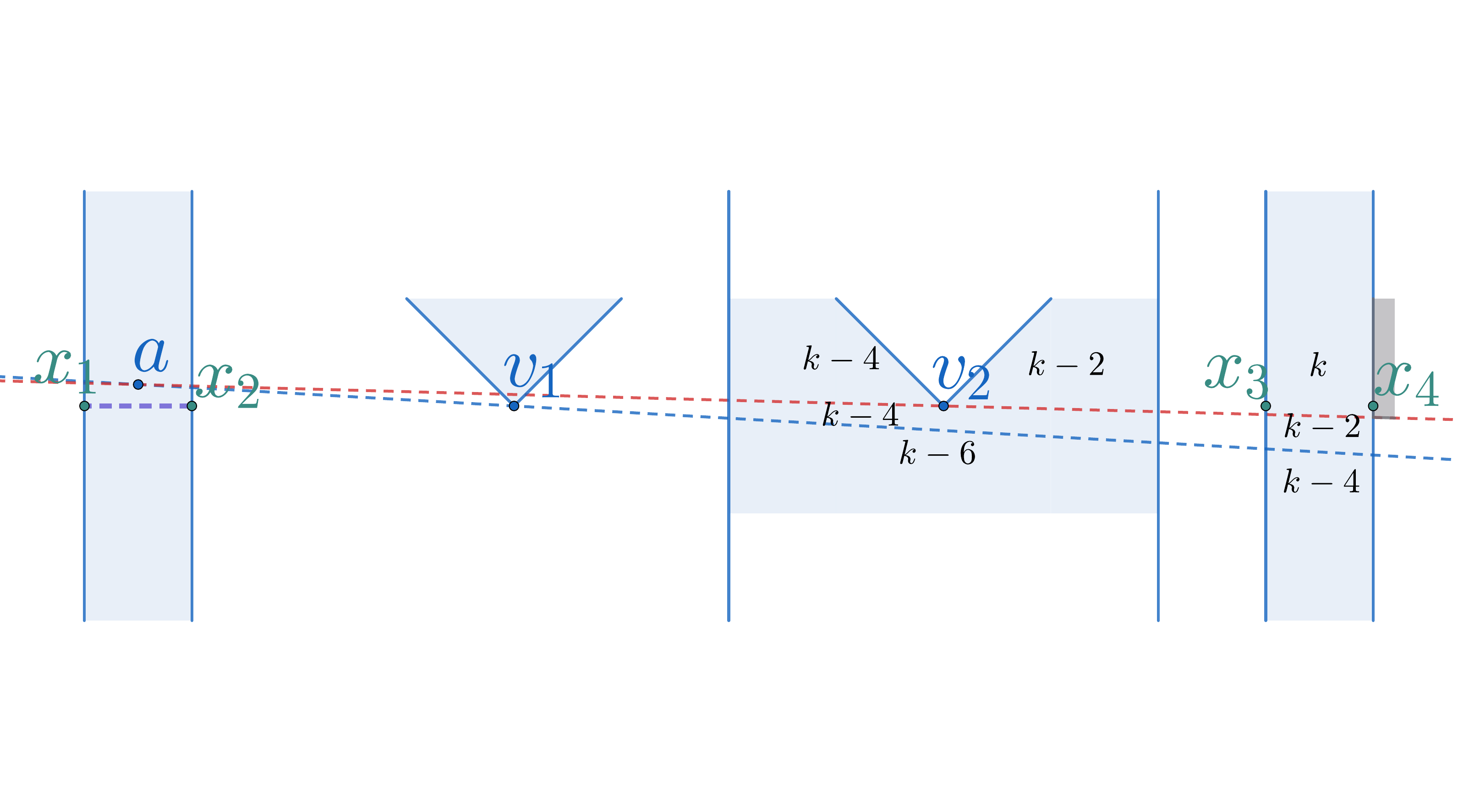}
\caption{Above $l_{g}$ }\label{fig:CRO-genericA8}
\end{subfigure}
\begin{subfigure}[b]{.49\linewidth}
\includegraphics[width=\linewidth]{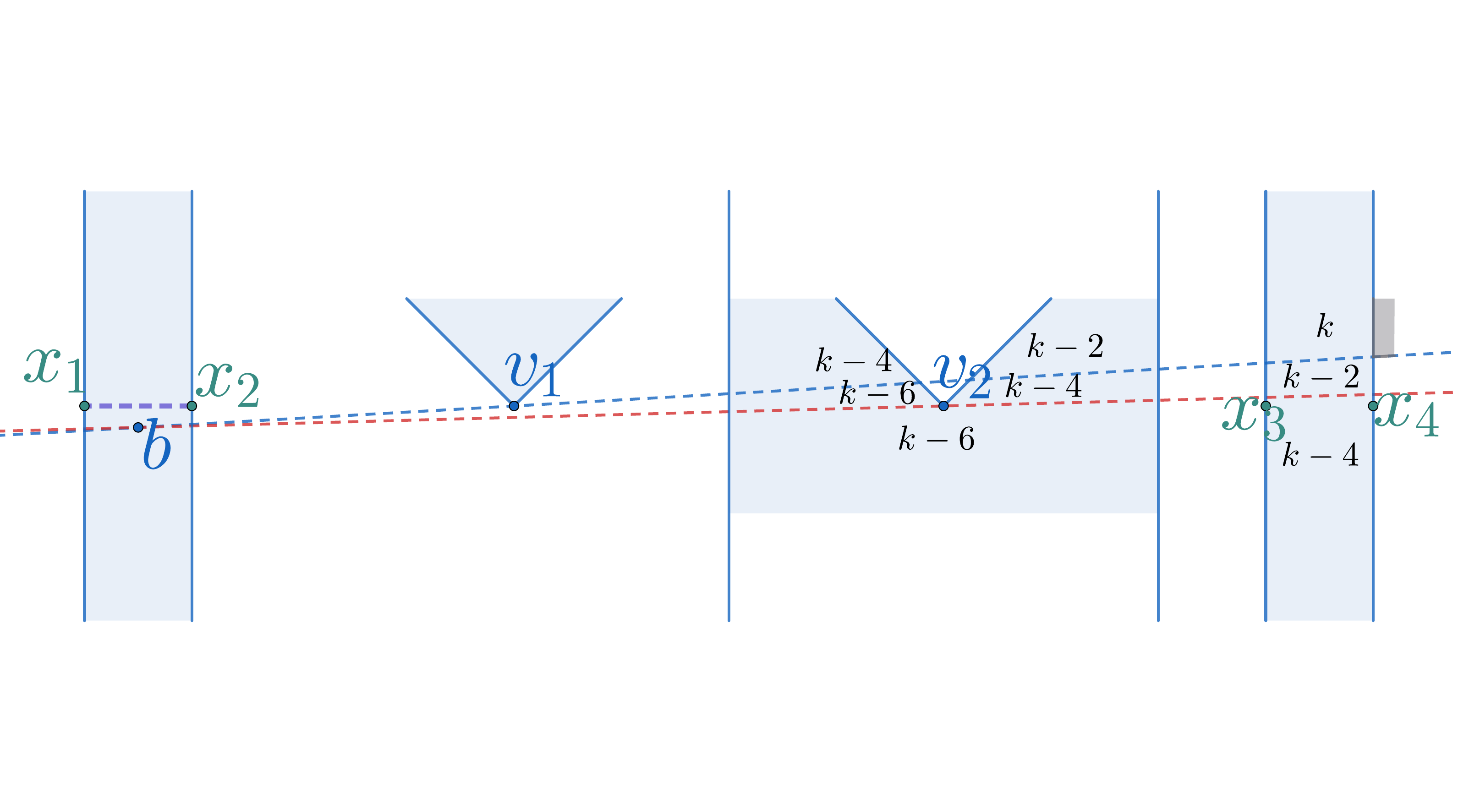}
\caption{Below $l_{g}$ }\label{fig:CRO-genericB8}
\end{subfigure}

\caption{CRS; $Z = k - 6$, $W = k - 4$}
\label{fig:CRO-generic8}

\end{figure}

For $Z \leq k - 8$, $v_{2}$ and its surroundings are entirely visible. For $W \leq k - 6$, $x_{3}x_{4}$ and its surroundings are entirely visible.
% CRS(2) at $v_{2}$
% CRS(4) at $v_{2}$ and $x_{3}x_{4}$
% CRS(6) at $x_{3}x_{4}$
\end{proof}

\section{CRO}
% CRO(2) at $v_{2}$ (below $l_{b}$, above $l_{r}$)\\
% CRO(4) at $v_{2}$ (below $l_{b}$ and above $l_{r}$) and $x_{3}x_{4}$ (invis; between $l_{b}$ and $l_{r}$)

\begin{lemma}
\label{lemma:CRO}
A partition line is needed in the following cases:
\begin{itemize}
\renewcommand{\labelitemi}{$\bullet$}
    \item $Z = k$ (Figure~\ref{fig:CRS-generic2})
    \item $Z = k - 2$ (Figure~\ref{fig:CRS-generic4})
    \item $W = k$ (Figure~\ref{fig:CRS-generic4}
    \item $Z = k - 4$ (Figure~\ref{fig:CRS-generic6})
    \item $W = k - 2$ (Figure~\ref{fig:CRS-generic6})
    \item $W = k - 4$ (Figure~\ref{fig:CRS-generic8})
\end{itemize}

\end{lemma}
\begin{proof}
    See Figures~\ref{fig:CRS-generic0}-~\ref{fig:CRS-generic8}.
\end{proof}
Note: for $Z \geq k + 4$, $v_{2}$ and its surroundings are entirely in shadow. For $W \geq k + 6$, $x_{3}x_{4}$ and its surroundings are entirely in shadow.

 \begin{figure}[H]
\centering
\begin{subfigure}[b]{.49\linewidth}
\includegraphics[width=\linewidth]{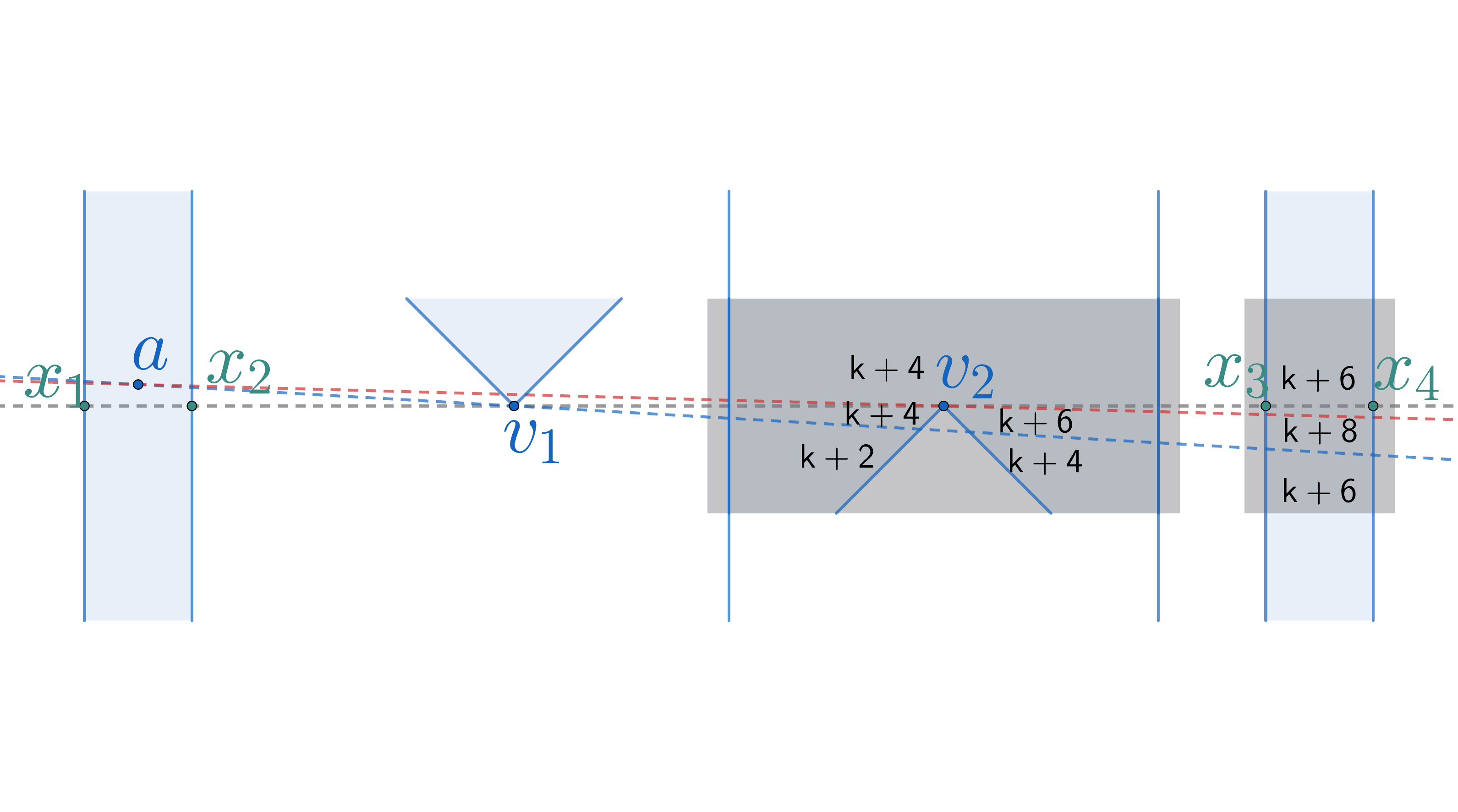}
\caption{Above $l_{g}$}\label{fig:CRS-genericA0}
\end{subfigure}
\begin{subfigure}[b]{.49\linewidth}
\includegraphics[width=\linewidth]{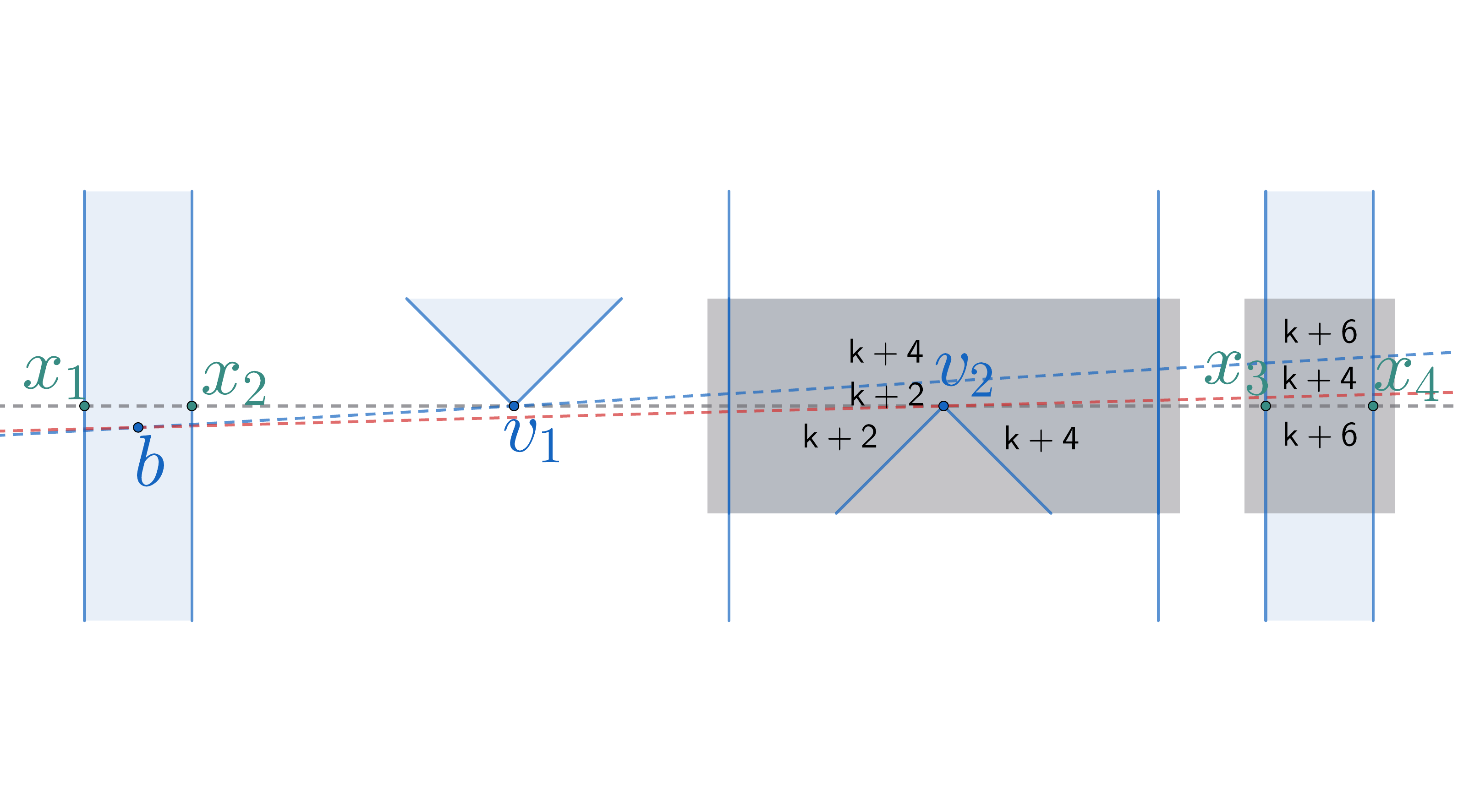}
\caption{Below $l_{g}$}\label{fig:CRS-genericB0}
\end{subfigure}

\caption{CRO; $Z=k+2$; $W=k+4$}
\label{fig:CRS-generic0}
\end{figure}

 \begin{figure}[H]
\centering
\begin{subfigure}[b]{.49\linewidth}
\includegraphics[width=\linewidth]{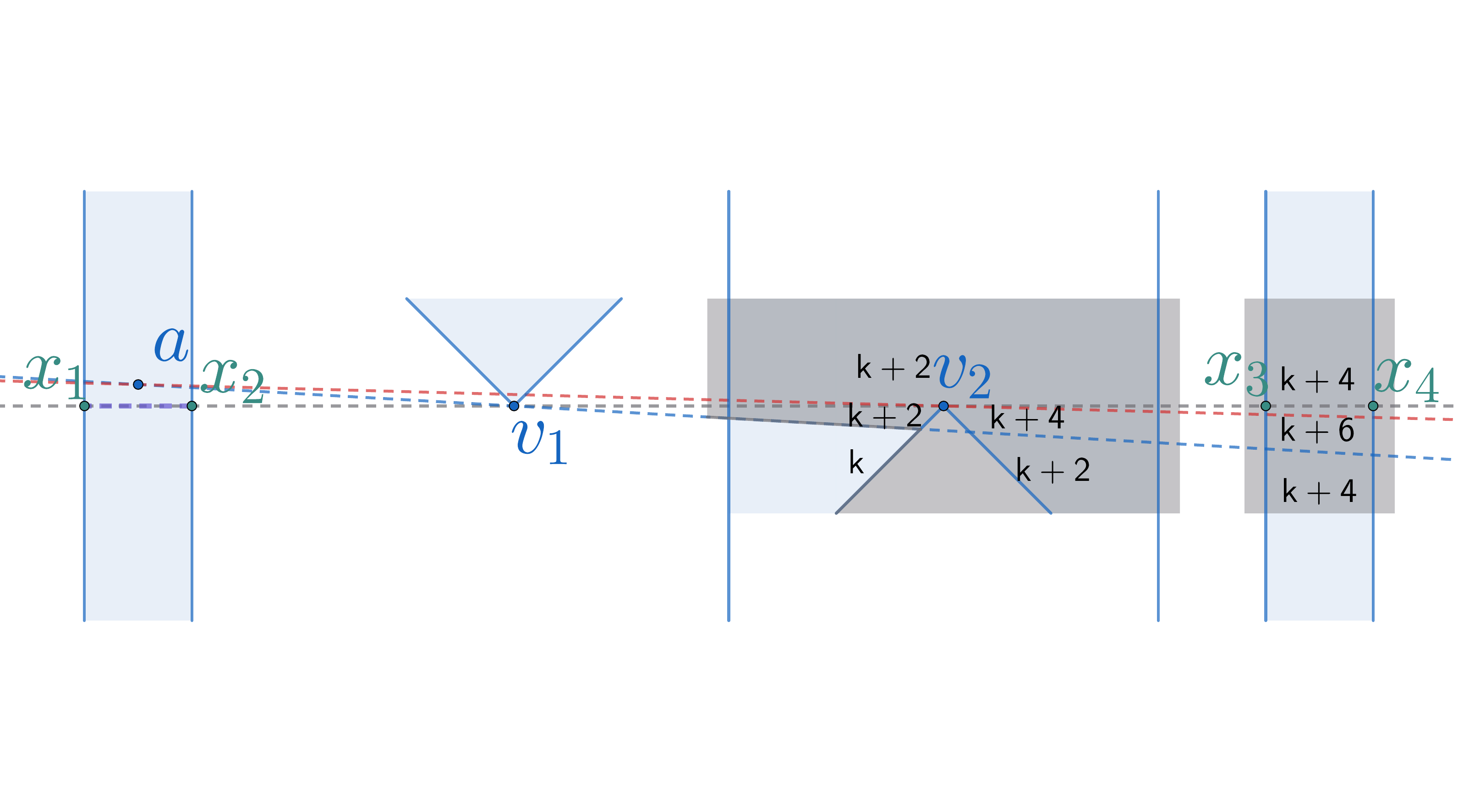}
\caption{Above $l_{g}$}\label{fig:CRS-genericA2}
\end{subfigure}
\begin{subfigure}[b]{.49\linewidth}
\includegraphics[width=\linewidth]{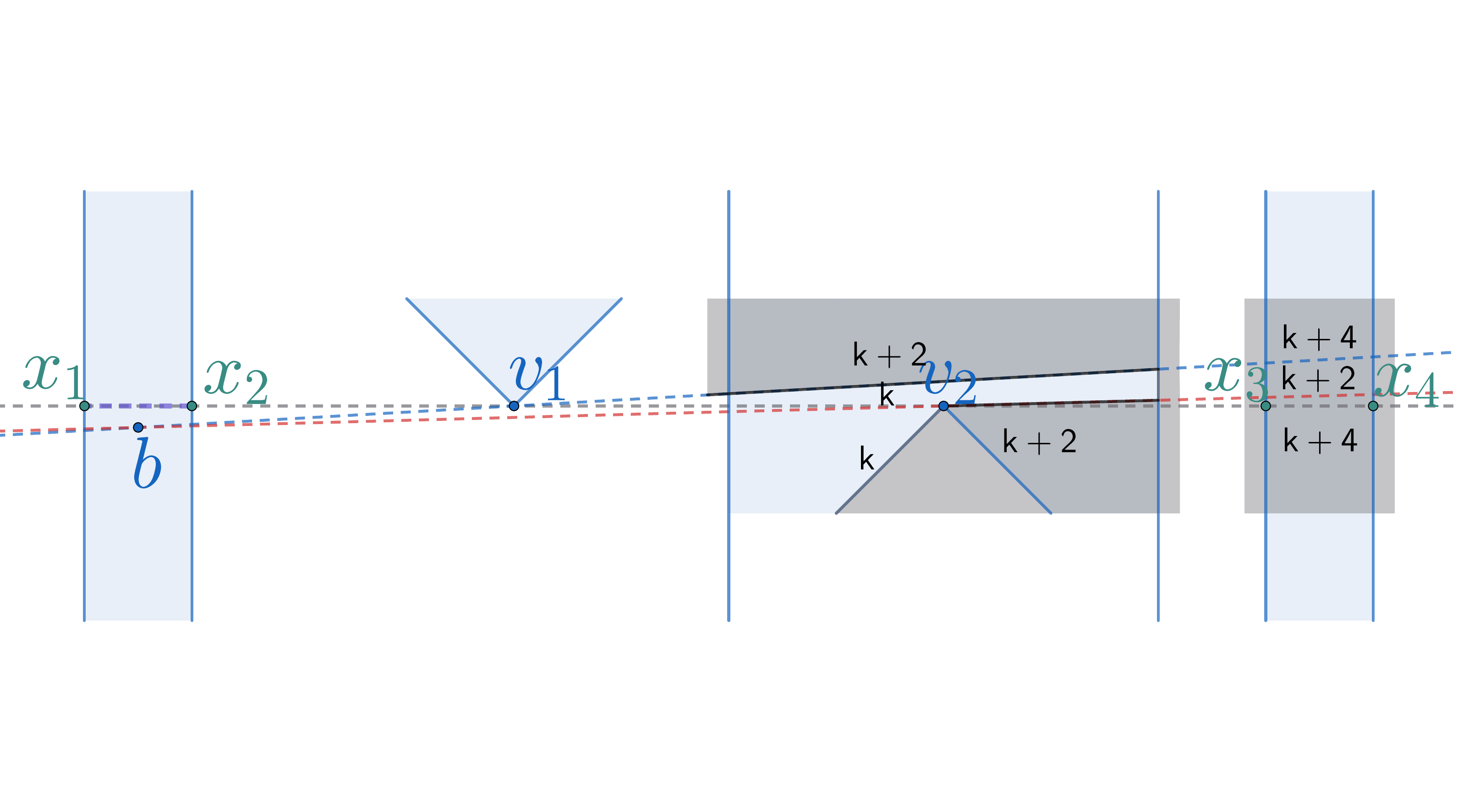}
\caption{Below $l_{g}$}\label{fig:CRS-genericB2}
\end{subfigure}

\caption{CRO; $Z=k$ and $W=k+2$}
\label{fig:CRS-generic2}
\end{figure}

 \begin{figure}[H]
\centering
\begin{subfigure}[b]{.49\linewidth}
\includegraphics[width=\linewidth]{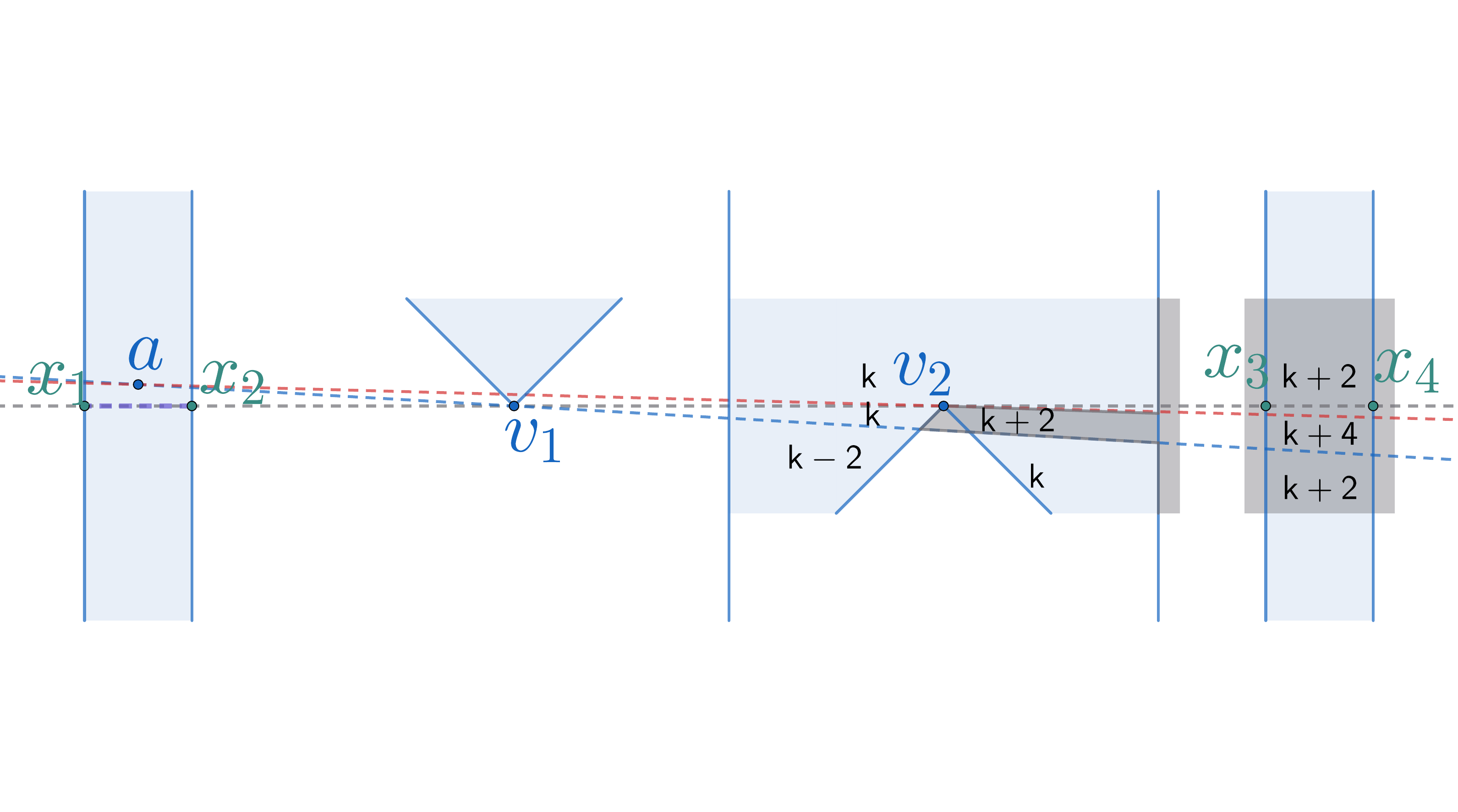}
\caption{Above $l_{g}$}\label{fig:CRS-genericA4}
\end{subfigure}
\begin{subfigure}[b]{.49\linewidth}
\includegraphics[width=\linewidth]{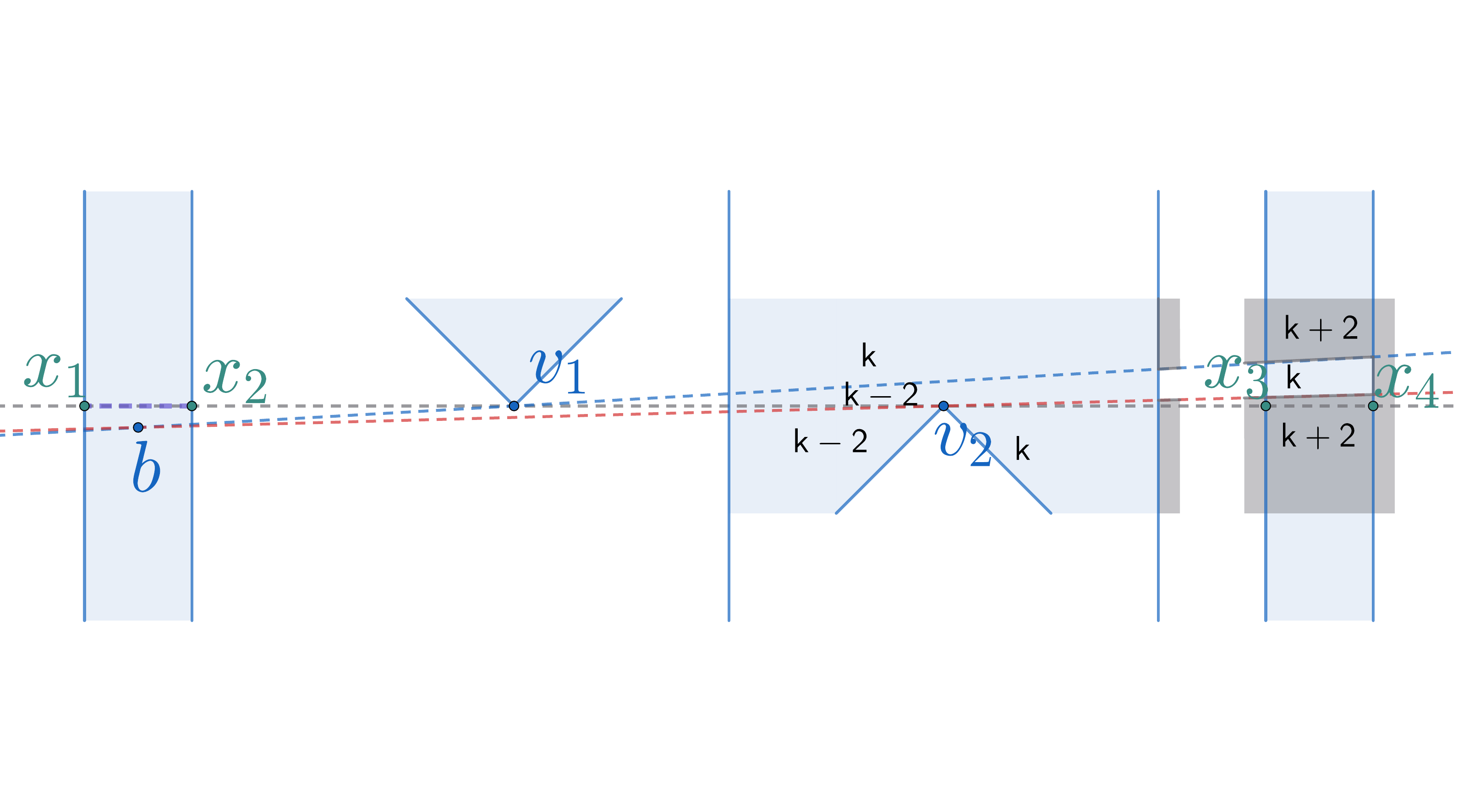}
\caption{Below $l_{g}$}\label{fig:CRS-genericB4}
\end{subfigure}

\caption{CRO; $Z = k - 2$, $W = k$}
\label{fig:CRS-generic4}
\end{figure}

 \begin{figure}[H]
\centering
\begin{subfigure}[b]{.49\linewidth}
\includegraphics[width=\linewidth]{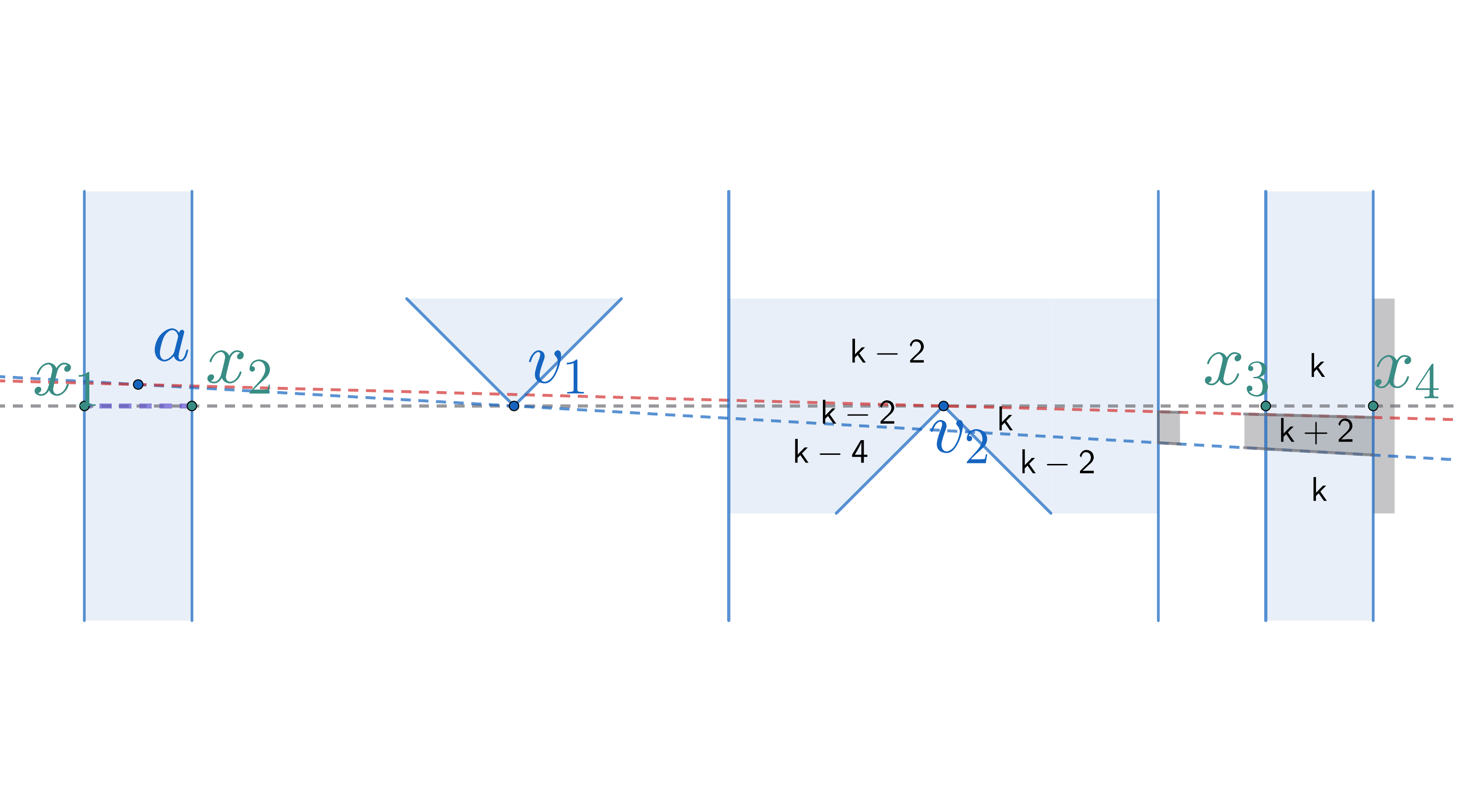}
\caption{Above $l_{g}$}\label{fig:CRS-genericA6}
\end{subfigure}
\begin{subfigure}[b]{.49\linewidth}
\includegraphics[width=\linewidth]{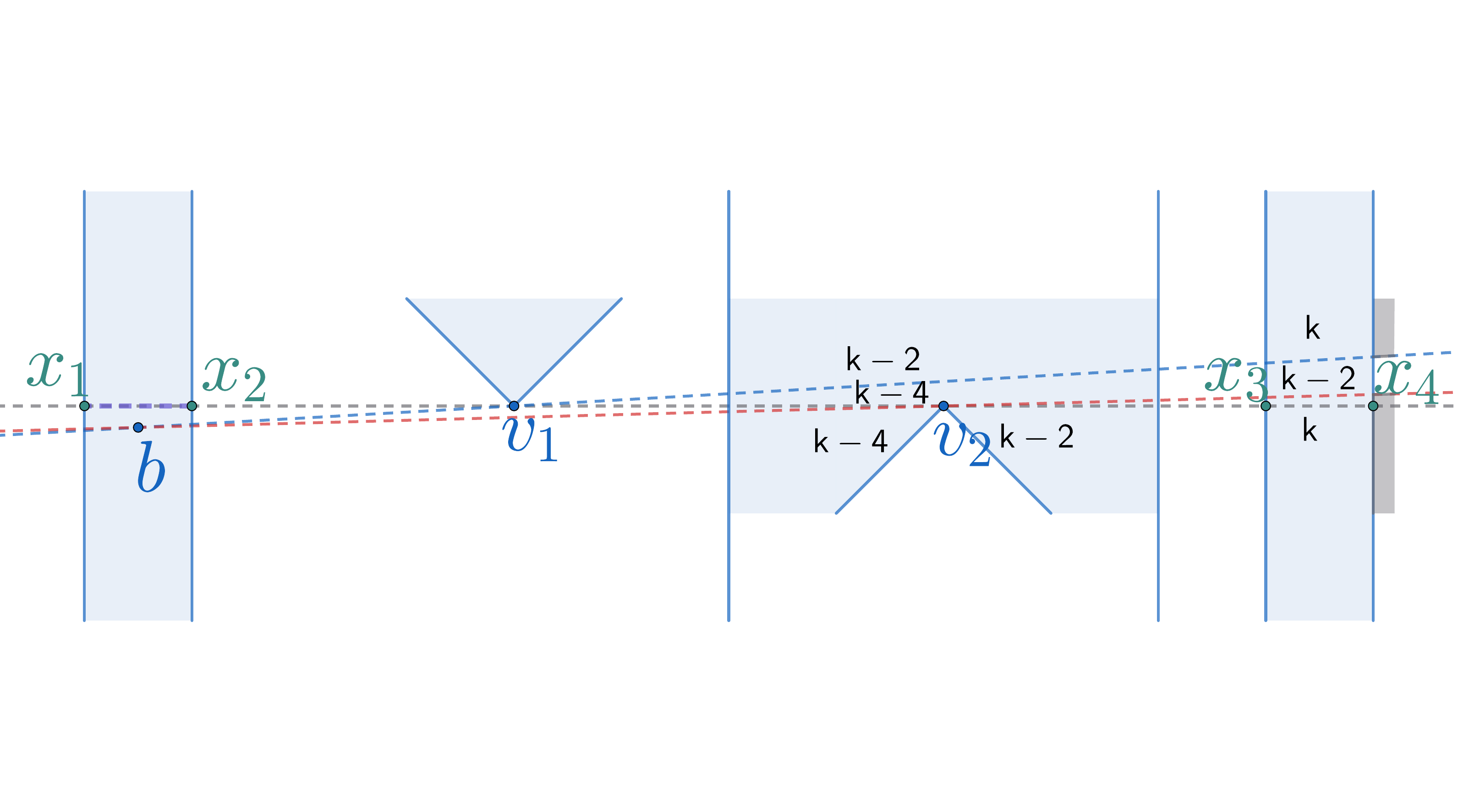}
\caption{Below $l_{g}$}\label{fig:CRS-genericB6}
\end{subfigure}

\caption{CRO; $Z = k - 4$, $W = k - 2$}
\label{fig:CRS-generic6}
\end{figure}

 \begin{figure}[H]
\centering
\begin{subfigure}[b]{.49\linewidth}
\includegraphics[width=\linewidth]{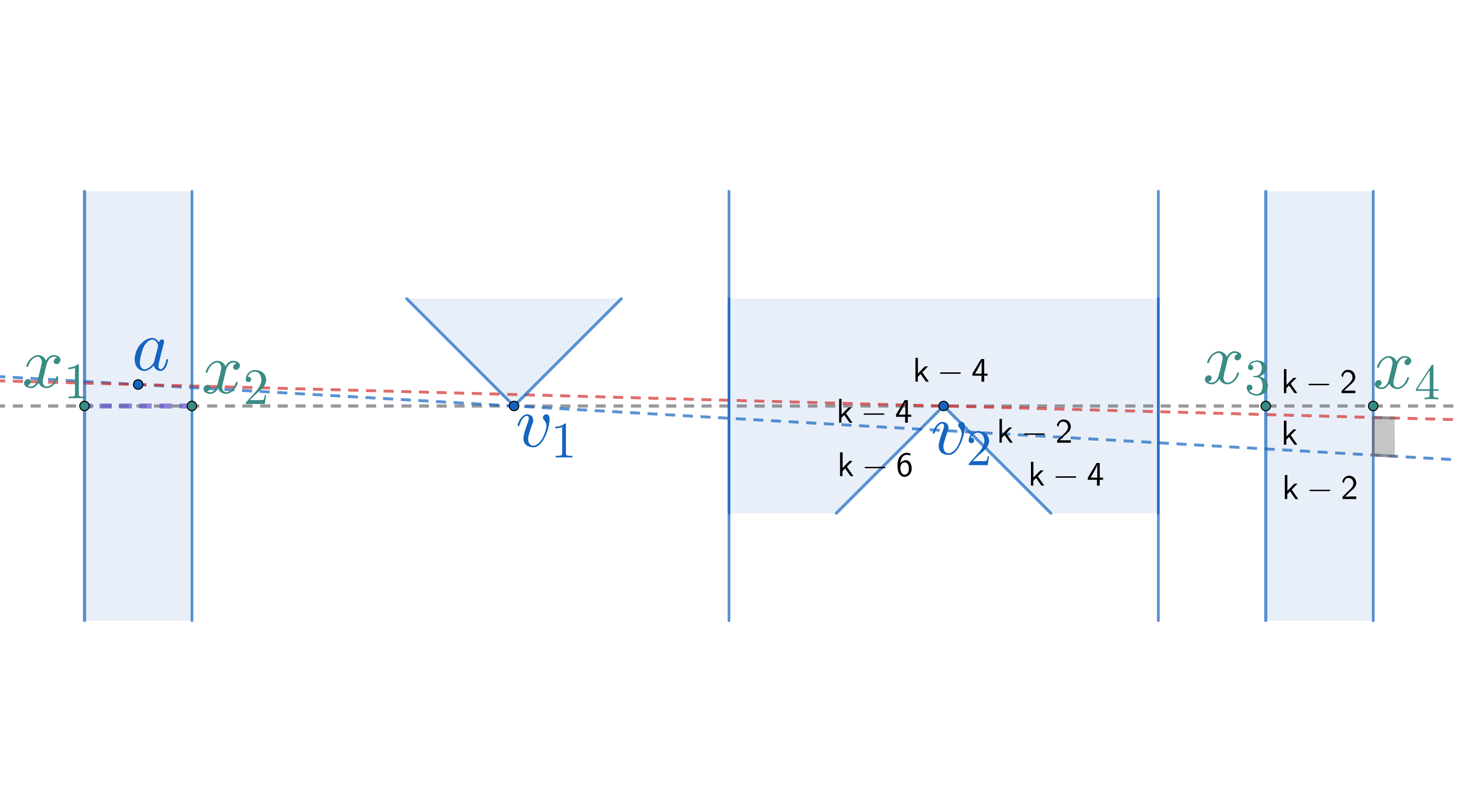}
\caption{Above $l_{g}$}\label{fig:CRS-genericA6}
\end{subfigure}
\begin{subfigure}[b]{.49\linewidth}
\includegraphics[width=\linewidth]{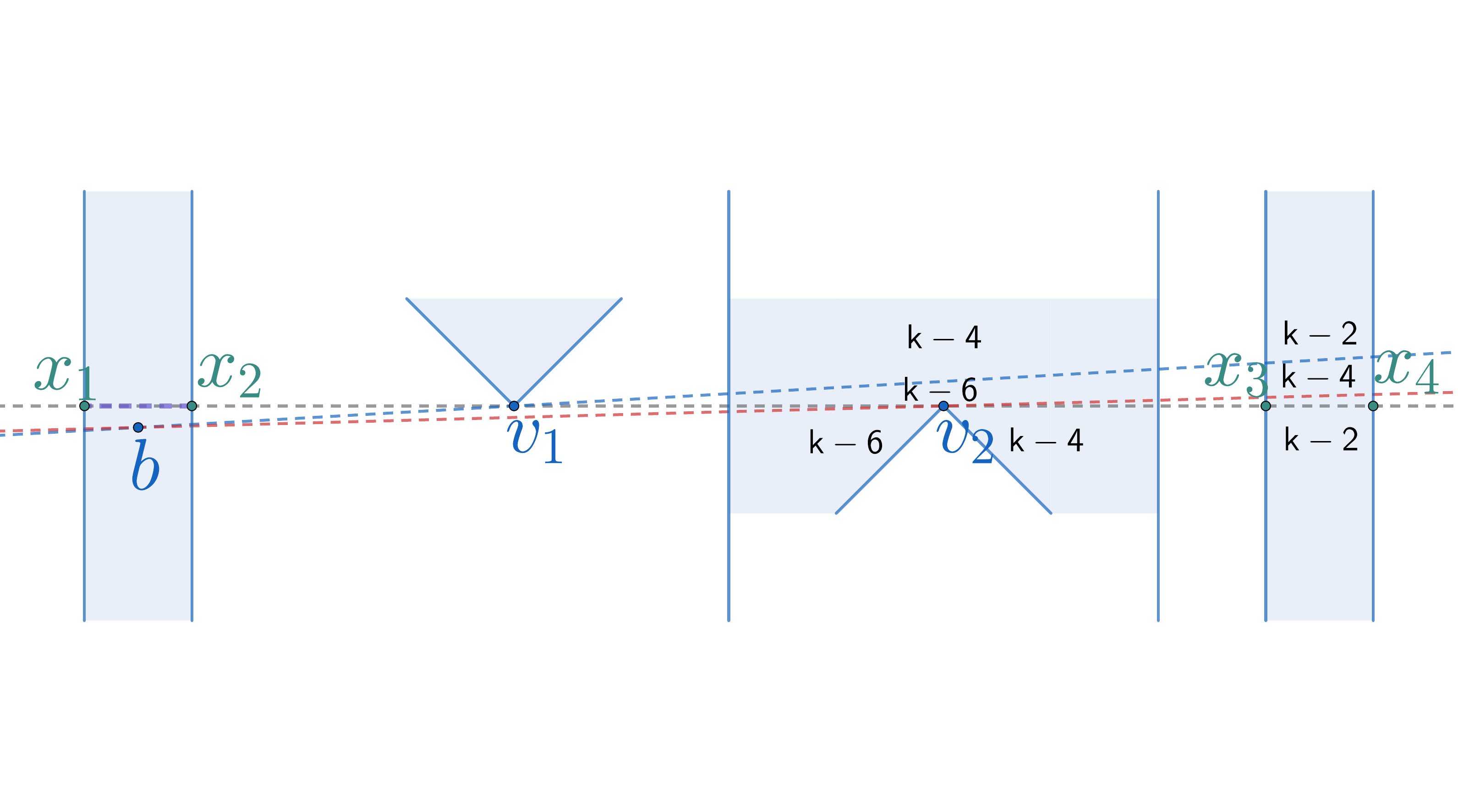}
\caption{Below $l_{g}$}\label{fig:CRS-genericB6}
\end{subfigure}

\caption{CRO; $Z = k - 6$, $W = k - 4$} 
\label{fig:CRS-generic8}
\end{figure}

For $Z \leq k - 8$, $v_{2}$ and its surroundings are entirely visible. For $W \leq k - 6$, $x_{3}x_{4}$ and its surroundings are entirely visible.

\section{RCS}

\begin{lemma}
\label{lemma:RCS}
    
A partition line is needed in the following cases:
\begin{itemize}
\renewcommand{\labelitemi}{$\bullet$}
    \item $Z = k-1$ (Figure~\ref{fig:RCS-generic-2}) 
    \item $W = k$ (Figure~\ref{fig:RCS-generic-2})
    \item $Z = k - 3$ (Figure~\ref{fig:RCS-generic-4})
    \item $W = k - 2$ (Figure~\ref{fig:RCS-generic-4})
    \item $W = k - 4$ (Figure~\ref{fig:RCS-generic-6})
\end{itemize}

\end{lemma}
\begin{proof}
    See Figures~\ref{fig:RCS-generic-2}-~\ref{fig:RCS-generic-6}.
\end{proof}

Note: for $Z \geq k + 1$, $v_{2}$ and its surroundings are entirely in shadow. For $W \geq k + 2$, $x_{3}x_{4}$ and its surroundings are entirely in shadow.

 \begin{figure}[H]
\centering
\begin{subfigure}[b]{.49\linewidth}
\includegraphics[width=\linewidth]{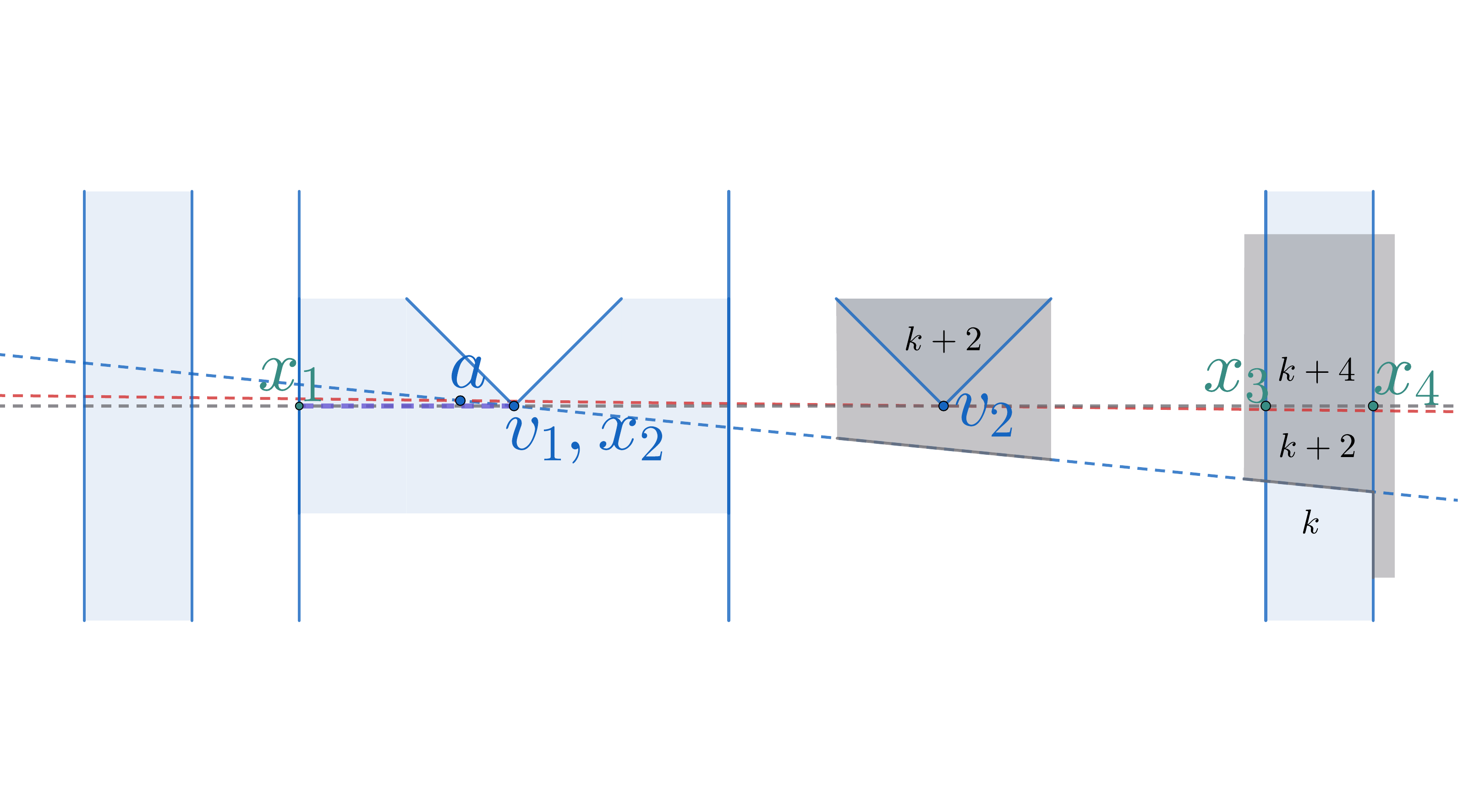}
\caption{Above $l_{g}$}\label{fig:RCS-genericA2}
\end{subfigure}
\begin{subfigure}[b]{.49\linewidth}
\includegraphics[width=\linewidth]{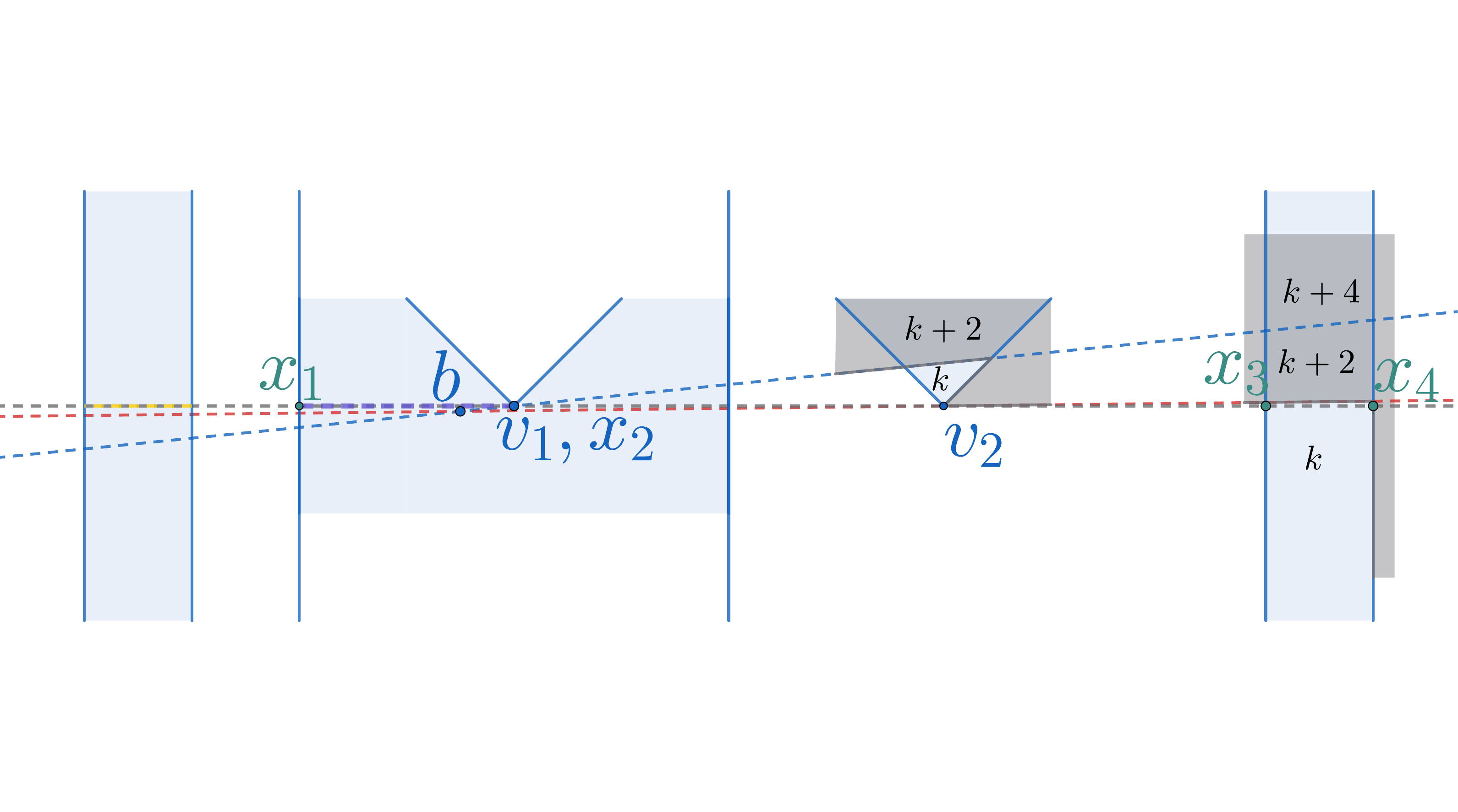}
\caption{Below $l_{g}$}\label{fig:RCS-genericB2}
\end{subfigure}

\caption{RCS; $Z = k - 1$, $W = k$}
\label{fig:RCS-generic-2}
\end{figure}

 \begin{figure}[H]
\centering
\begin{subfigure}[b]{.49\linewidth}
\includegraphics[width=\linewidth]{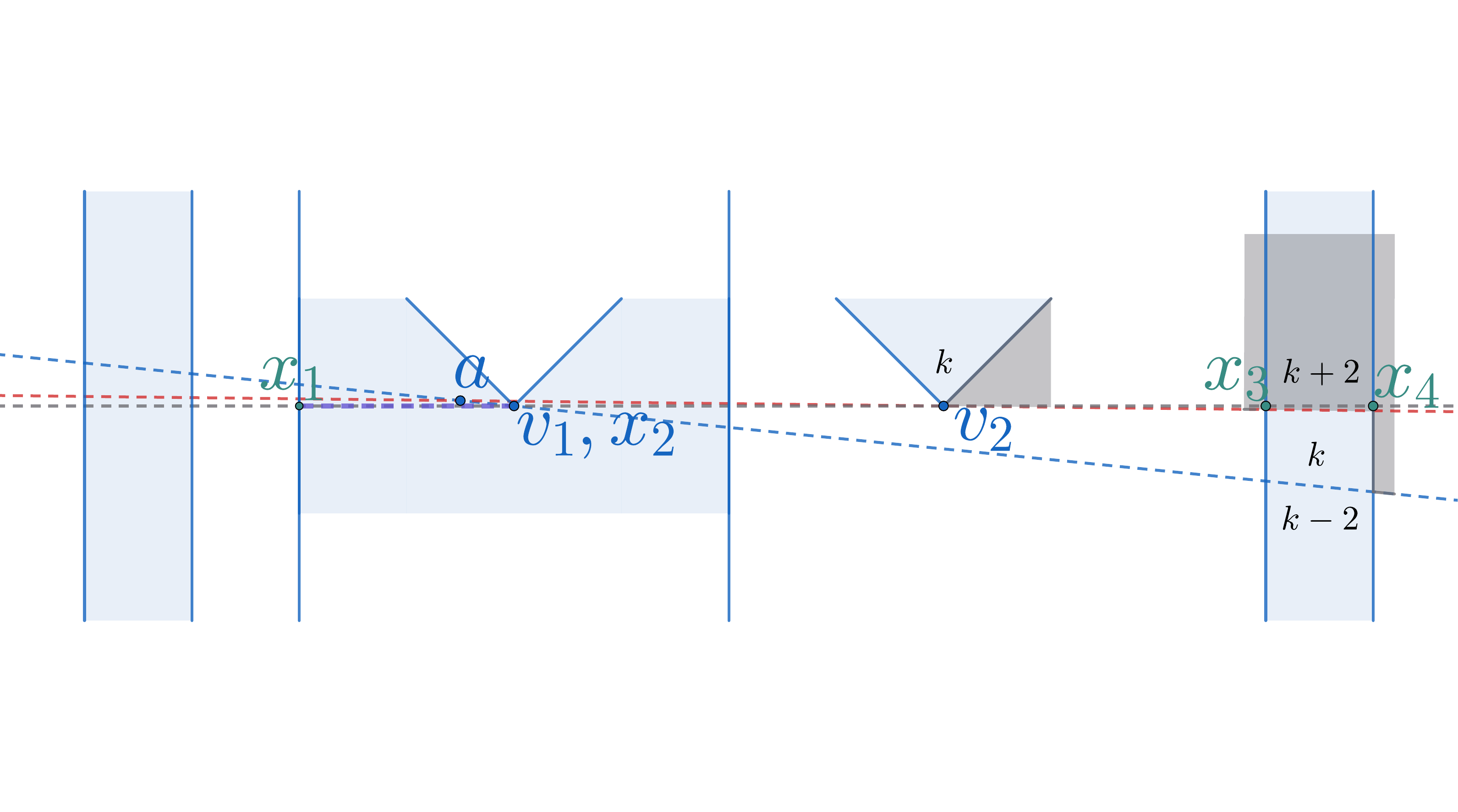}
\caption{Above $l_{g}$}\label{fig:RCS-genericA4}
\end{subfigure}
\begin{subfigure}[b]{.49\linewidth}
\includegraphics[width=\linewidth]{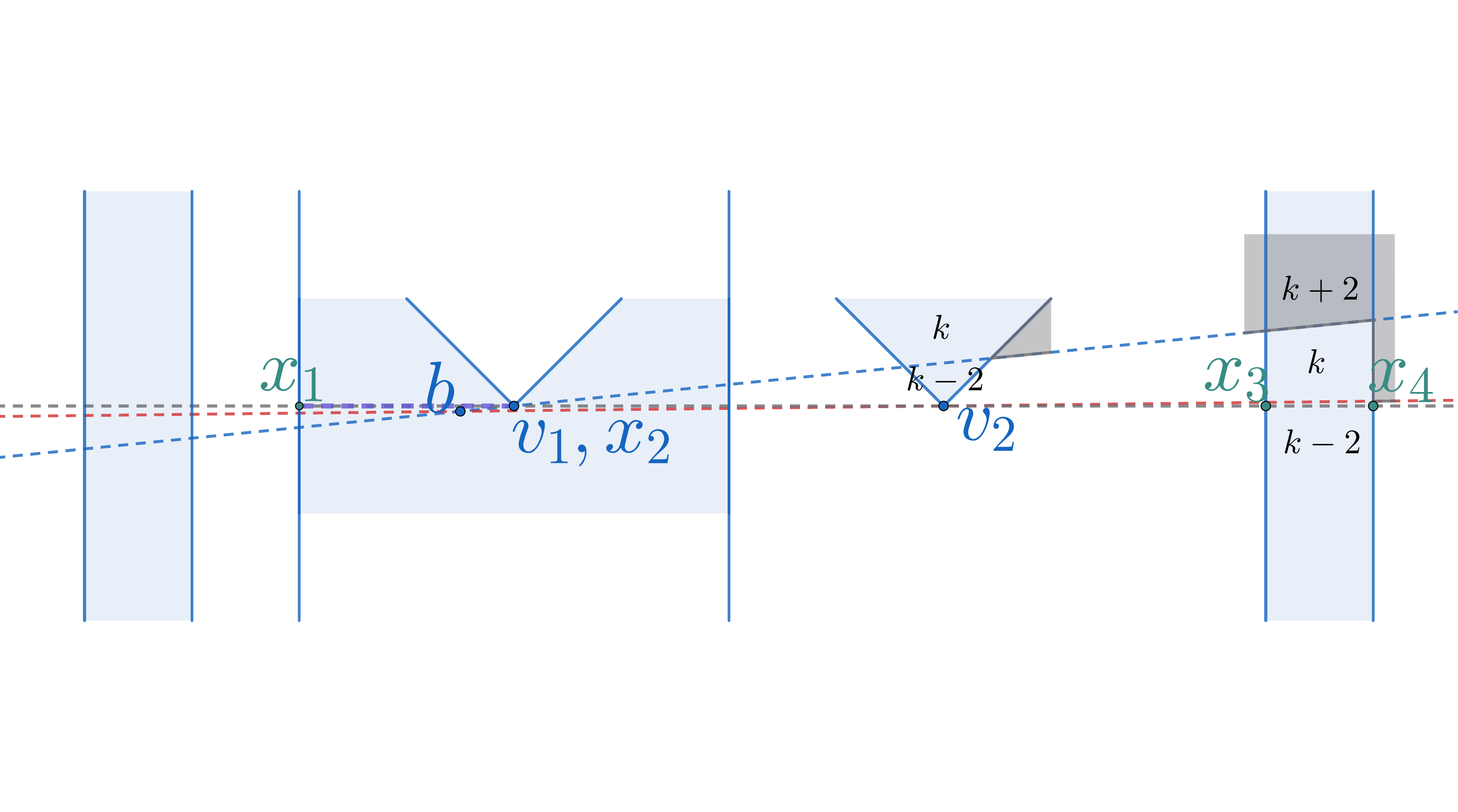}
\caption{Below $l_{g}$}\label{fig:RCS-genericB4}
\end{subfigure}

\caption{RCS; $Z = k - 3$, $W = k - 2$}
\label{fig:RCS-generic-4}
\end{figure}

 \begin{figure}[H]
\centering
\begin{subfigure}[b]{.49\linewidth}
\includegraphics[width=\linewidth]{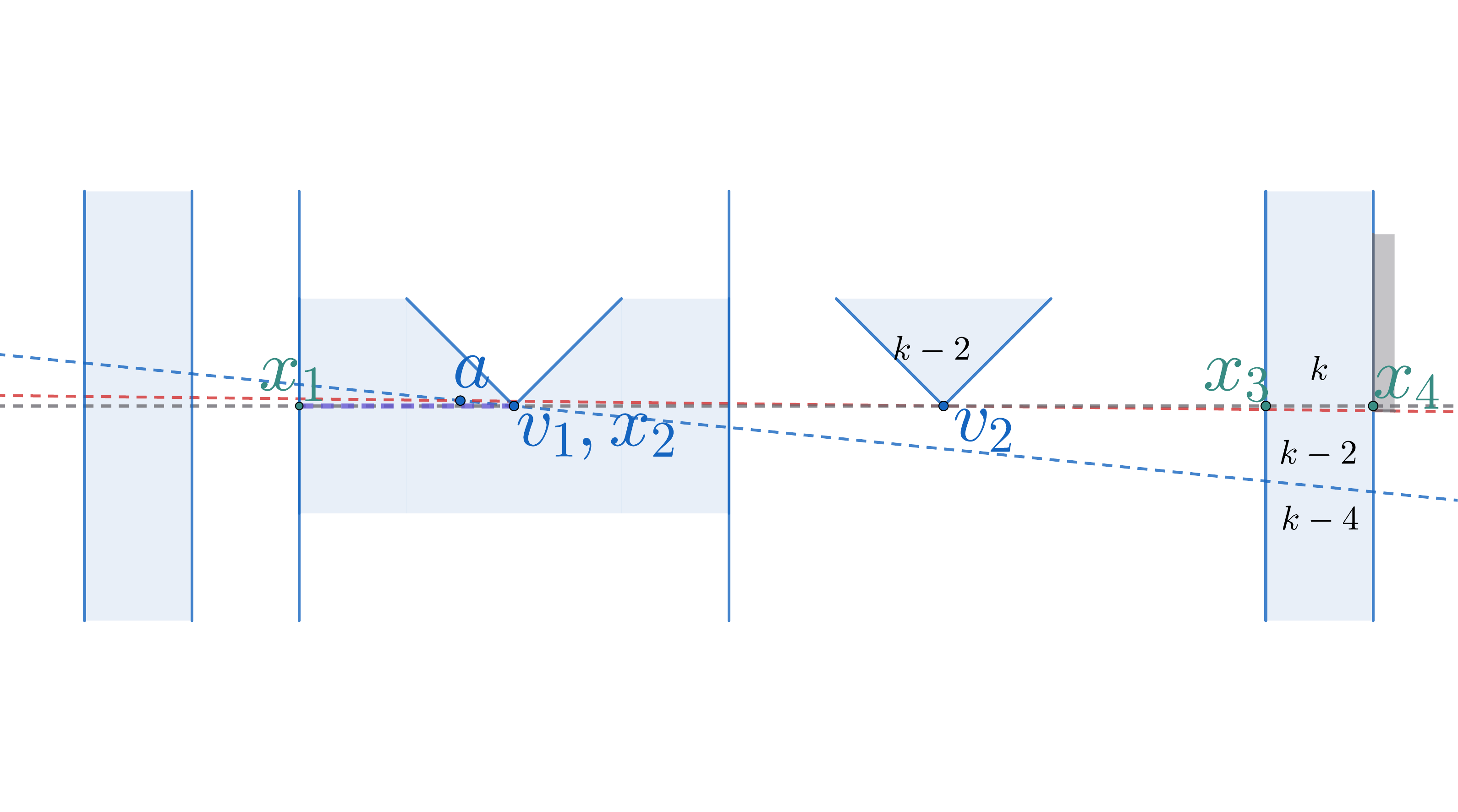}
\caption{Above $l_{g}$}\label{fig:RCS-genericA6}
\end{subfigure}
\begin{subfigure}[b]{.49\linewidth}
\includegraphics[width=\linewidth]{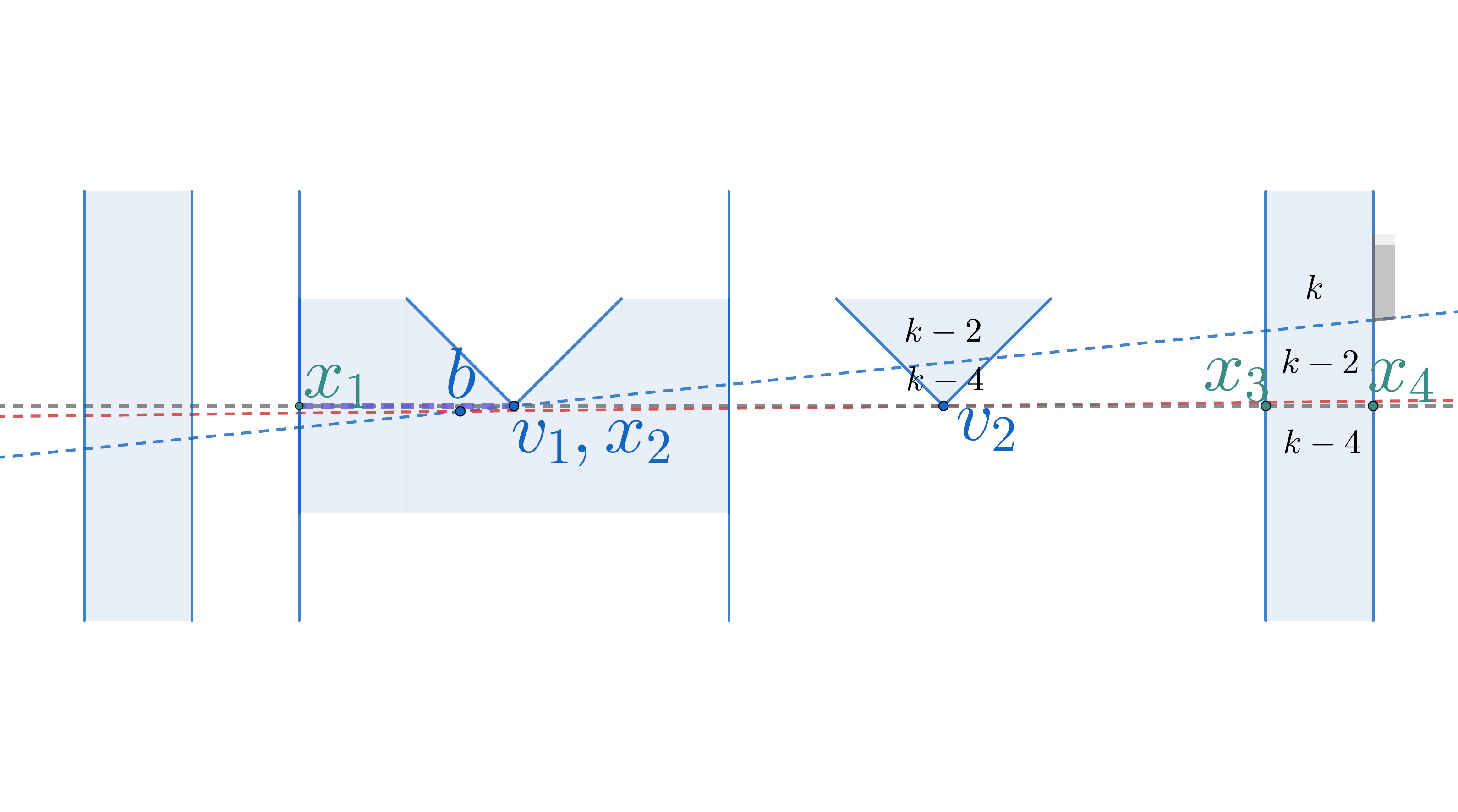}
\caption{Below $l_{g}$}\label{fig:RCS-genericB6}
\end{subfigure}

\caption{RCS; $Z = k - 5$, $W = k - 4$}
\label{fig:RCS-generic-6}
\end{figure}

For $Z \leq k - 7$, $v_{2}$ and its surroundings are entirely visible. For $W \leq k - 6$, $x_{3}x_{4}$ and its surroundings are entirely visible.

\section{RCO}

\begin{lemma}
\label{lemma:RCO}
A partition line is needed in the following cases:
\renewcommand{\labelitemi}{$\bullet$}
\begin{itemize}
    \item $Z = k - 1$ (Figure~\ref{fig:RCO-generic-2})
    \item $W = k$ (Figure~\ref{fig:RCO-generic-2})
    \item $Z = k - 3$ (Figure~\ref{fig:RCO-generic-4})
    \item $W = k - 2$ (Figure~\ref{fig:RCO-generic-4})
    \item $W = k - 4$  (Figure~\ref{fig:RCO-generic-6})
\end{itemize}
\end{lemma}
\begin{proof}
    See Figures~\ref{fig:RCO-generic-2}-~\ref{fig:RCO-generic-6}.
\end{proof}

Note: for $Z \geq k + 1$, $v_{2}$ and its surroundings are entirely in shadow. For $W \geq k + 2$, $x_{3}x_{4}$ and its surroundings are entirely in shadow.

 \begin{figure}[H]
\centering
\begin{subfigure}[b]{.49\linewidth}
\includegraphics[width=\linewidth]{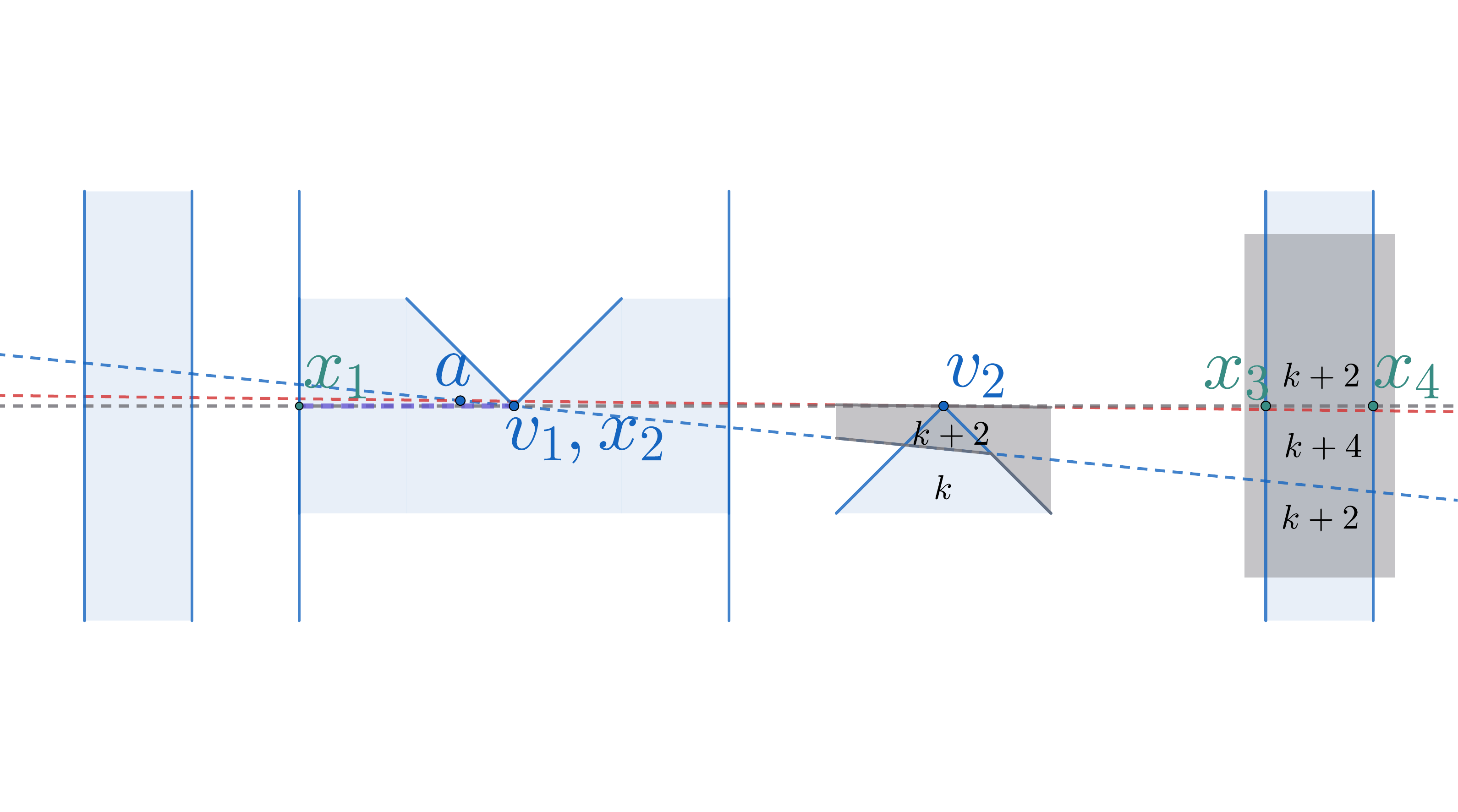}
\caption{Above $l_{g}$}\label{fig:RCO-genericA2}
\end{subfigure}
\begin{subfigure}[b]{.49\linewidth}
\includegraphics[width=\linewidth]{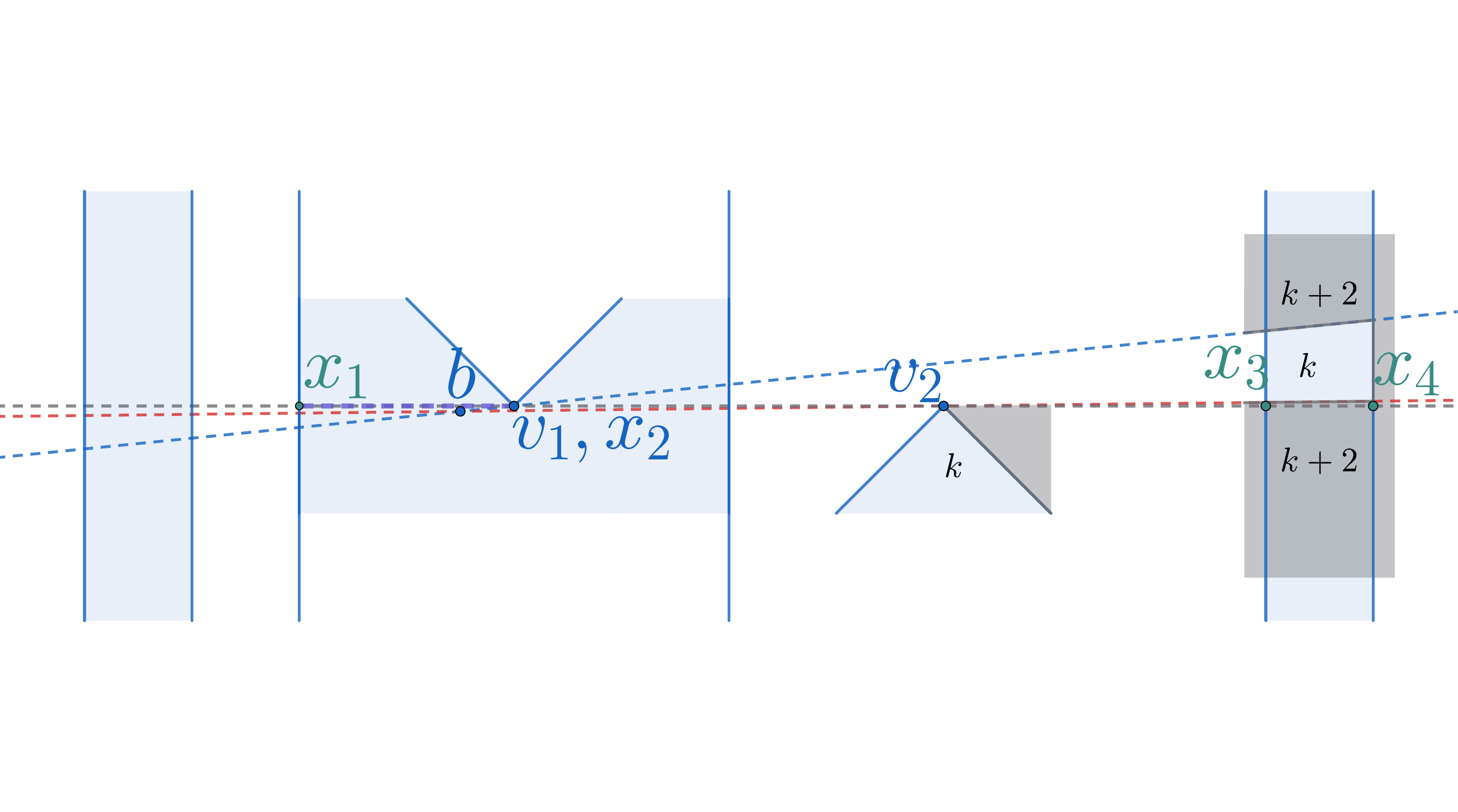}
\caption{Below $l_{g}$}\label{fig:RCO-genericB2}
\end{subfigure}

\caption{RCO; $Z = k - 1$, $W = k$}
\label{fig:RCO-generic-2}
\end{figure}

 \begin{figure}[H]
\centering
\begin{subfigure}[b]{.49\linewidth}
\includegraphics[width=\linewidth]{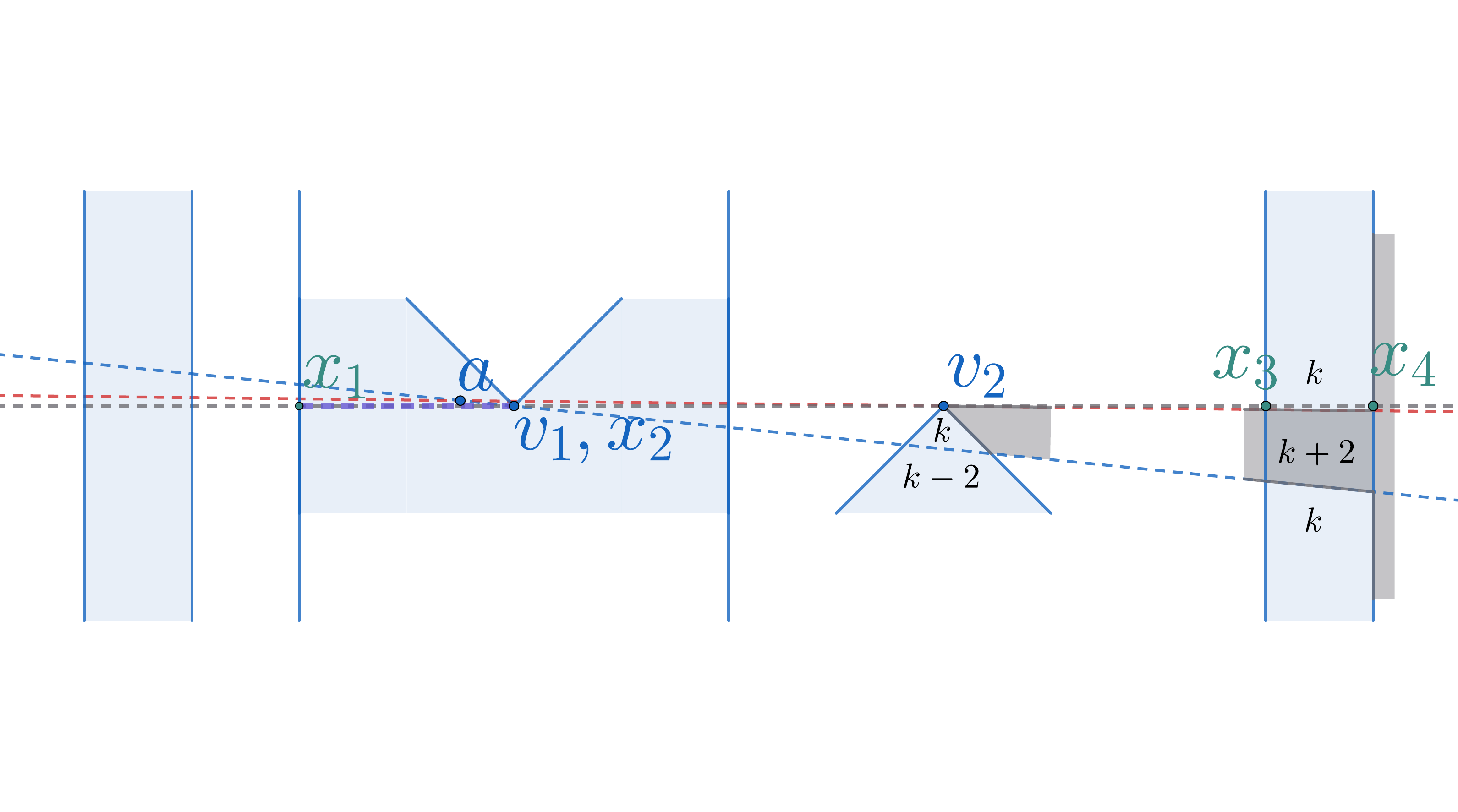}
\caption{Above $l_{g}$}\label{fig:RCO-genericA4}
\end{subfigure}
\begin{subfigure}[b]{.49\linewidth}
\includegraphics[width=\linewidth]{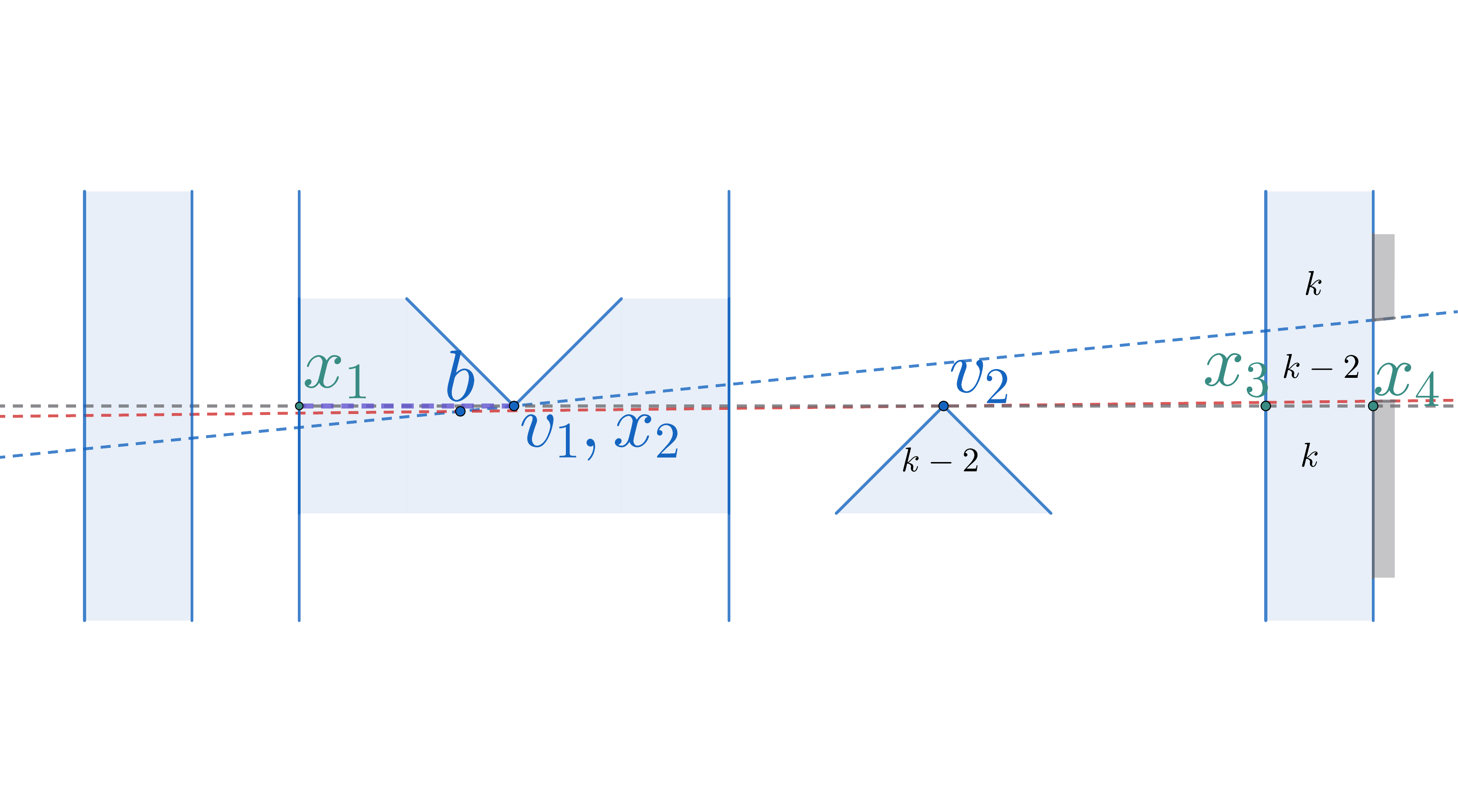}
\caption{Below $l_{g}$}\label{fig:RCO-genericB4}
\end{subfigure}

\caption{RCO; $Z = k - 3$, $W = k - 2$}
\label{fig:RCO-generic-4}
\end{figure}

 \begin{figure}[H]
\centering
\begin{subfigure}[b]{.49\linewidth}
\includegraphics[width=\linewidth]{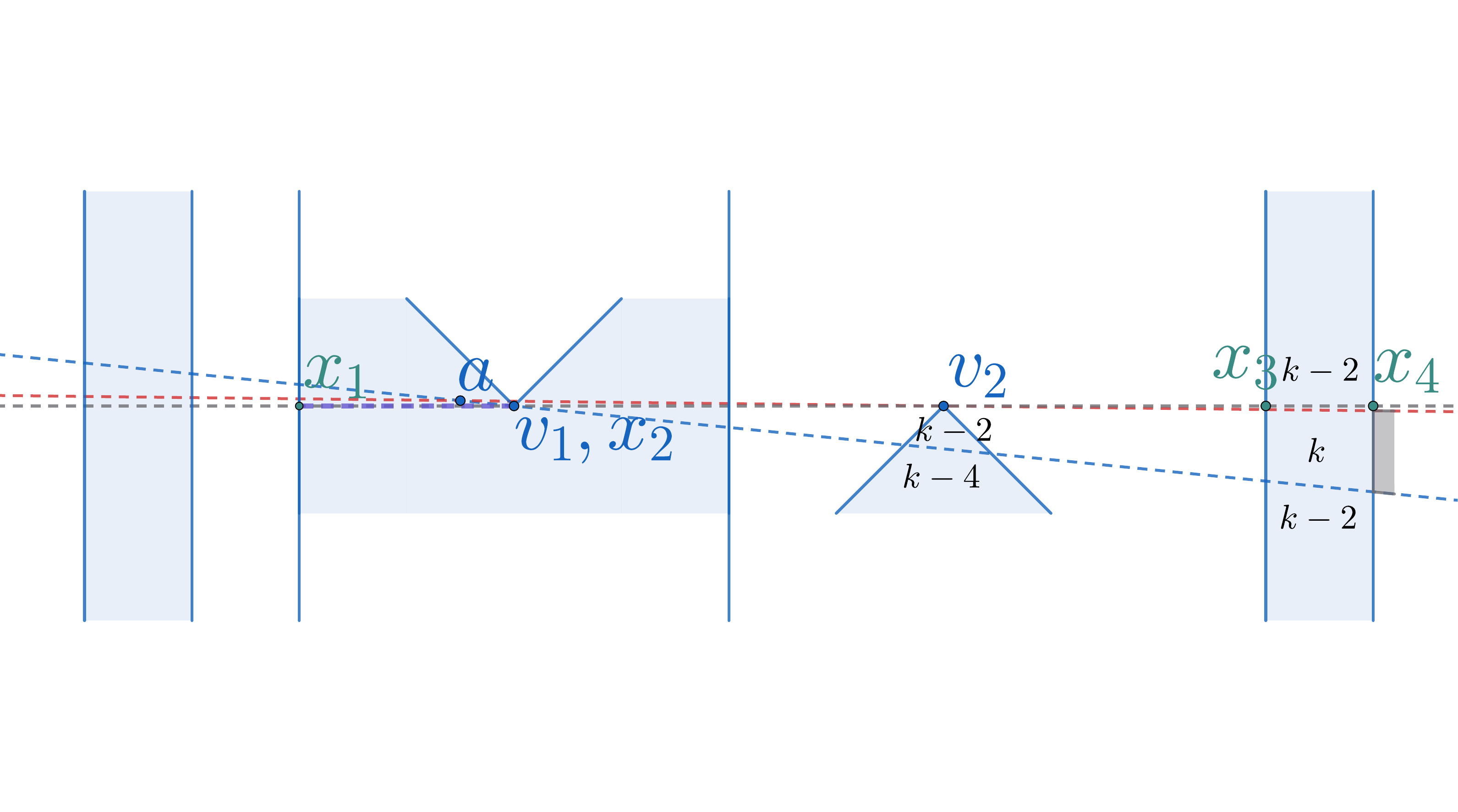}
\caption{Above $l_{g}$}\label{fig:RCO-genericA6}
\end{subfigure}
\begin{subfigure}[b]{.49\linewidth}
\includegraphics[width=\linewidth]{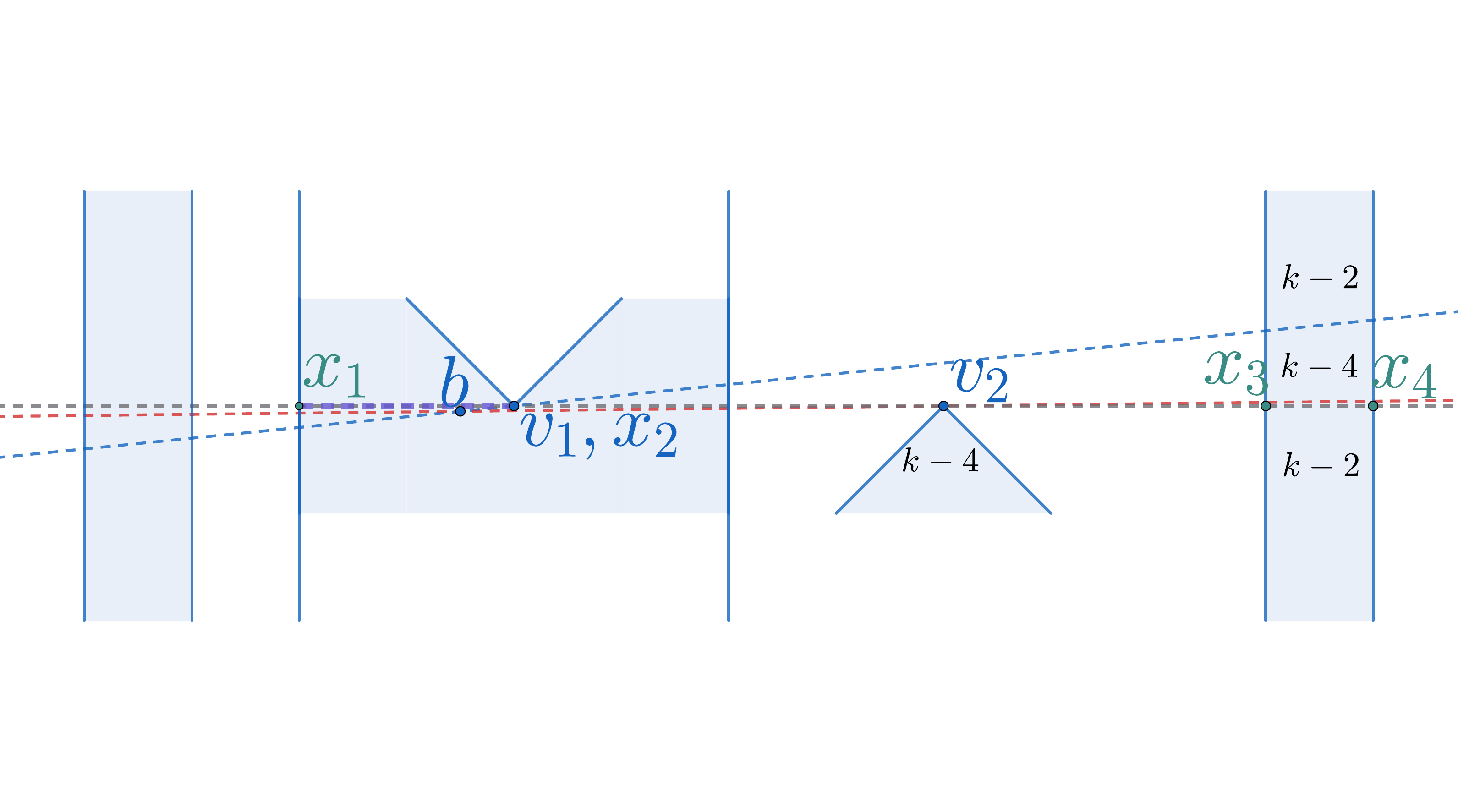}
\caption{Below $l_{g}$}\label{fig:RCO-genericB6}
\end{subfigure}

\caption{RCO; $Z = k - 5$, $W = k - 4$}
\label{fig:RCO-generic-6}
\end{figure}

Note: For $Z \leq k - 7$, $v_{2}$ and its surroundings are entirely visible. For $W \leq k - 6$, $x_{3}x_{4}$ and its surroundings are entirely visible.

\section{RRS}
\begin{lemma} 
\label{lemma:RRS}
A partition line at $x_1x_2$ is needed in the following cases:
\renewcommand{\labelitemi}{$\bullet$}
\begin{itemize}
    \item $Z = k$ (Figure~\ref{fig:RRS-generic-0})
    \item $Z = k - 2$   (Figure~\ref{fig:RRS-generic-2})
    \item $W = k$  (Figure~\ref{fig:RRS-generic-2})
    \item $Z = k - 4$  (Figure~\ref{fig:RRS-generic-4})
    \item $W = k - 2$  (Figure~\ref{fig:RRS-generic-4})
    \item $W = k - 4$  (Figure~\ref{fig:RRS-generic-6})
\end{itemize}

\end{lemma}

\begin{proof}
See Figures~\ref{fig:RRS-generic-0}-~\ref{fig:RRS-generic-6}.

Note: for $Z \geq k + 2$, $v_{2}$ and its surroundings are entirely in shadow. For $W \geq k + 4$, $x_{3}x_{4}$ and its surroundings are entirely in shadow.

 \begin{figure}[H]
\centering
\begin{subfigure}[b]{.49\linewidth}
\includegraphics[width=\linewidth]{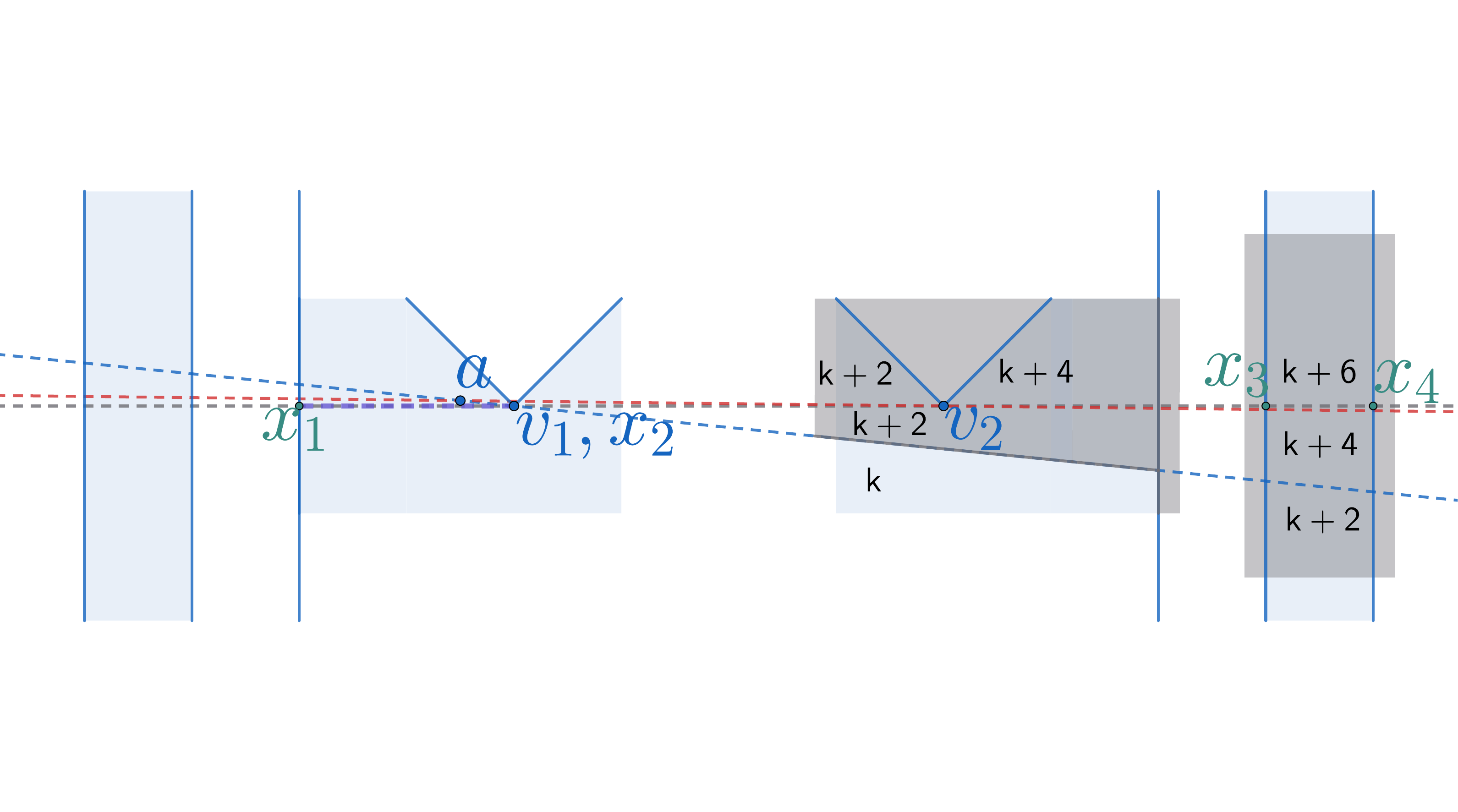}
\caption{Above $l_{g}$}\label{fig:RRS-genericA0}
\end{subfigure}
\begin{subfigure}[b]{.49\linewidth}
\includegraphics[width=\linewidth]{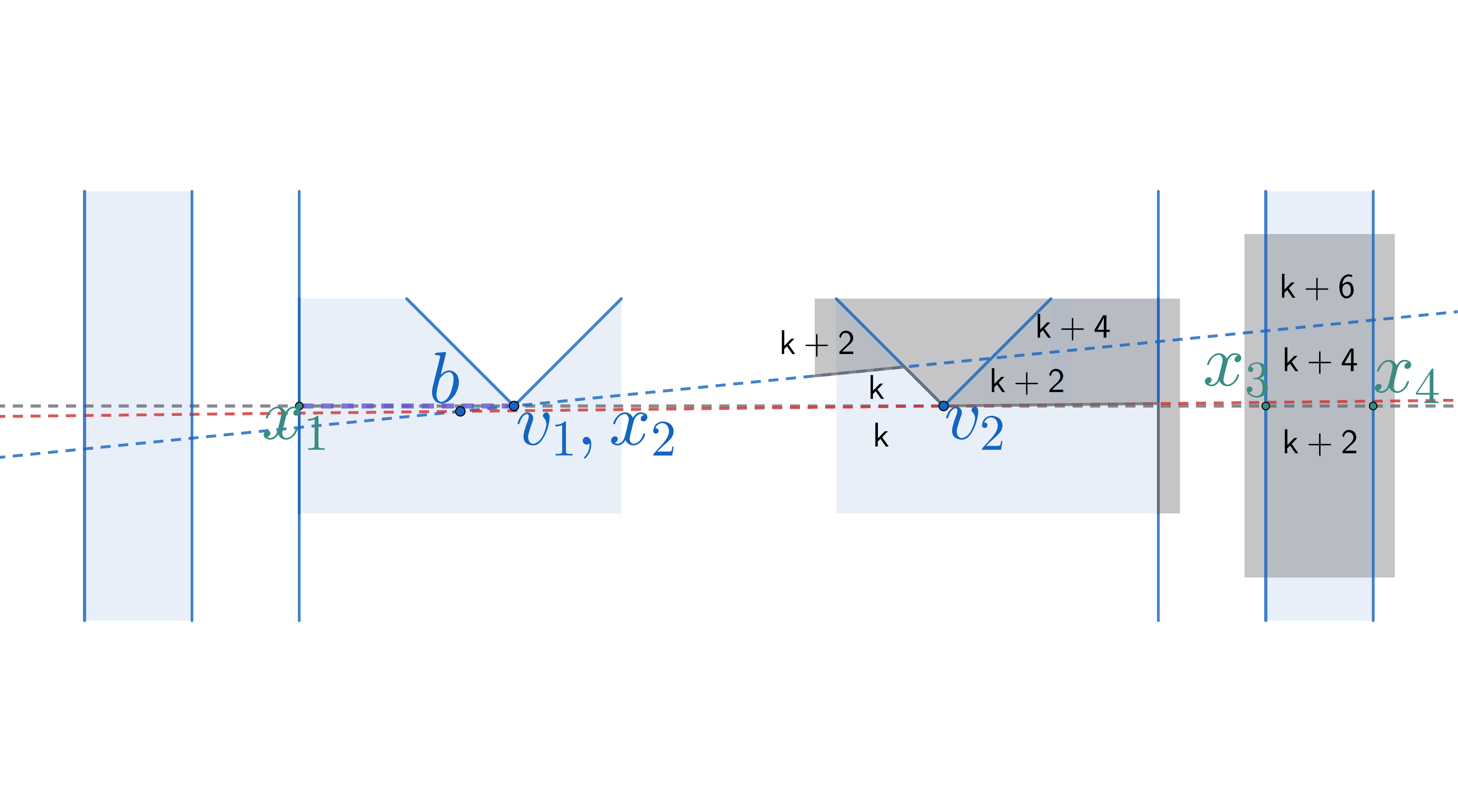}
\caption{Below $l_{g}$}\label{fig:RRS-genericB0}
\end{subfigure}

\caption{RRS; $Z = k$, $W = k + 2$}
\label{fig:RRS-generic-0}
\end{figure}

 \begin{figure}[H]
\centering
\begin{subfigure}[b]{.49\linewidth}
\includegraphics[width=\linewidth]{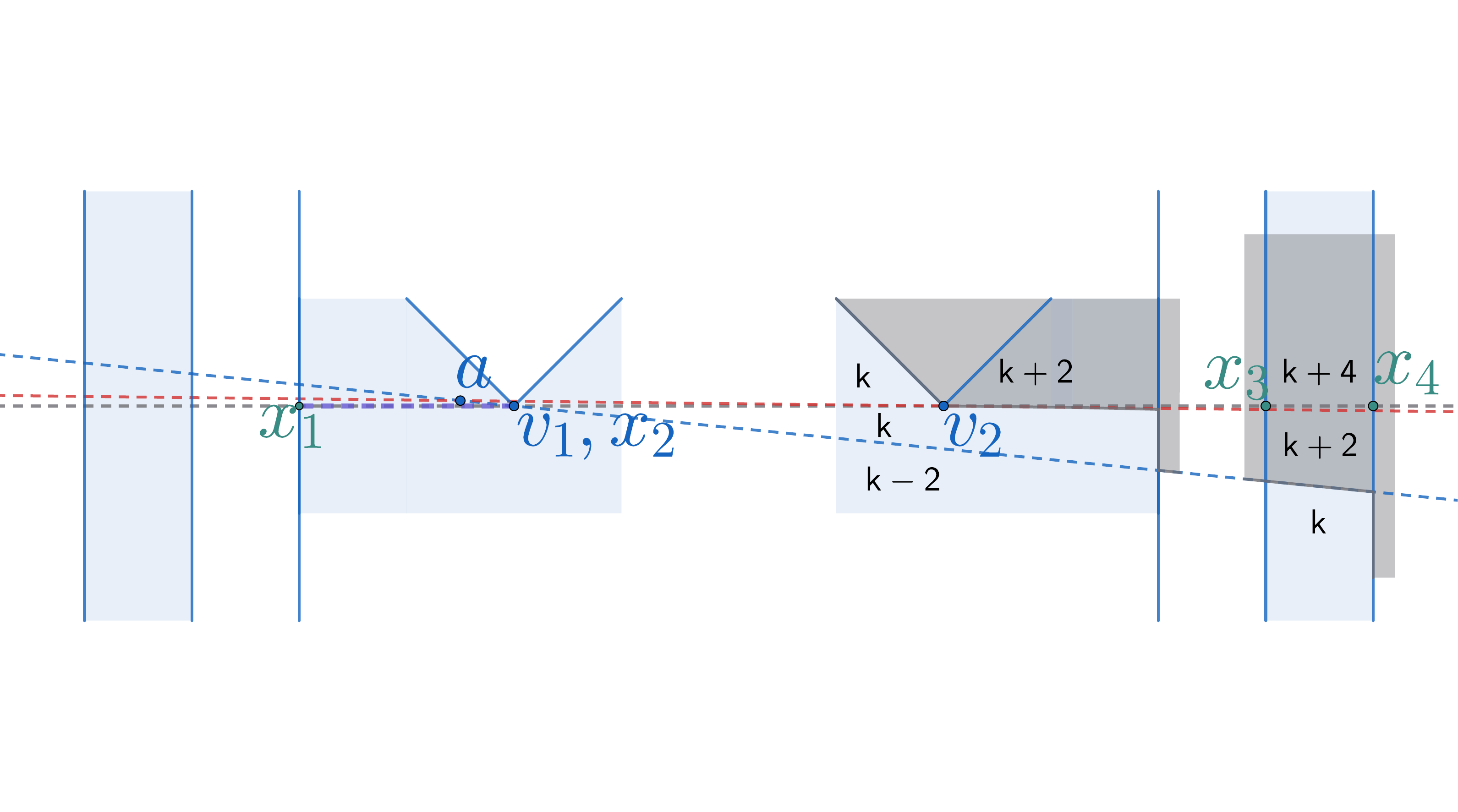}
\caption{Above $l_{g}$}\label{fig:RRS-genericA2}
\end{subfigure}
\begin{subfigure}[b]{.49\linewidth}
\includegraphics[width=\linewidth]{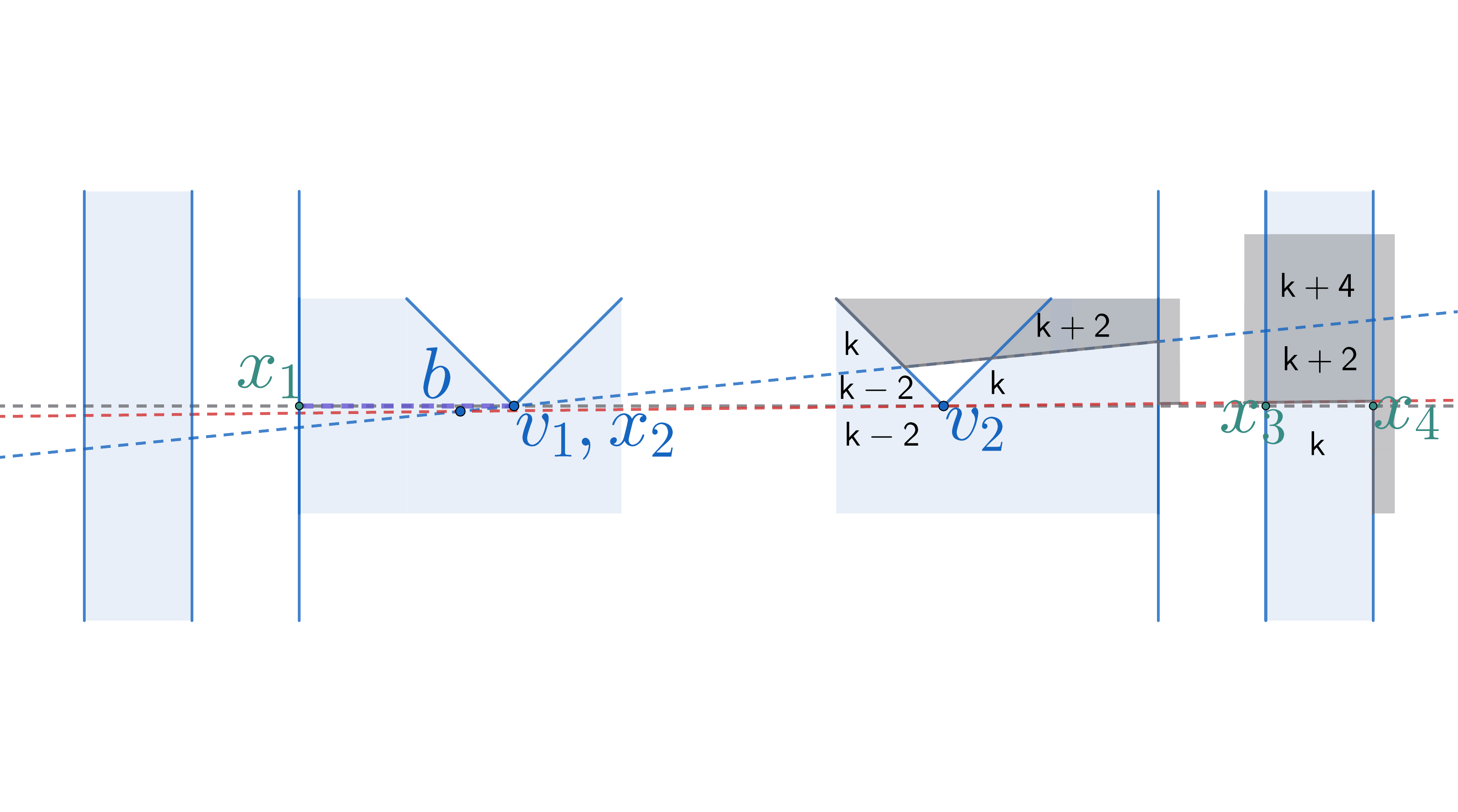}
\caption{Below $l_{g}$}\label{fig:RRS-genericB2}
\end{subfigure}

\caption{RRS; $Z = k - 2$, $W = k$}
\label{fig:RRS-generic-2}
\end{figure}

 \begin{figure}[H]
\centering
\begin{subfigure}[b]{.49\linewidth}
\includegraphics[width=\linewidth]{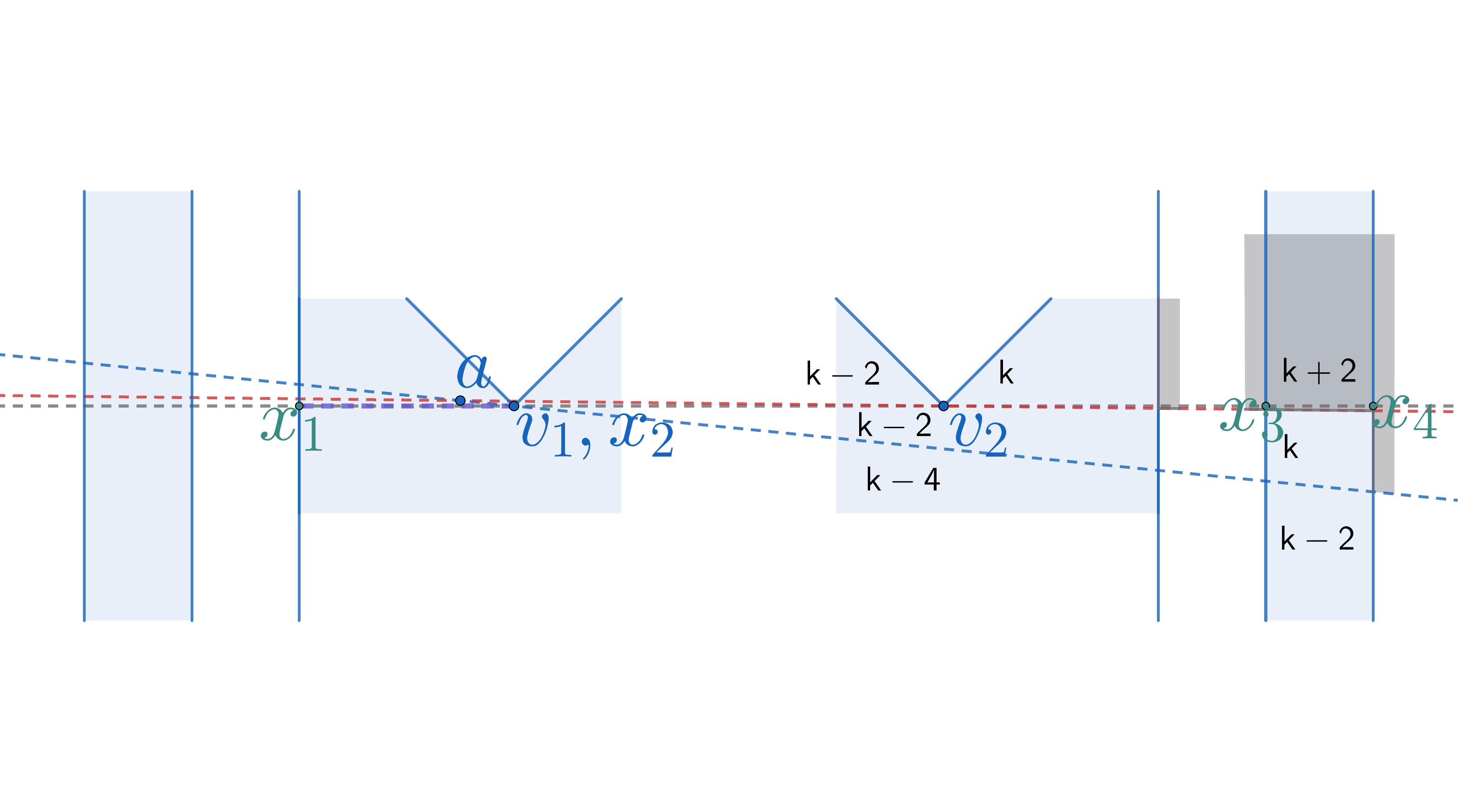}
\caption{Above $l_{g}$}\label{fig:RRS-genericA4}
\end{subfigure}
\begin{subfigure}[b]{.49\linewidth}
\includegraphics[width=\linewidth]{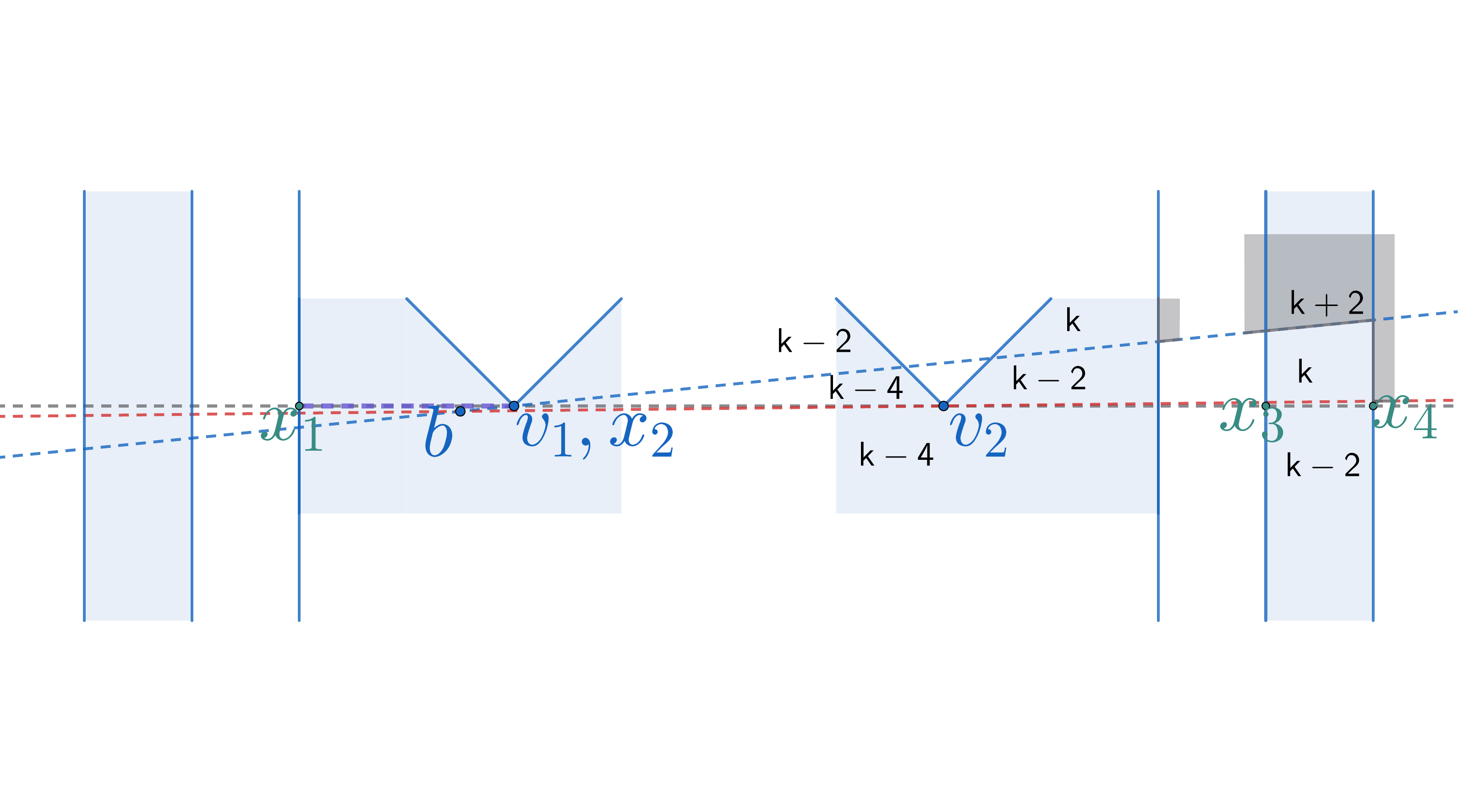}
\caption{Below $l_{g}$}\label{fig:RRS-genericB4}
\end{subfigure}

\caption{RRS; $Z = k - 4$, $W = k - 2$}
\label{fig:RRS-generic-4}
\end{figure}

 \begin{figure}[H]
\centering
\begin{subfigure}[b]{.49\linewidth}
\includegraphics[width=\linewidth]{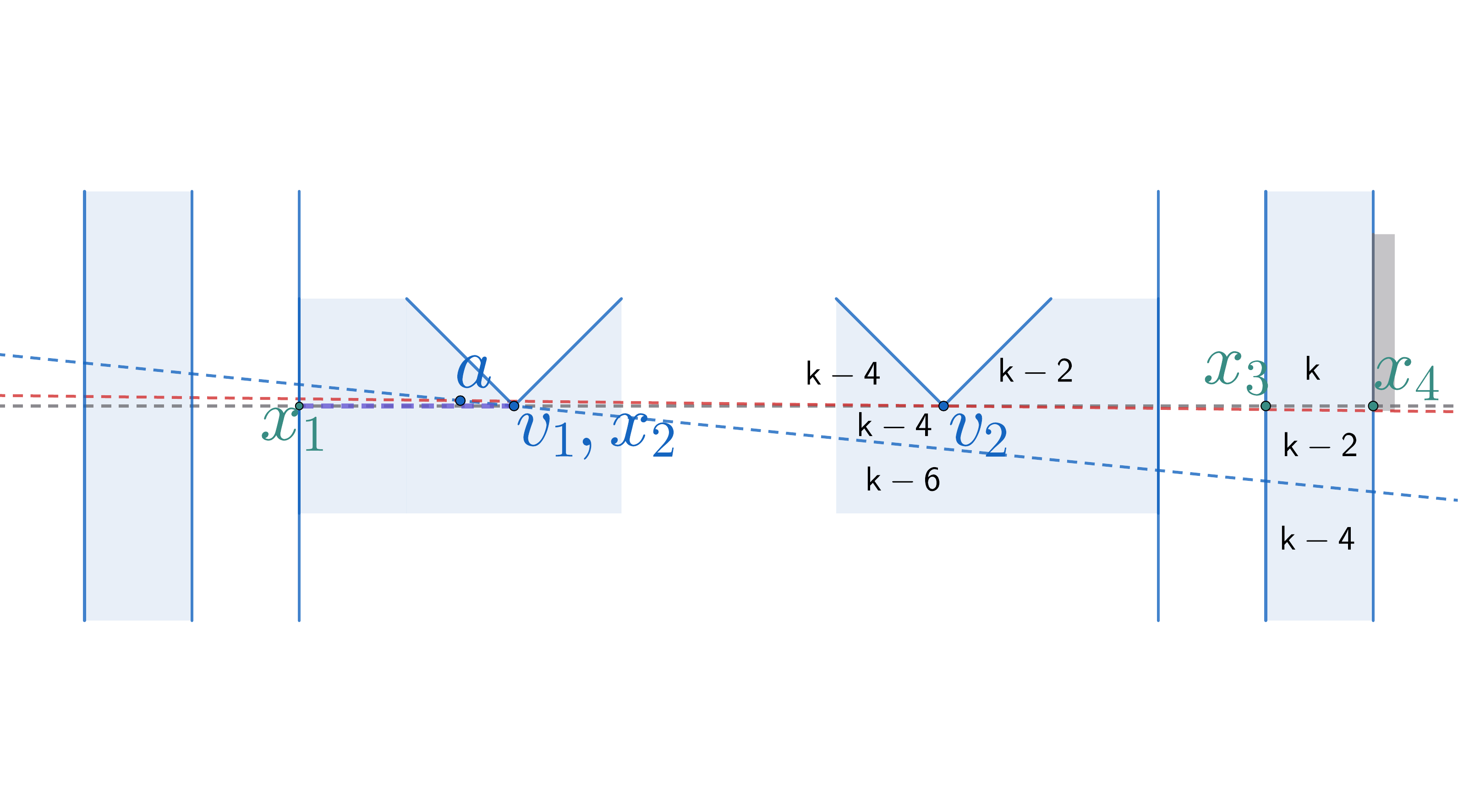}
\caption{Above $l_{g}$}\label{fig:RRS-genericA6}
\end{subfigure}
\begin{subfigure}[b]{.49\linewidth}
\includegraphics[width=\linewidth]{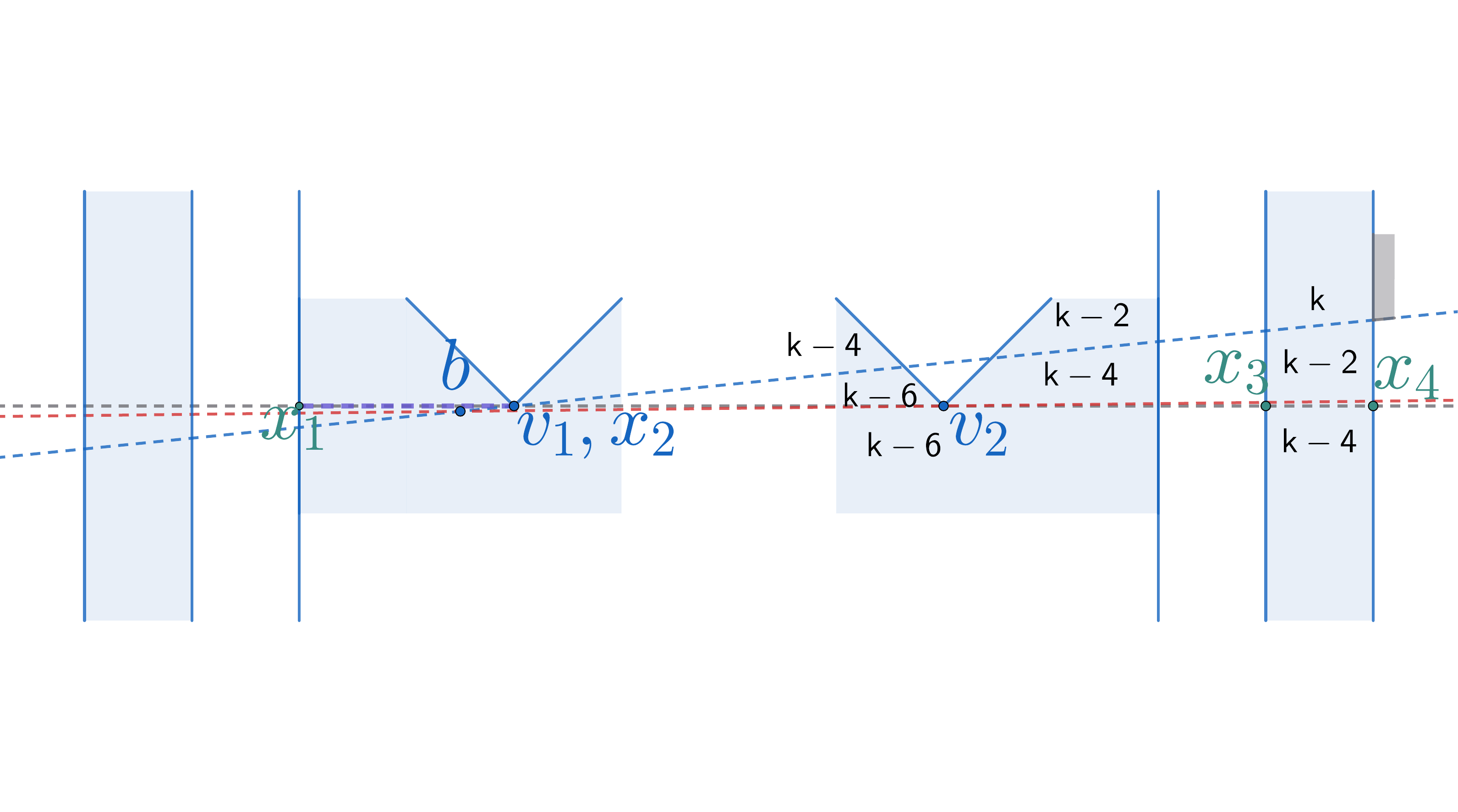}
\caption{Below $l_{g}$}\label{fig:RRS-genericB6}
\end{subfigure}

\caption{RRS; $Z = k - 6$, $W = k - 4$}
\label{fig:RRS-generic-6}
\end{figure}

Note: For $Z \leq k - 8$, $v_{2}$ and its surroundings are entirely visible. For $W \leq k - 6$, $x_{3}x_{4}$ and its surroundings are entirely visible.

\end{proof}
\section{RRO}

\begin{lemma}
\label{lemma:RRO}
A partition line at $x_1x_2$ is needed in the following cases: 
\renewcommand{\labelitemi}{$\bullet$}
\begin{itemize}
    \item $Z = k$ (Figure~\ref{fig:RRO-generic-0})
    \item $Z = k - 2$ (Figure~\ref{fig:RRO-generic-2})
    \item $W = k$ (Figure~\ref{fig:RRO-generic-2})
    \item $Z = k - 4$ (Figure~\ref{fig:RRO-generic-4})
    \item $W = k-2$ (Figure~\ref{fig:RRO-generic-4})
    \item $W=k-4$ (Figure~\ref{fig:RRO-generic-6})
\end{itemize}
\end{lemma}
\begin{proof}
See Figures~\ref{fig:RRO-generic-0}-~\ref{fig:RRO-generic-8}.

Note: for $Z \geq k + 2$, all of $v_{2}$ would be in shadow. For $W \geq k + 4$, $x_{3}x_{4}$ and its surroundings is entirely in shadow.

 \begin{figure}[H]
\centering
\begin{subfigure}[b]{.49\linewidth}
\includegraphics[width=\linewidth]{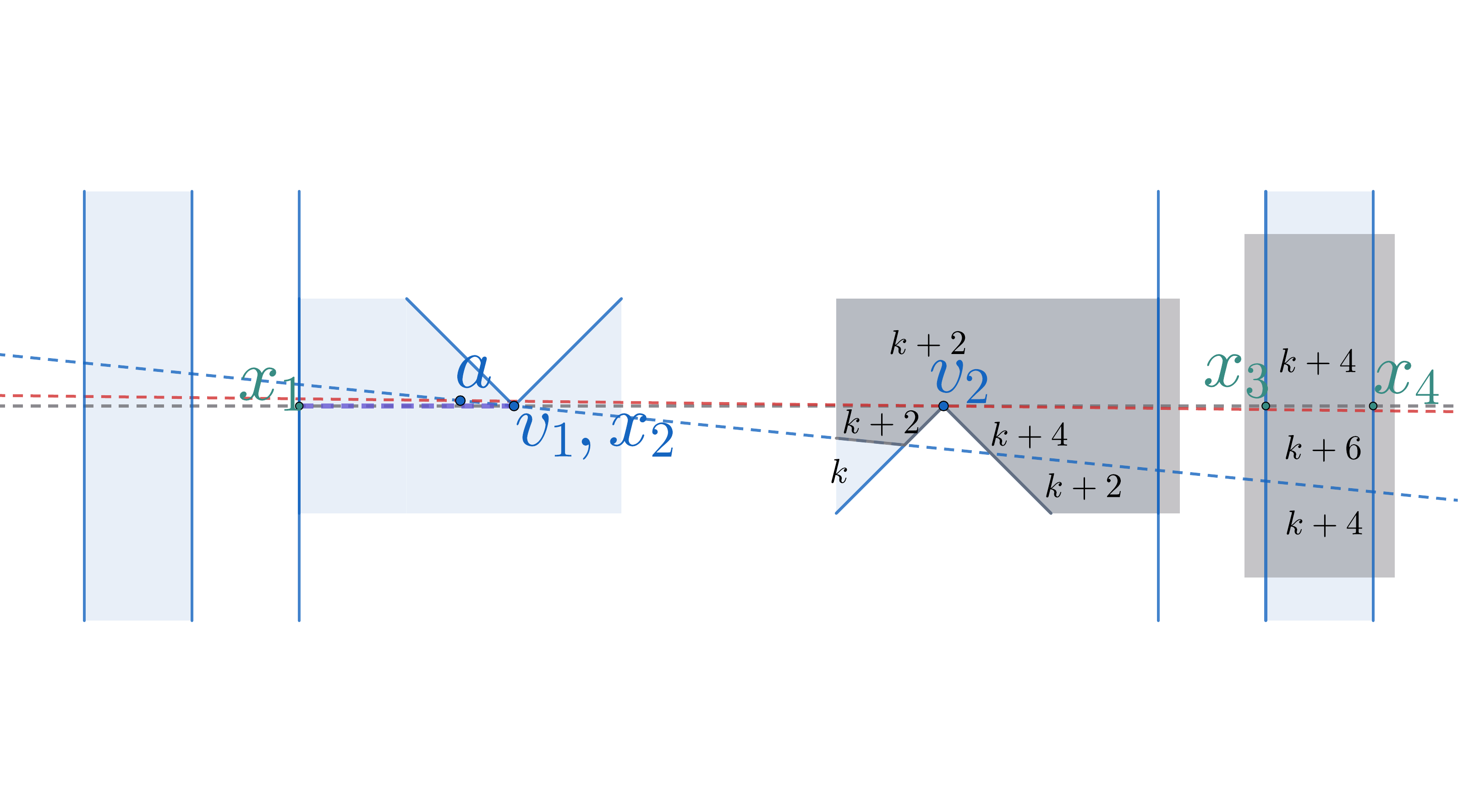}
\caption{Above $l_{g}$}\label{fig:RRO-genericA0}
\end{subfigure}
\begin{subfigure}[b]{.49\linewidth}
\includegraphics[width=\linewidth]{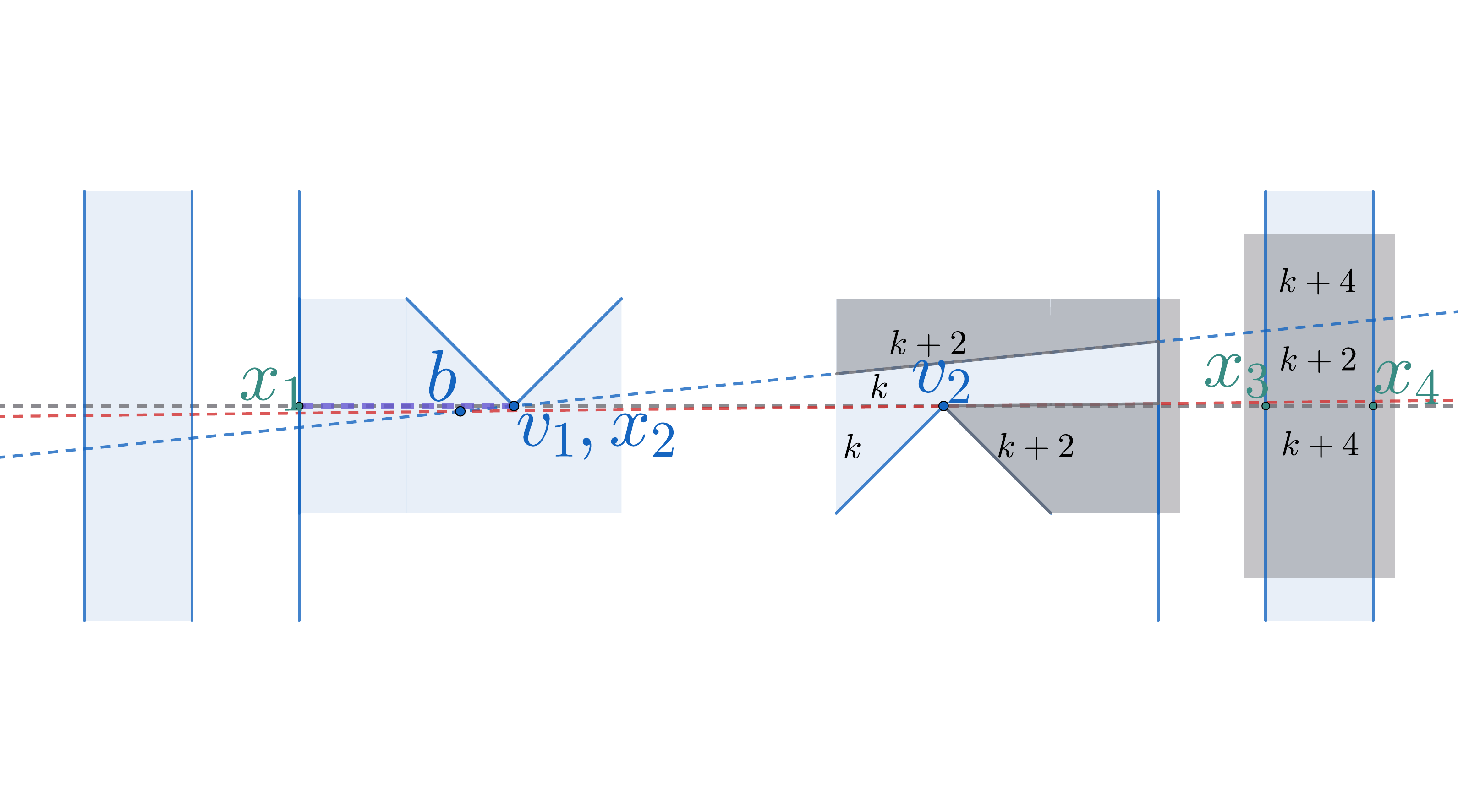}
\caption{Below $l_{g}$}\label{fig:RRO-genericB0}
\end{subfigure}

\caption{RRO; $Z = k$, $W = k + 2$}
\label{fig:RRO-generic-0}
\end{figure}

 \begin{figure}[H]
\centering
\begin{subfigure}[b]{.49\linewidth}
\includegraphics[width=\linewidth]{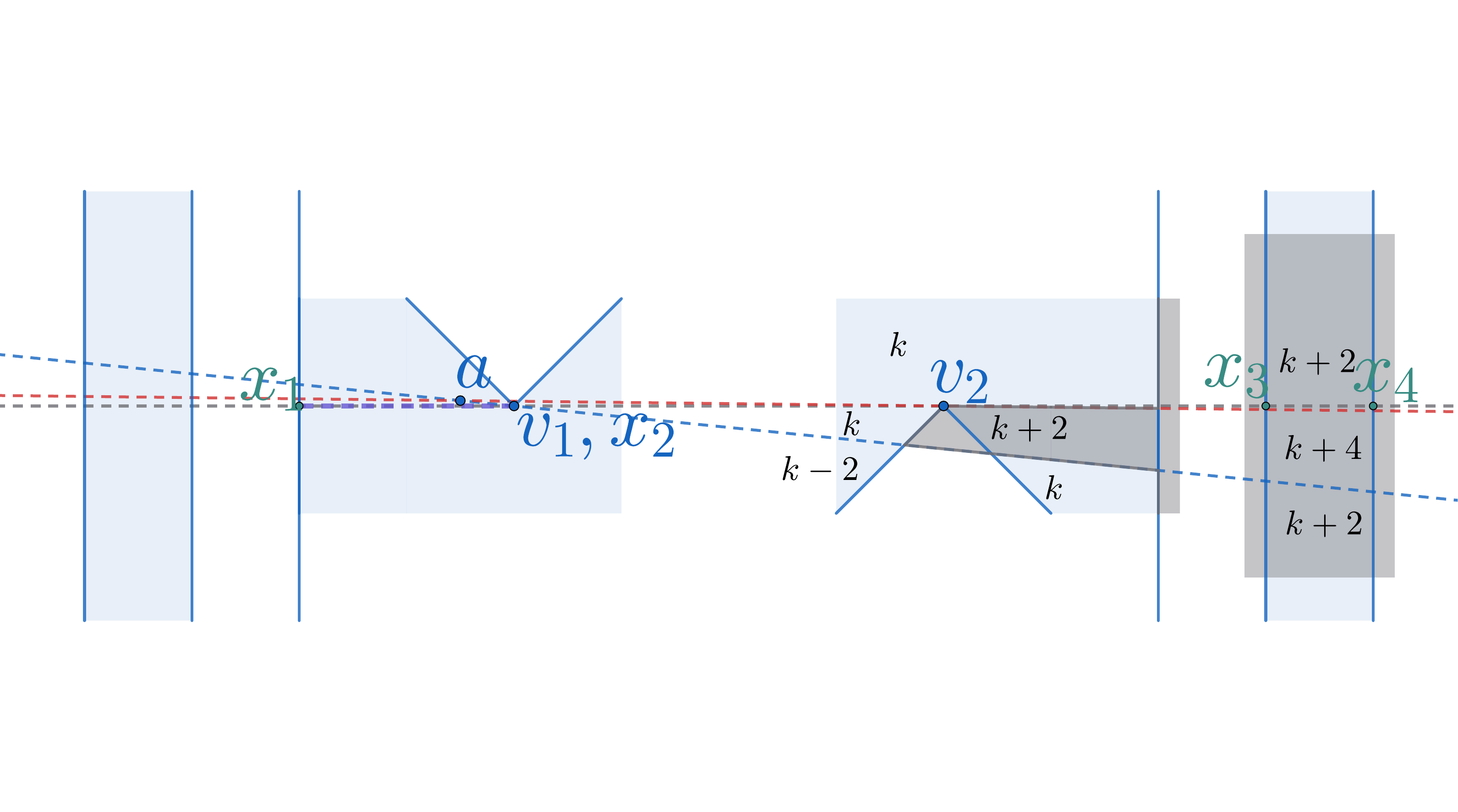}
\caption{Above $l_{g}$}\label{fig:RRO-genericA2}
\end{subfigure}
\begin{subfigure}[b]{.49\linewidth}
\includegraphics[width=\linewidth]{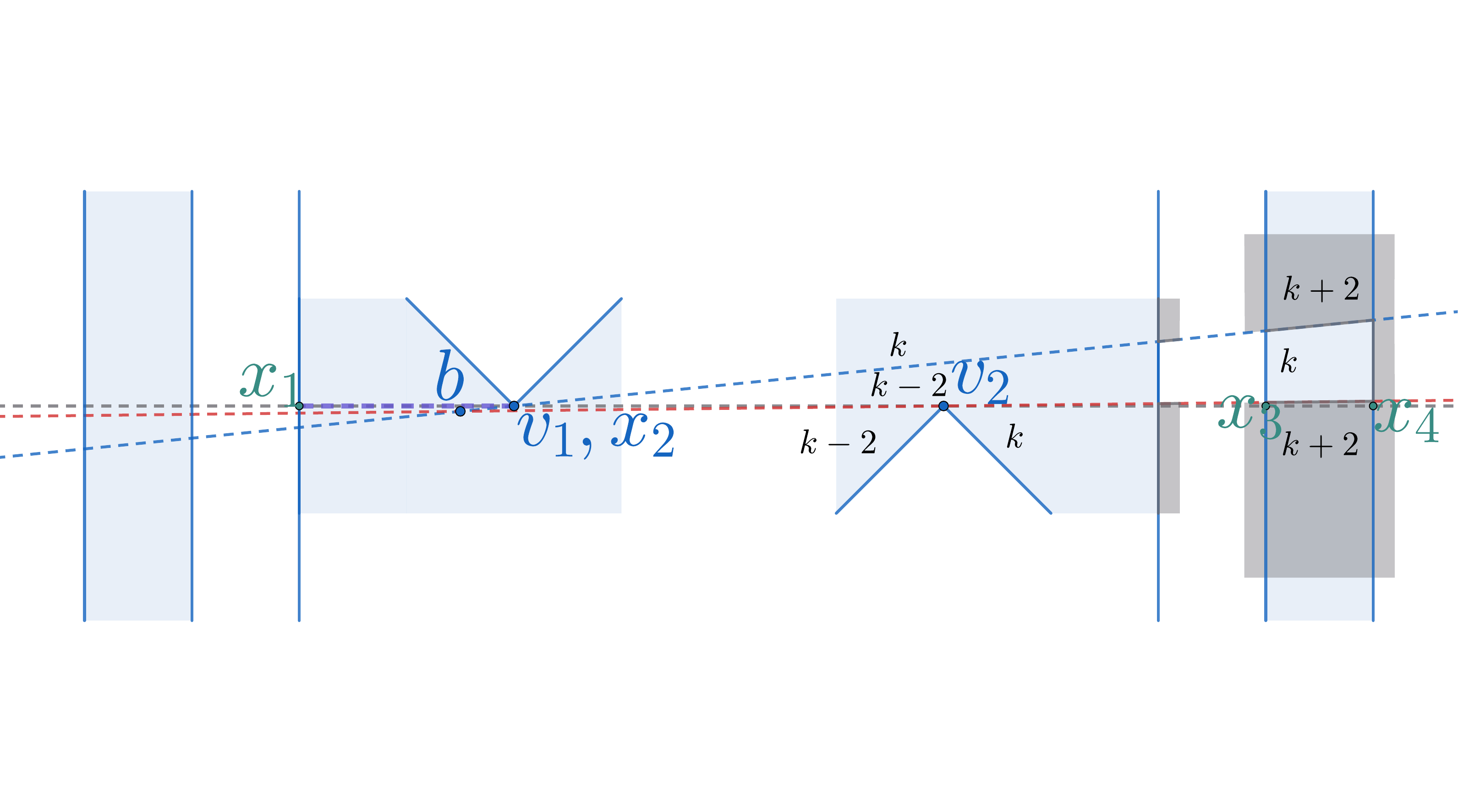}
\caption{Below $l_{g}$}\label{fig:RRO-genericB2}
\end{subfigure}

\caption{RRO; $Z = k - 2$, $W = k$}
\label{fig:RRO-generic-2}
\end{figure}

 \begin{figure}[H]
\centering
\begin{subfigure}[b]{.49\linewidth}
\includegraphics[width=\linewidth]{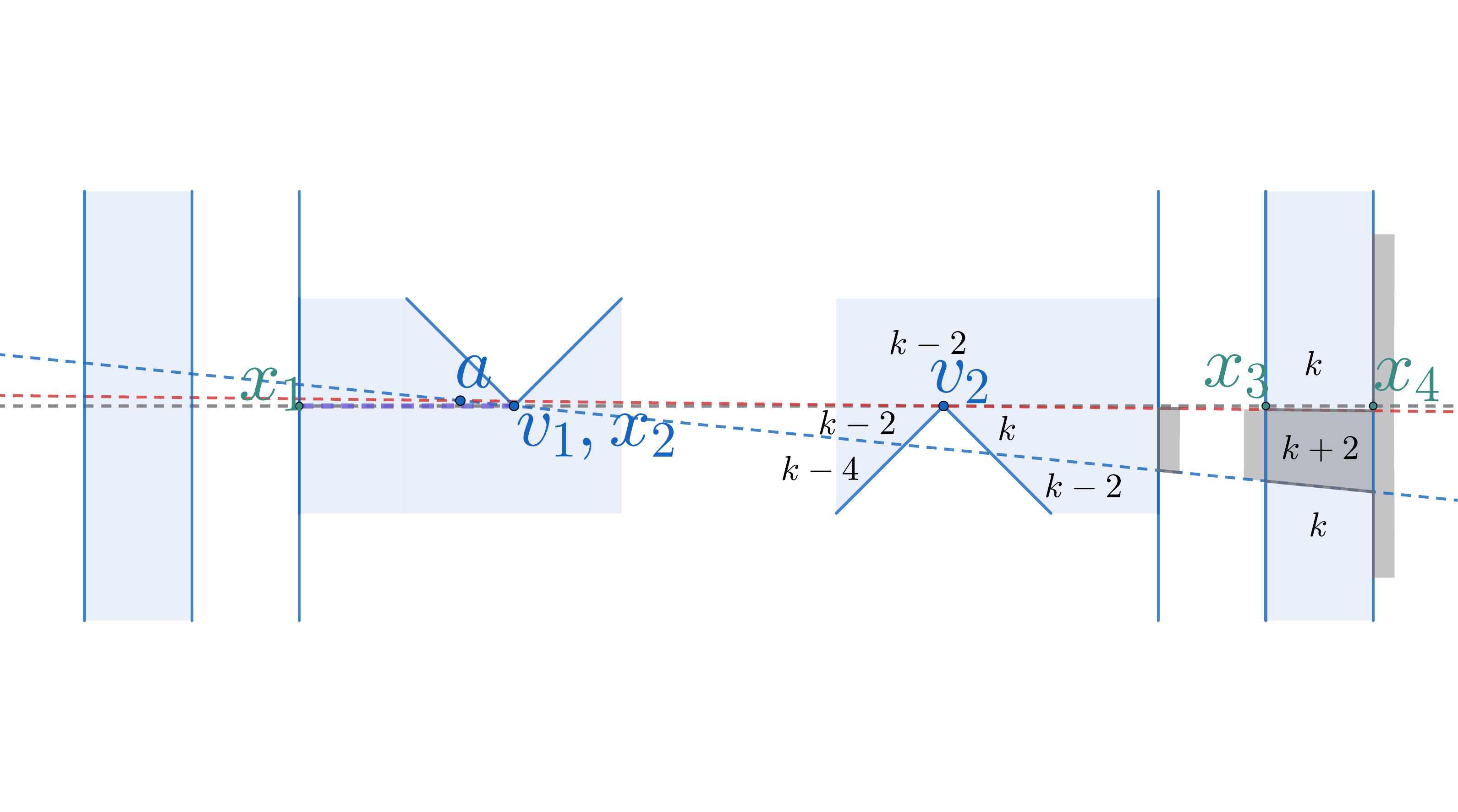}
\caption{Above $l_{g}$}\label{fig:RRO-genericA4}
\end{subfigure}
\begin{subfigure}[b]{.49\linewidth}
\includegraphics[width=\linewidth]{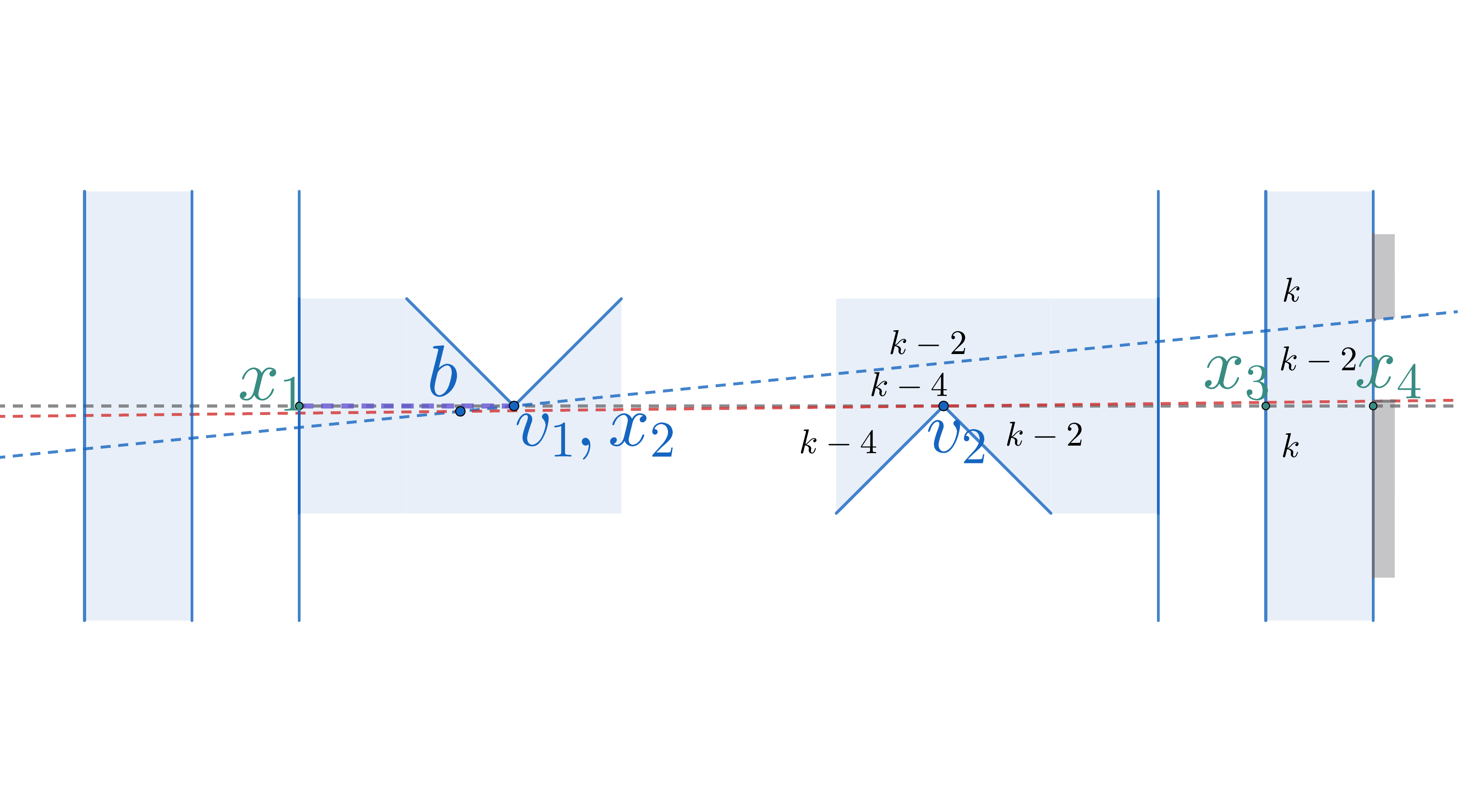}
\caption{Below $l_{g}$}\label{fig:RRO-genericB4}
\end{subfigure}

\caption{RRO; $Z = k - 4$, $W = k - 2$}
\label{fig:RRO-generic-4}
\end{figure}

 \begin{figure}[H]
\centering
\begin{subfigure}[b]{.49\linewidth}
\includegraphics[width=\linewidth]{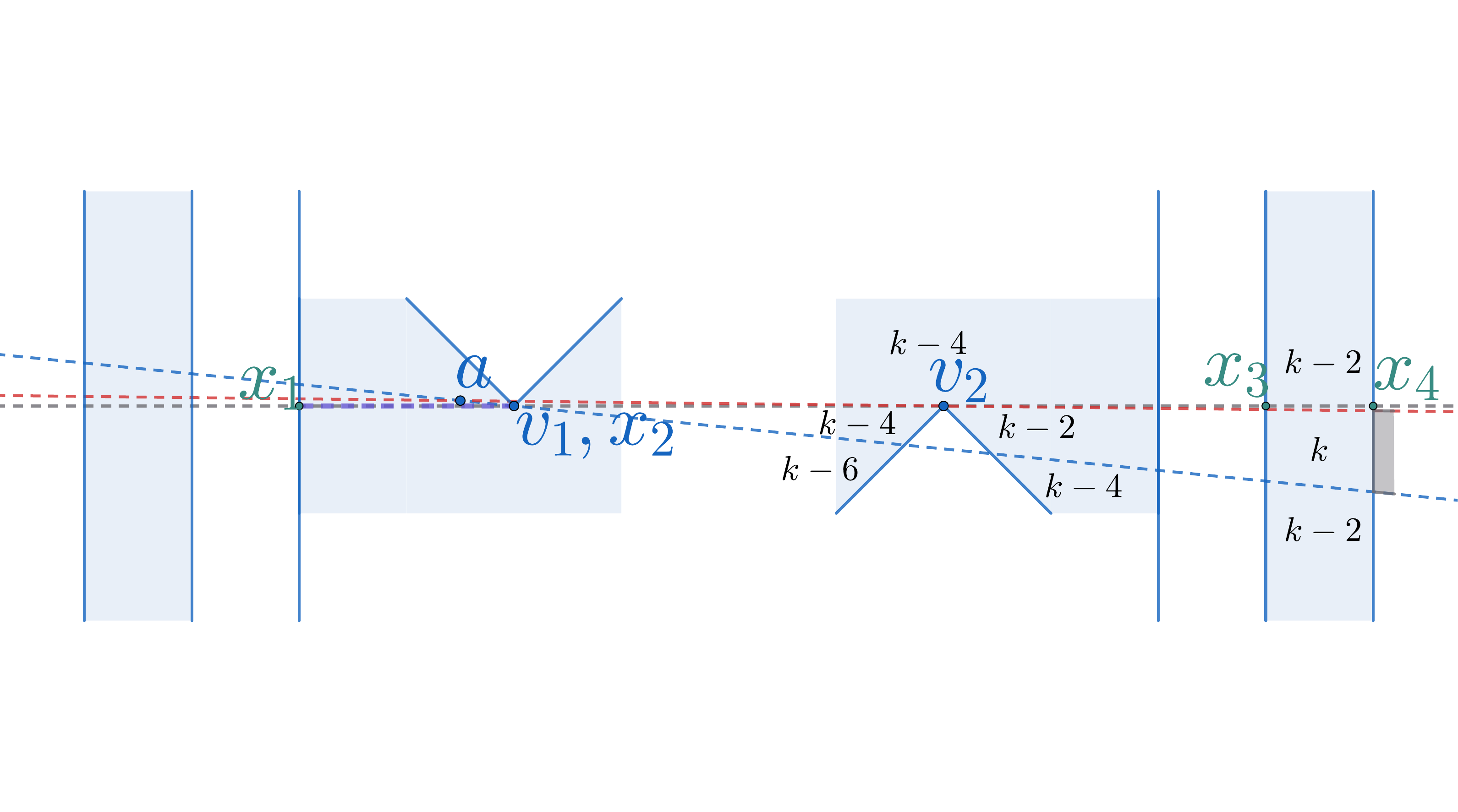}
\caption{Above $l_{g}$}\label{fig:RRO-genericA6}
\end{subfigure}
\begin{subfigure}[b]{.49\linewidth}
\includegraphics[width=\linewidth]{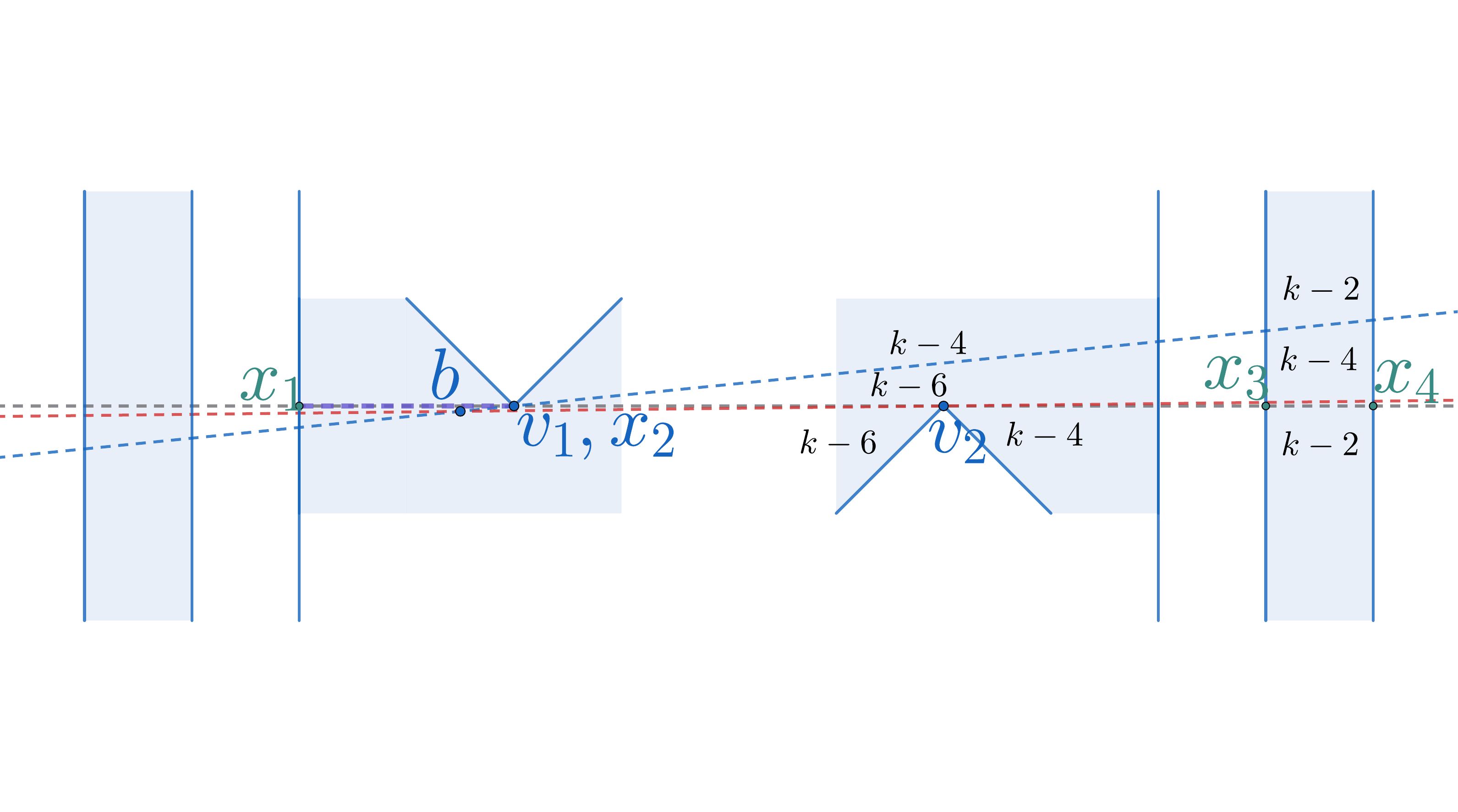}
\caption{Below $l_{g}$}\label{fig:RRO-genericB6}
\end{subfigure}

\caption{RRO; $Z = k - 6$, $W = k - 4$}
\label{fig:RRO-generic-6}
\end{figure}

 \begin{figure}[H]
\centering
\begin{subfigure}[b]{.49\linewidth}
\includegraphics[width=\linewidth]{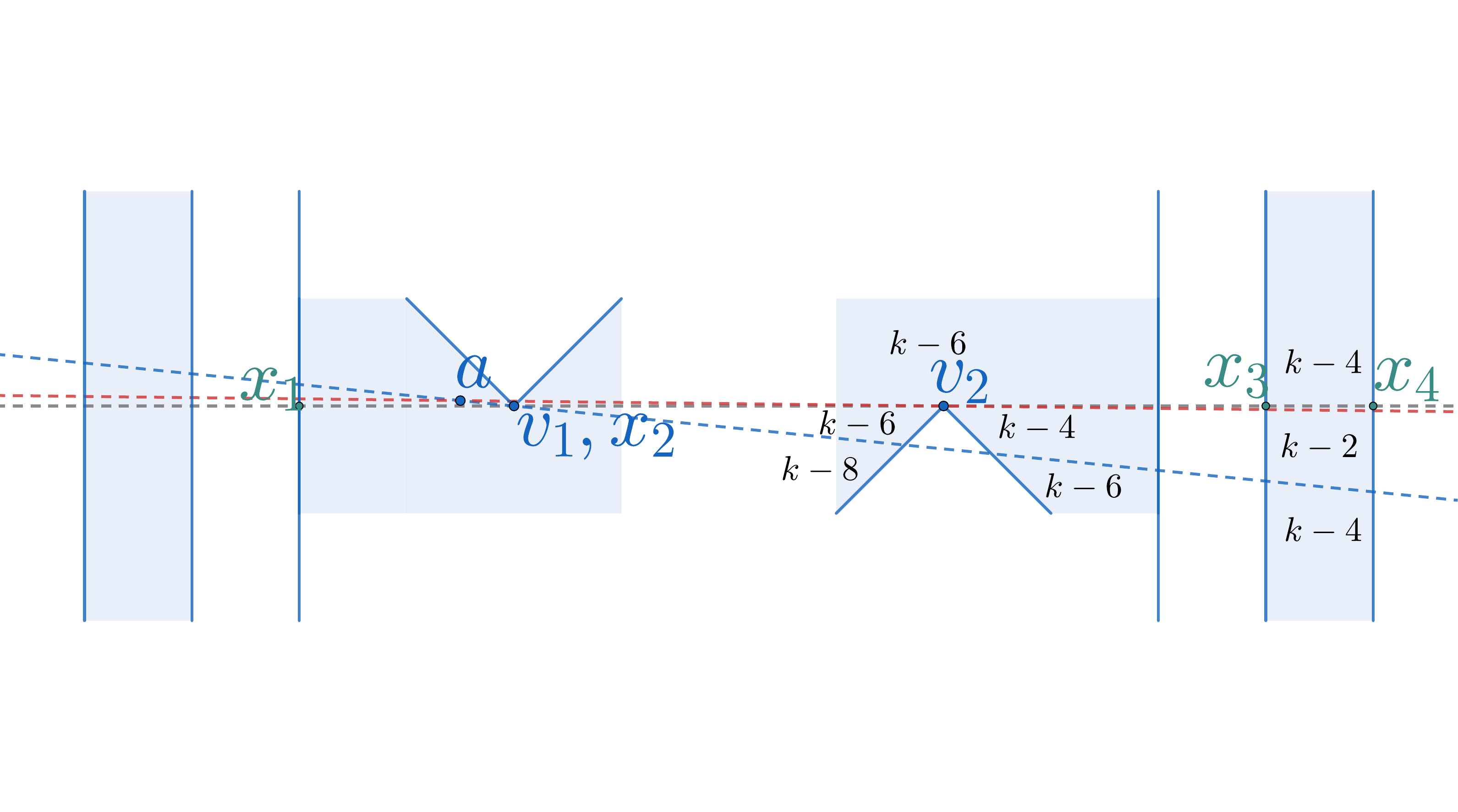}
\caption{Above $l_{g}$}\label{fig:RRO-genericA8}
\end{subfigure}
\begin{subfigure}[b]{.49\linewidth}
\includegraphics[width=\linewidth]{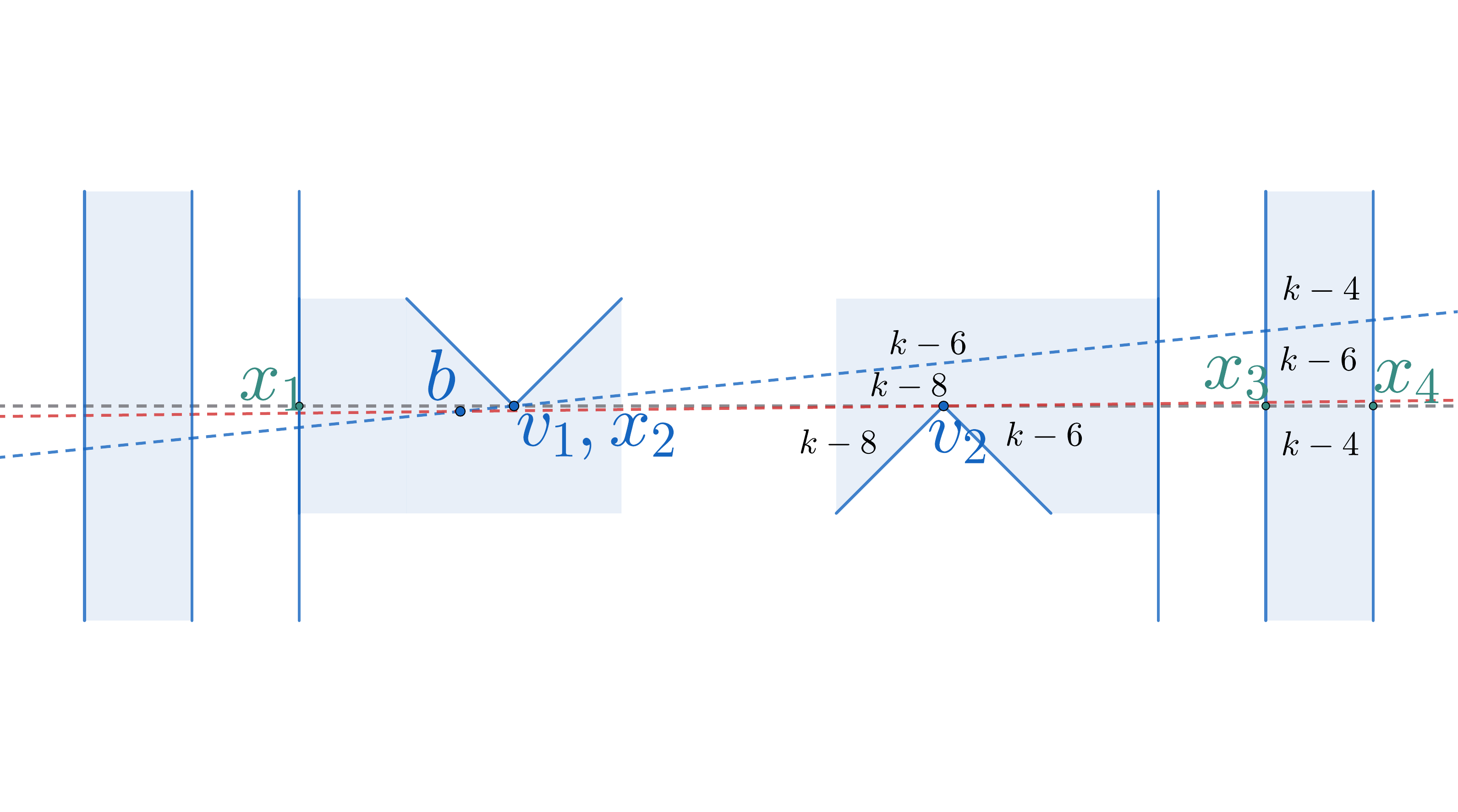}
\caption{Below $l_{g}$}\label{fig:RRO-genericB8}
\end{subfigure}

\caption{RRO; $Z = k - 8$, $W = k - 6$}
\label{fig:RRO-generic-8}
\end{figure}

for $Z \leq k - 10$,  $v_{2}$ and its surroundings is entirely visible. For $W \leq k - 8$, $x_{3}x_{4}$ and its surroundings is entirely visible. 

\end{proof}
\section{RC-SC}

\begin{lemma}
\label{lemma:RC-SC}
A partition line is needed in the following cases: 
\renewcommand{\labelitemi}{$\bullet$}
\begin{itemize}
    \item $Z = k$ (Figure~\ref{fig:RC-SpecialCase-0})
    \item $Z = k - 2$ (Figure~\ref{fig:RC-SpecialCase-2})
    \item $W = k - 1$ (Figure~\ref{fig:RC-SpecialCase-2})
    \item $W = k - 3$ (Figure~\ref{fig:RC-SpecialCase-4})
 
\end{itemize}
\end{lemma}
\begin{proof}
See Figures~\ref{fig:RC-SpecialCase-0}-~\ref{fig:RC-SpecialCase-4}.

Note: for $Z \geq k + 2$, all of $v_{2}$ would be in shadow. For $W \geq k + 3$, $x_{3}x_{4}$ and its surroundings is entirely in shadow as in Figure~\ref{fig:RC-SpecialCase-0}. for $Z \leq k - 6$, $W \leq k - 5$, $v_{2}$ and $x_{3}x_{4}$ are entirely visible as Figure~\ref{fig:RC-SpecialCase-4}.

 \begin{figure}[H]
\centering
\begin{subfigure}[b]{.49\linewidth}
\includegraphics[width=\linewidth]{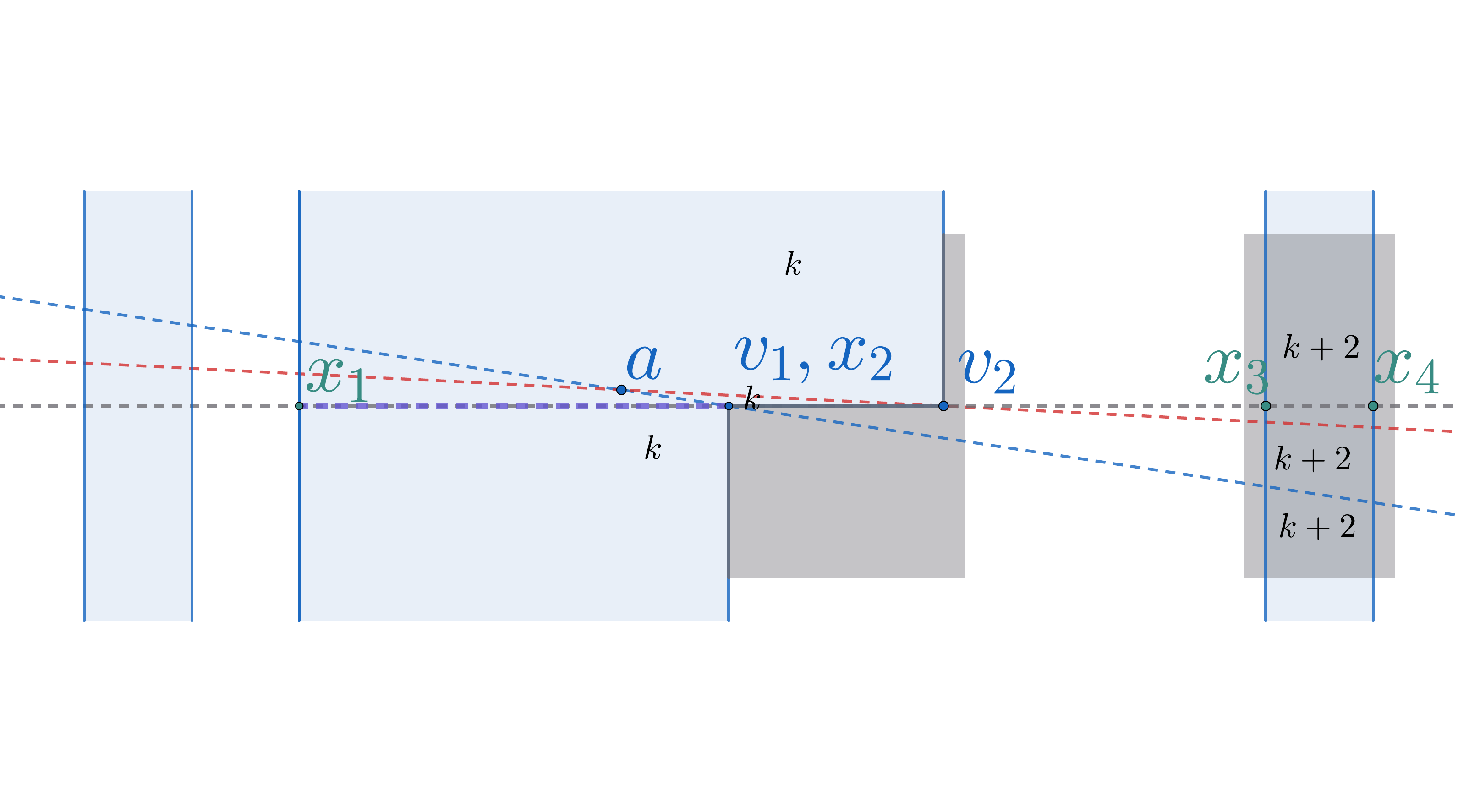}
\caption{Above $l_{g}$}\label{fig:RC-SpecialCase-A0}
\end{subfigure}
\begin{subfigure}[b]{.49\linewidth}
\includegraphics[width=\linewidth]{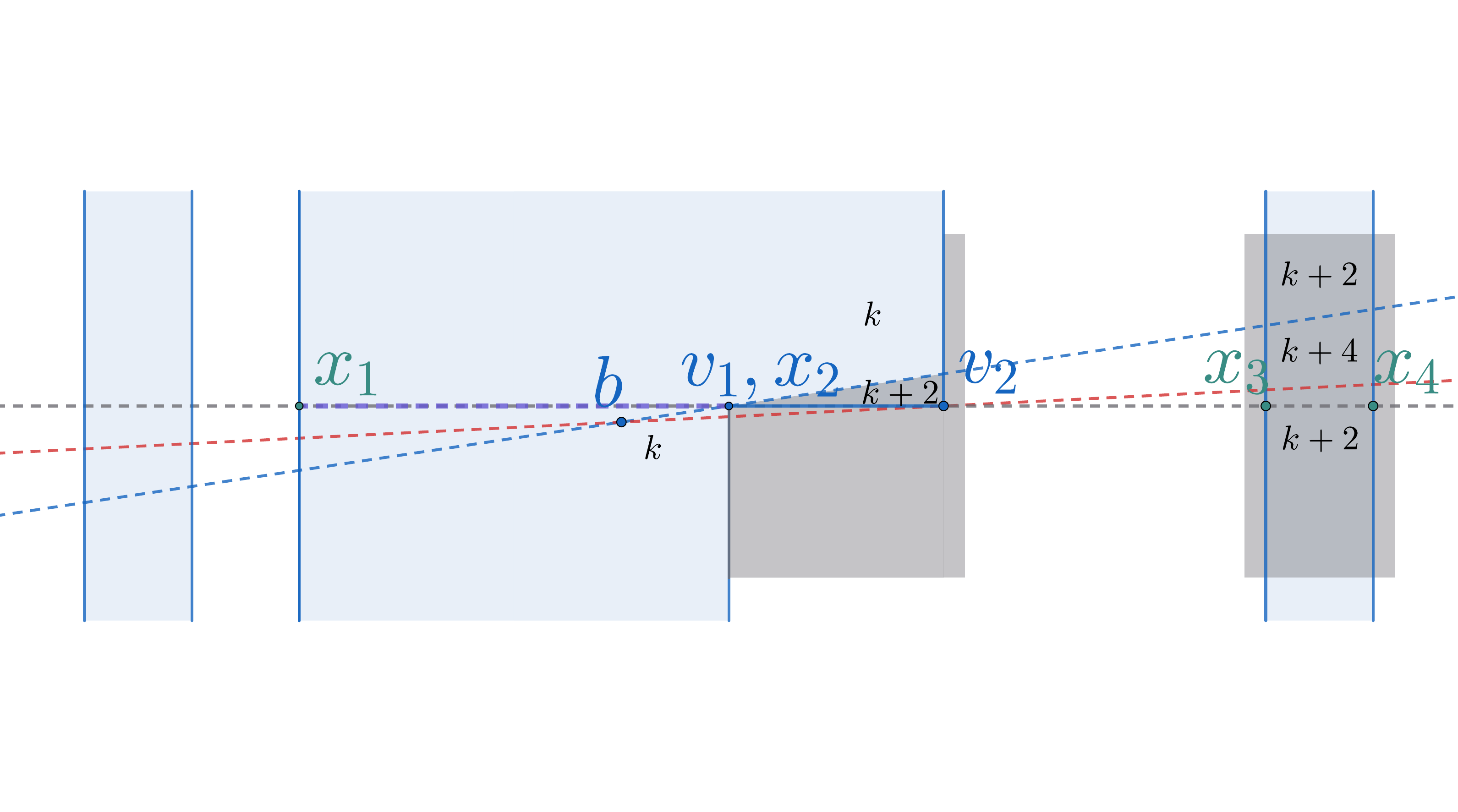}
\caption{Below $l_{g}$}\label{fig:RC-SpecialCase-B0}
\end{subfigure}

\caption{RC-SC; $Z = k$, $W = k + 1$}
\label{fig:RC-SpecialCase-0}
\end{figure}

 \begin{figure}[H]
\centering
\begin{subfigure}[b]{.49\linewidth}
\includegraphics[width=\linewidth]{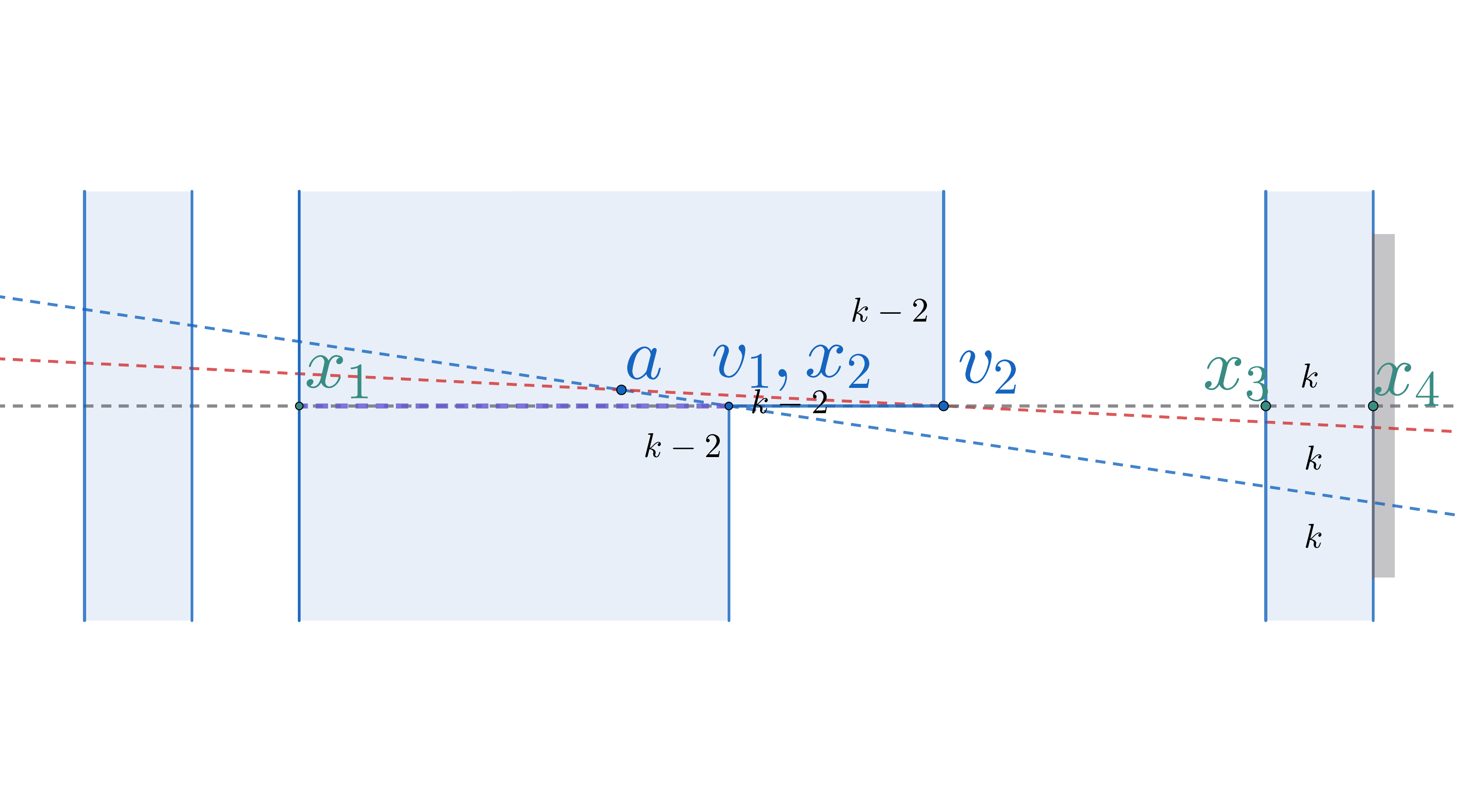}
\caption{Above $l_{g}$}\label{fig:RC-SpecialCase-A2}
\end{subfigure}
\begin{subfigure}[b]{.49\linewidth}
\includegraphics[width=\linewidth]{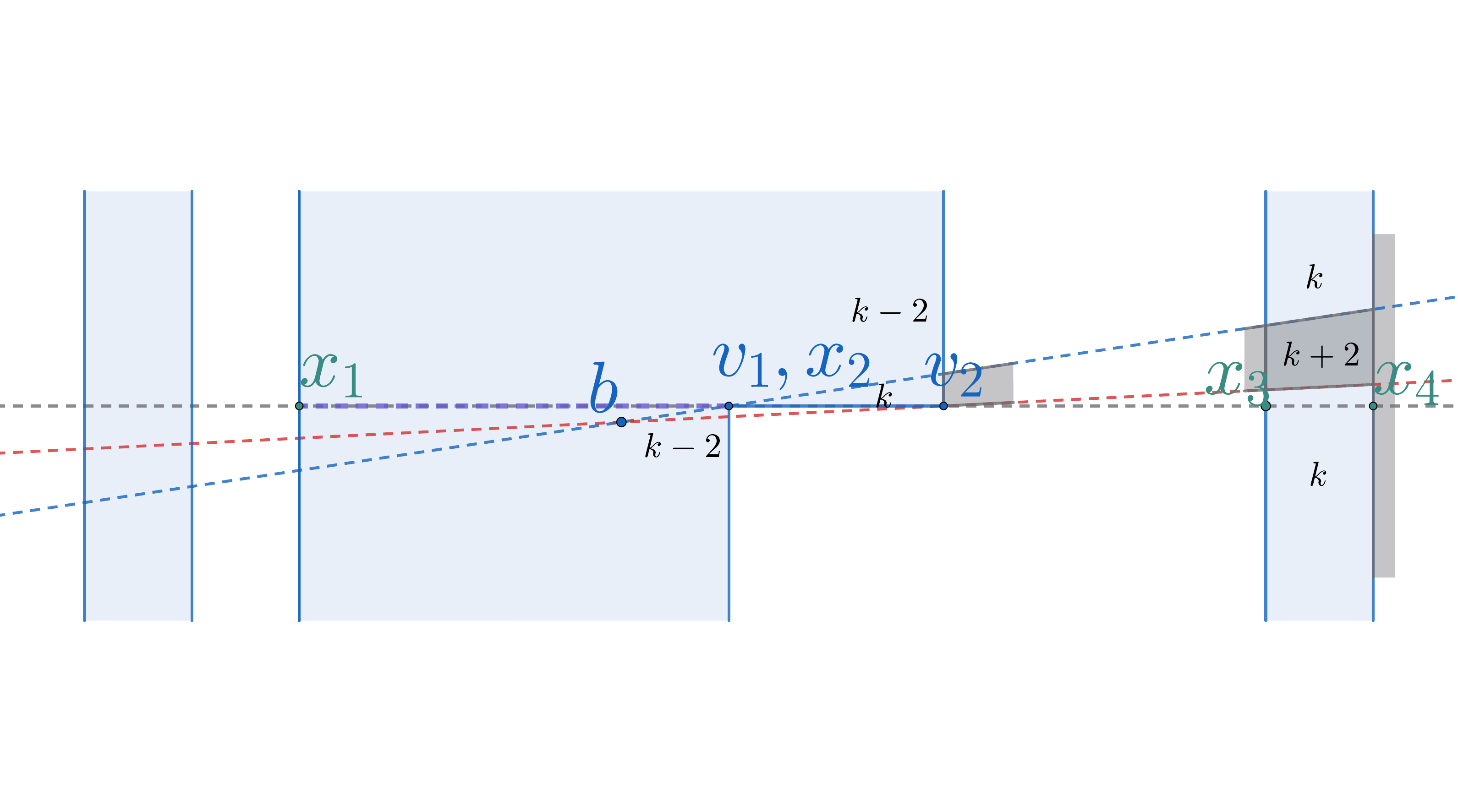}
\caption{Below $l_{g}$}\label{fig:RC-SpecialCase-B2}
\end{subfigure}

\caption{RC-SC; $Z = k - 2$, $W = k - 1$}
\label{fig:RC-SpecialCase-2}
\end{figure}

  \begin{figure}[H]
\centering
\begin{subfigure}[b]{.49\linewidth}
\includegraphics[width=\linewidth]{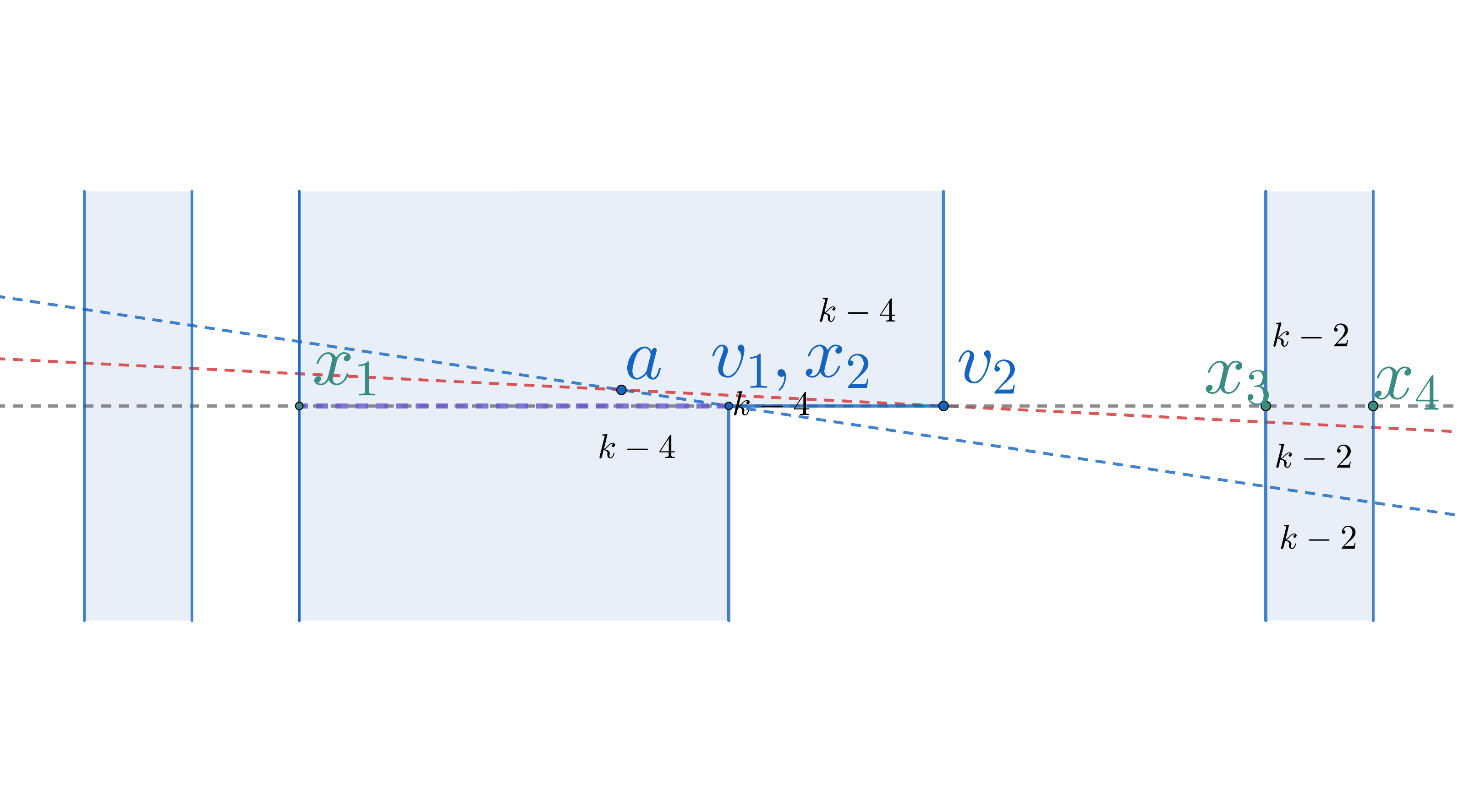}
\caption{Above $l_{g}$}\label{fig:RC-SpecialCase-A4}
\end{subfigure}
\begin{subfigure}[b]{.49\linewidth}
\includegraphics[width=\linewidth]{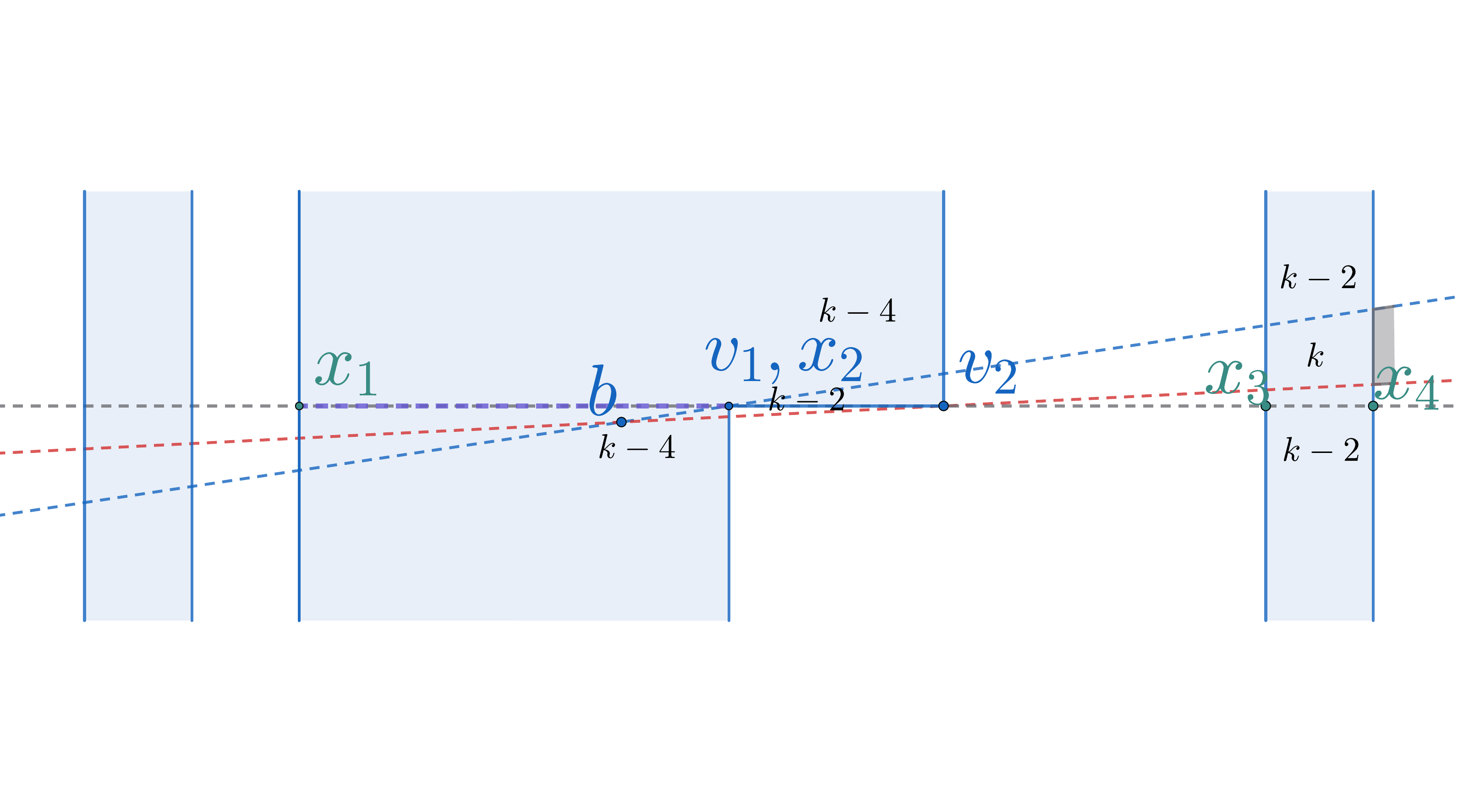}
\caption{Below $l_{g}$}\label{fig:RC-SpecialCase-B4}
\end{subfigure}

\caption{RC-SC; $Z = k - 4$, $W = k - 3$}
\label{fig:RC-SpecialCase-4}
\end{figure}

\end{proof}

\section{RR-SC}

\begin{lemma}
\label{lemma:RR-SC}

A partition line is needed in the following cases: 
\renewcommand{\labelitemi}{$\bullet$}
\label{appendix:RR-SC}
\begin{itemize}
    \item $Z = k$  (Figure~\ref{fig:RR-SpecialCase-0})
    \item $Z = k - 2$  (Figure~\ref{fig:RR-SpecialCase-2})
    \item $W = k$  (Figure~\ref{fig:RR-SpecialCase-2})
    \item $W = k - 2$  (Figure~\ref{fig:RR-SpecialCase-4})
\end{itemize}

\end{lemma}

\begin{proof}
See Figures~\ref{fig:RR-SpecialCase-0}-~\ref{fig:RR-SpecialCase-4}.

Note: 
For $Z \geq k + 2$, $v_{2}$ and its surroundings are completely invisible, For $W \geq k + 4$, $x_{3}x_{4}$ and its surroundings are completely invisible.

for $Z \leq k - 6$, $v_{2}$ and its surroundings are completely visible. For $W \leq k - 4$,  $x_{3}x_{4}$ and its surroundings are completely visible.

 \begin{figure}[H]
\centering
\begin{subfigure}[b]{.49\linewidth}
\includegraphics[width=\linewidth]{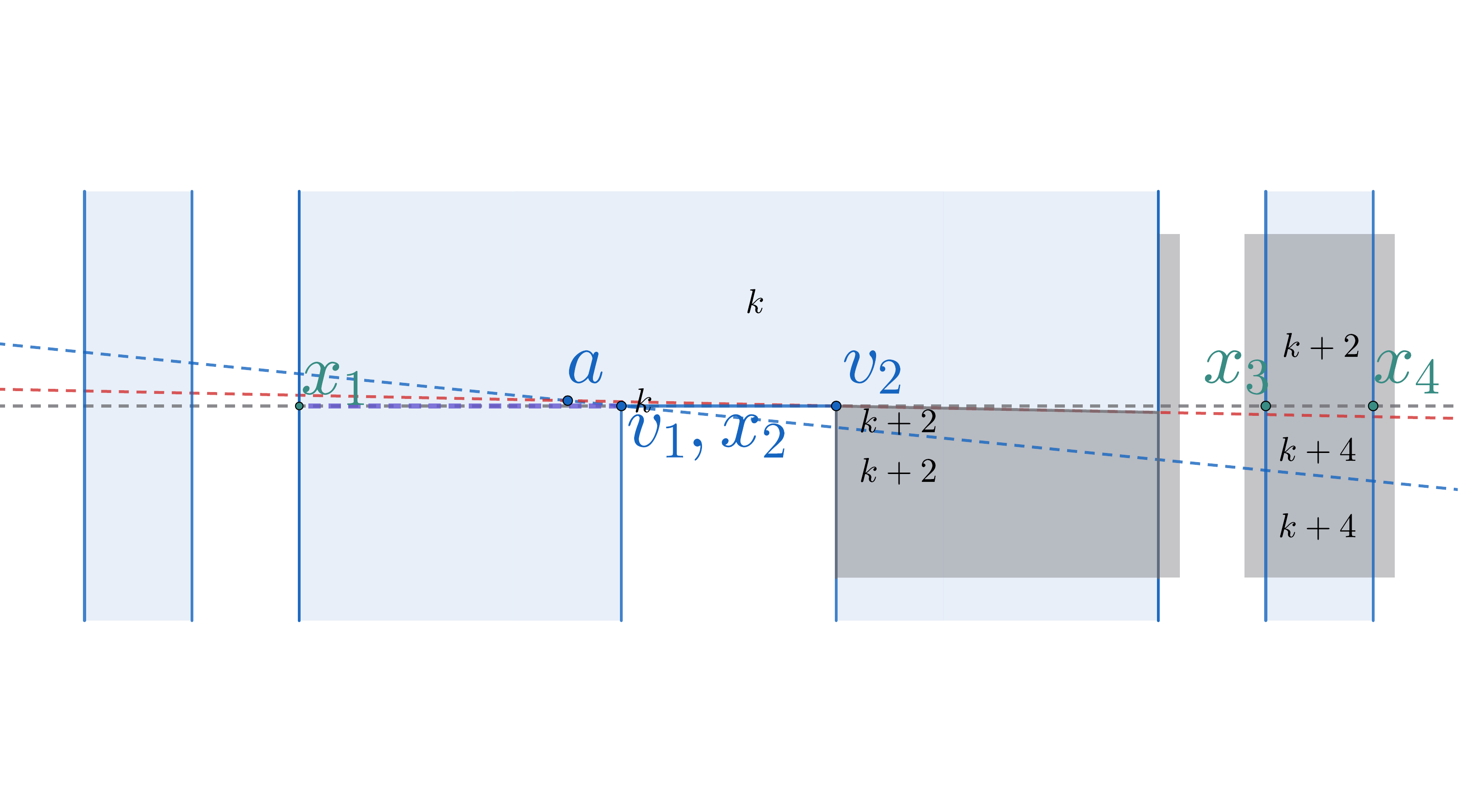}
\caption{Above $l_{g}$}\label{fig:RR-SpecialCase-A0}
\end{subfigure}
\begin{subfigure}[b]{.49\linewidth}
\includegraphics[width=\linewidth]{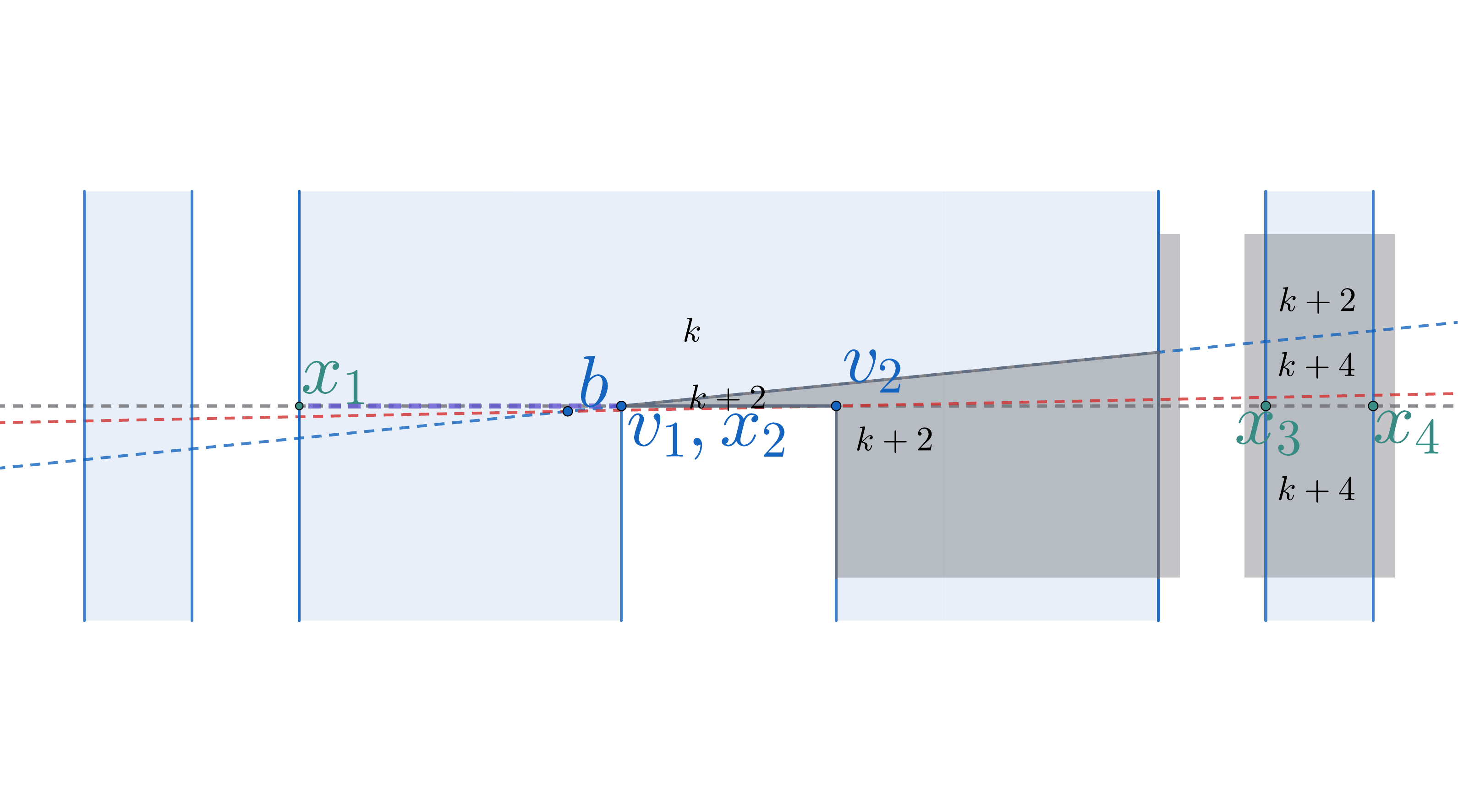}
\caption{Below $l_{g}$}\label{fig:RR-SpecialCase-B0}
\end{subfigure}

\caption{RR-SC; $Z = k$, $W = k + 2$}
\label{fig:RR-SpecialCase-0}
\end{figure}

 \begin{figure}[H]
\centering
\begin{subfigure}[b]{.49\linewidth}
\includegraphics[width=\linewidth]{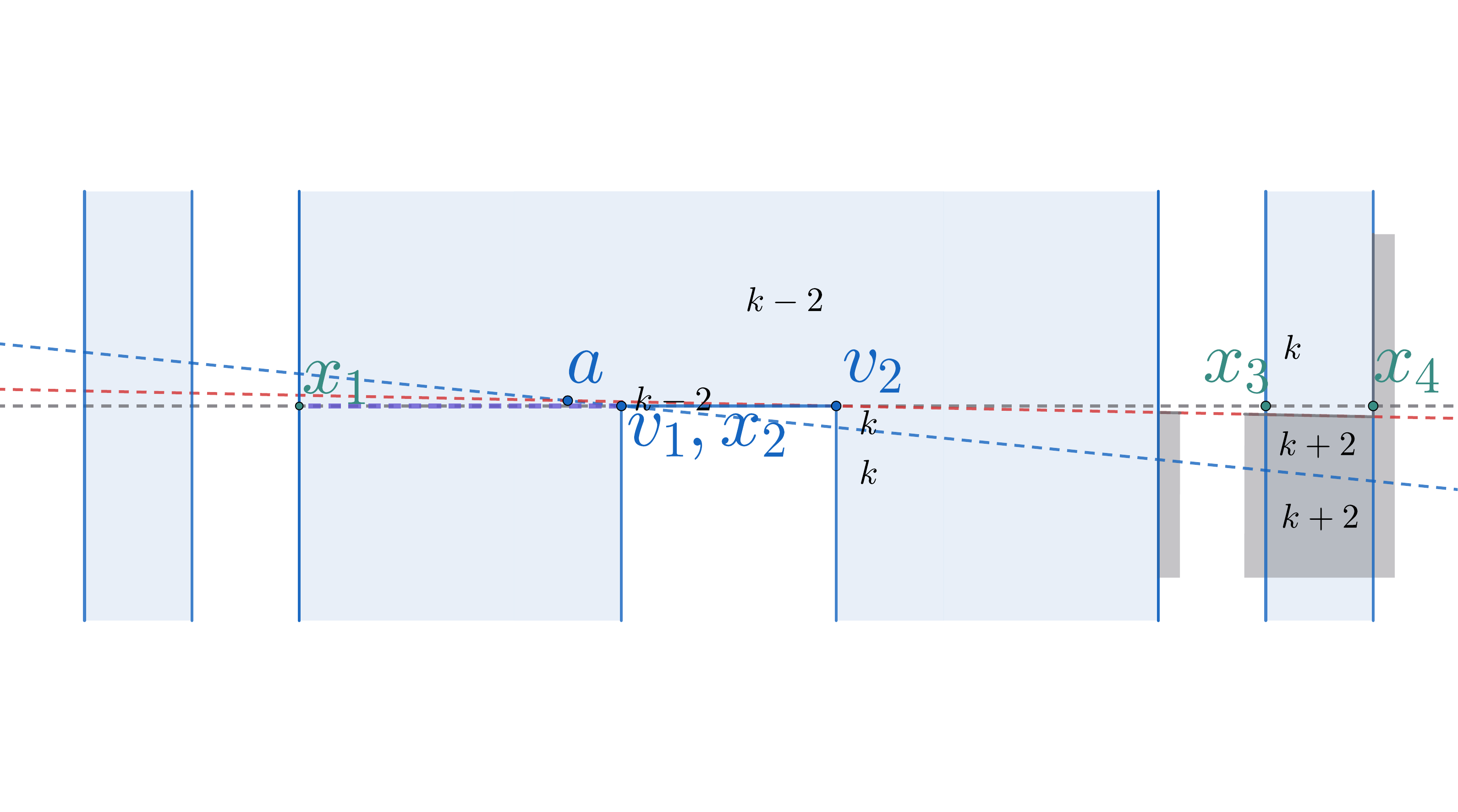}
\caption{Above $l_{g}$}\label{fig:RR-SpecialCase-A2}
\end{subfigure}
\begin{subfigure}[b]{.49\linewidth}
\includegraphics[width=\linewidth]{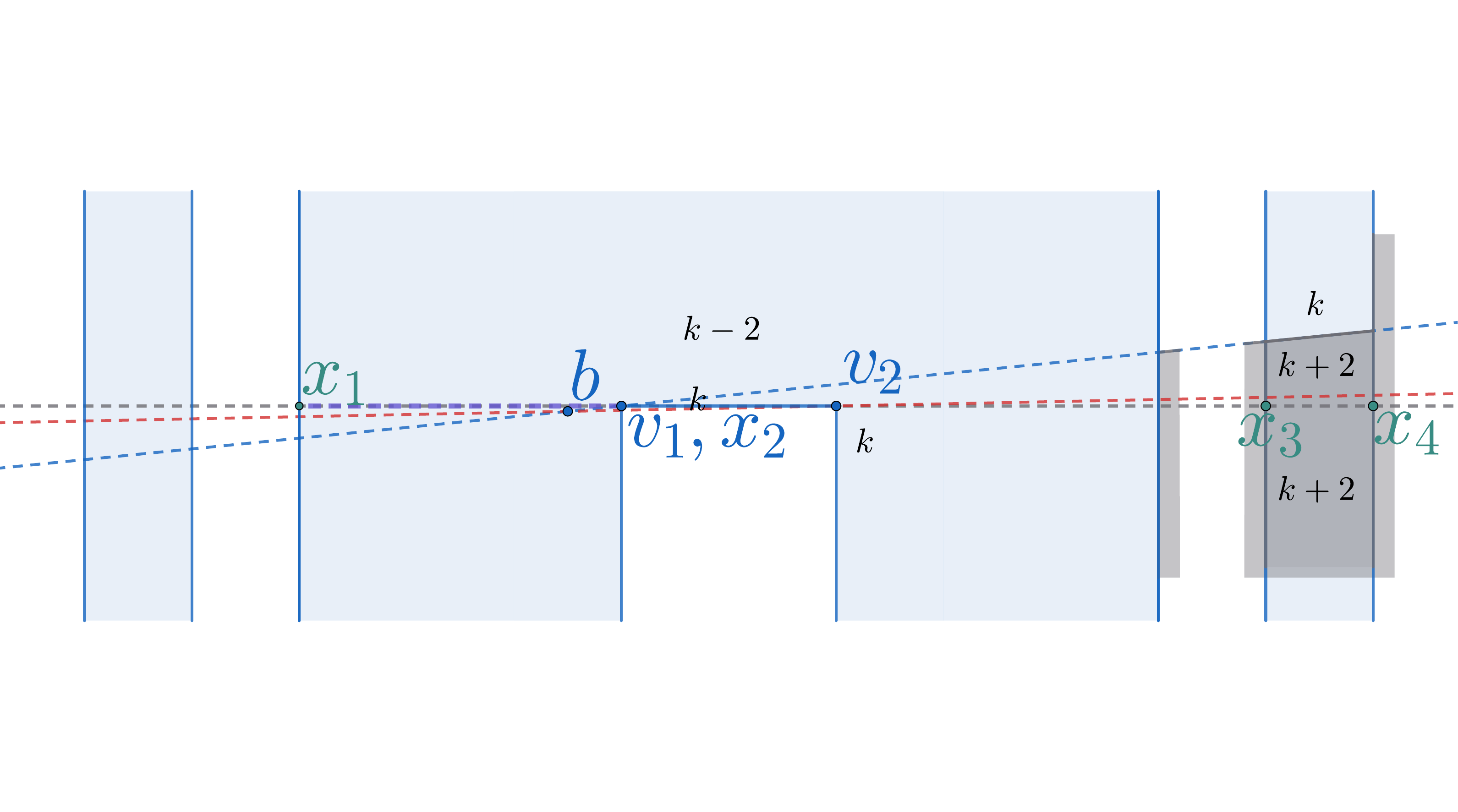}
\caption{Below $l_{g}$}\label{fig:RR-SpecialCase-B2}
\end{subfigure}

\caption{RR-SC; $Z = k - 2$, $W = k$}
\label{fig:RR-SpecialCase-2}
\end{figure}

  \begin{figure}[H]
\centering
\begin{subfigure}[b]{.49\linewidth}
\includegraphics[width=\linewidth]{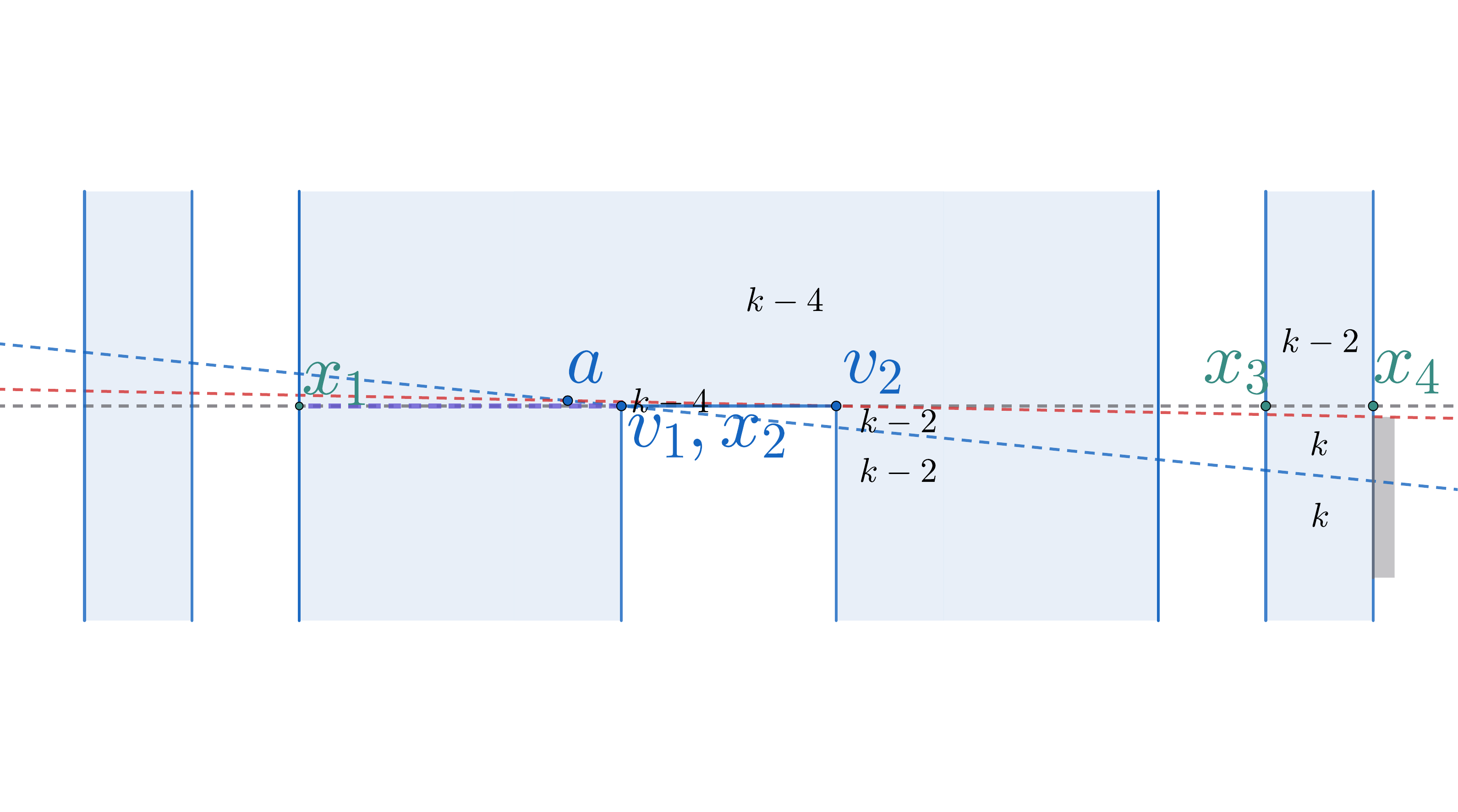}
\caption{Above $l_{g}$}\label{fig:RR-SpecialCase-A4}
\end{subfigure}
\begin{subfigure}[b]{.49\linewidth}
\includegraphics[width=\linewidth]{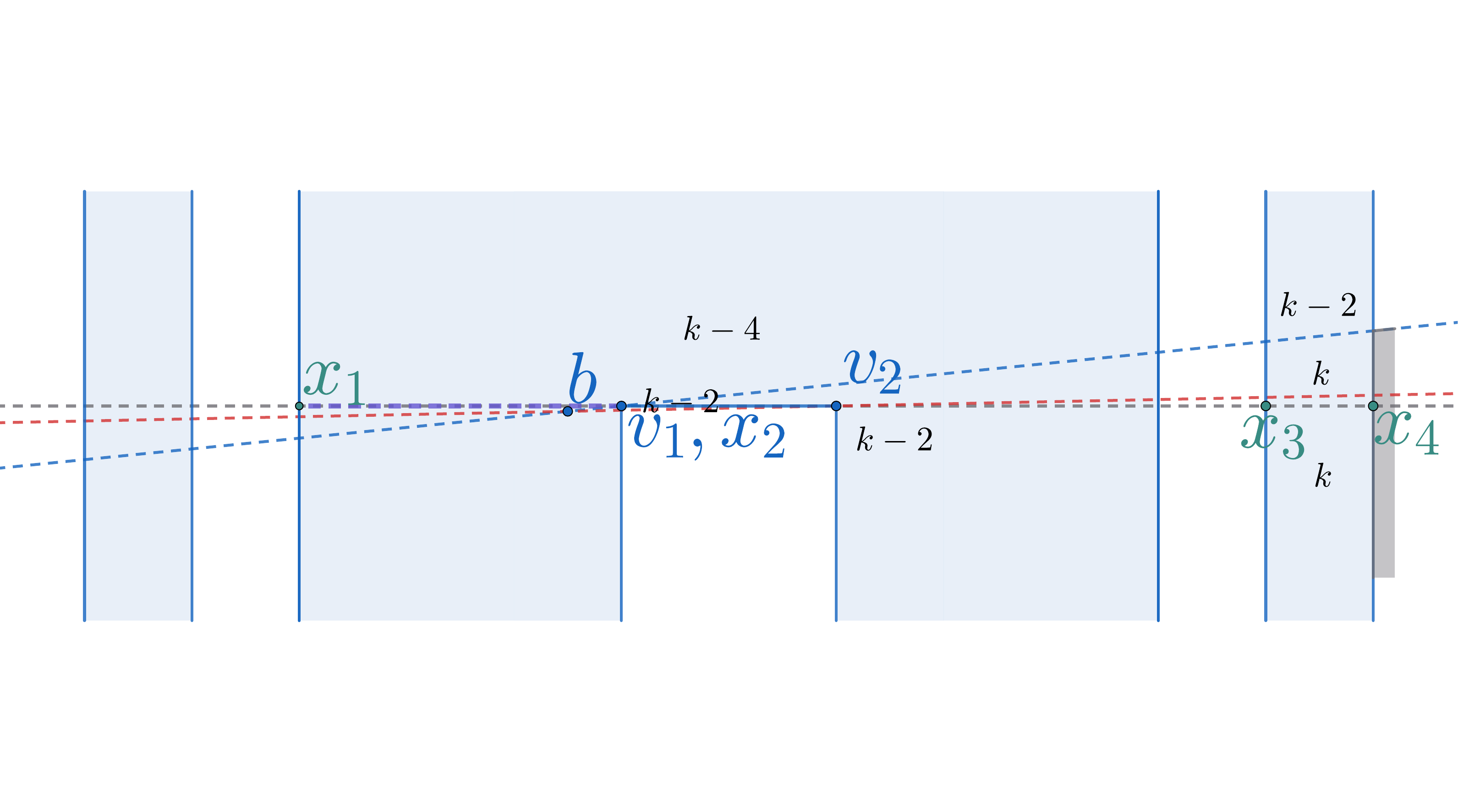}
\caption{Below $l_{g}$}\label{fig:RR-SpecialCase-B4}
\end{subfigure}

\caption{RR-SC; $Z = k - 4$, $W = k - 2$}
\label{fig:RR-SpecialCase-4}
\end{figure}

\end{proof}
\section{CR-SC}

\begin{lemma}
\label{lemma:CR-SC}

A partition line is needed in the following cases: 
\renewcommand{\labelitemi}{$\bullet$}
\label{appendix:RR-SC}
\begin{itemize}
    \item $Z = k - 1$ (Figure~\ref{fig:CR-SpecialCase-2})
    \item $Z = k - 3$ (Figure~\ref{fig:CR-SpecialCase-4})
    \item $W = k - 1$ (Figure~\ref{fig:CR-SpecialCase-4})
    \item $W = k - 3$ (Figure~\ref{fig:CR-SpecialCase-6})
\end{itemize}
\end{lemma}
\begin{proof}

See Figures~\ref{fig:CR-SpecialCase-2}-~\ref{fig:CR-SpecialCase-6}.

For $Z \geq k + 1$, $v_{2}$ and its surroundings are completely invisible, For $W \geq k + 3$, $x_{3}x_{4}$ and its surroundings are completely invisible.

for $Z \leq k - 7$, $v_{2}$ and its surroundings are completely visible. For $W \leq k - 5$,  $x_{3}x_{4}$ and its surroundings are completely visible.

 \begin{figure}[H]
\centering
\begin{subfigure}[b]{.49\linewidth}
\includegraphics[width=\linewidth]{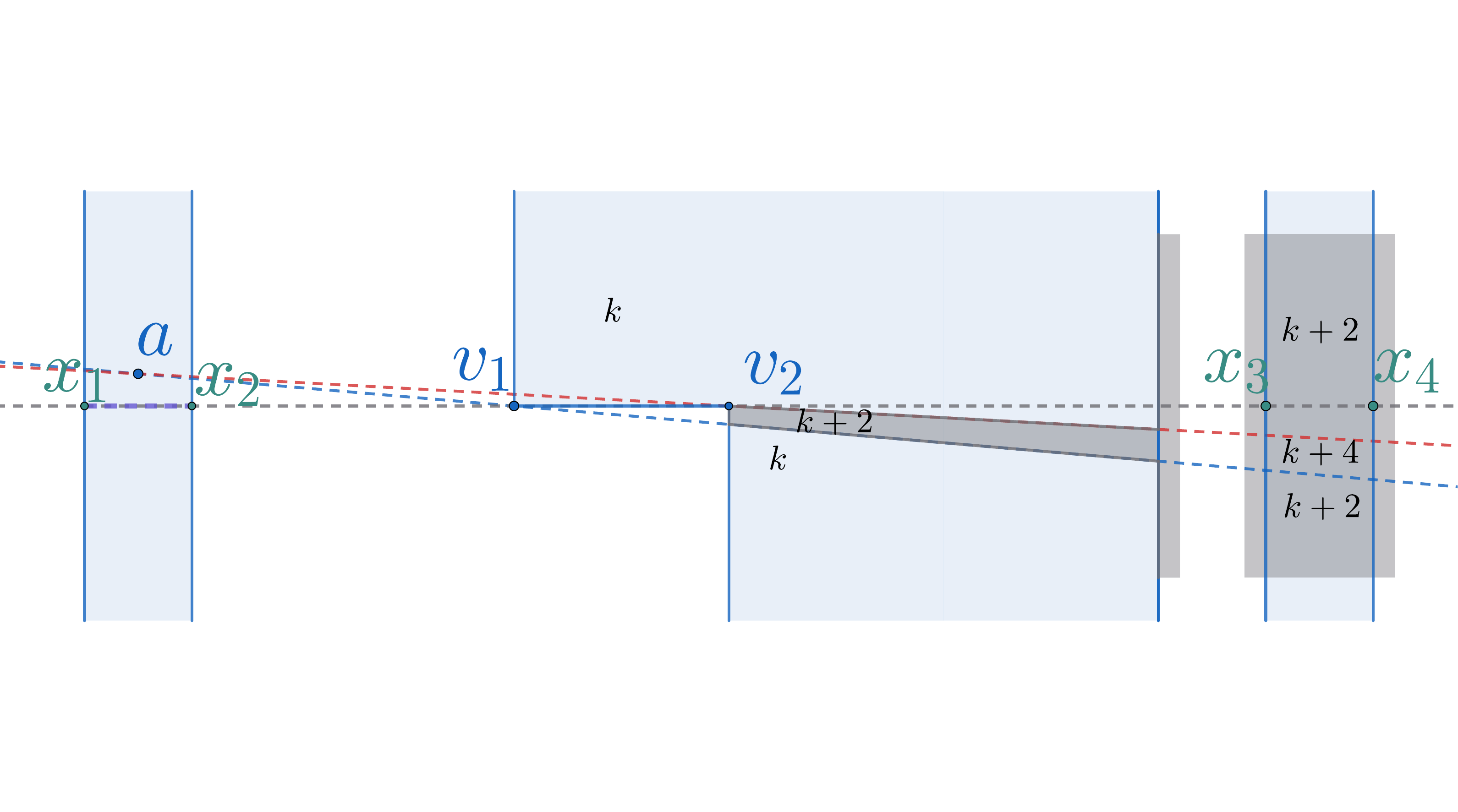}
\caption{Above $l_{g}$}\label{fig:CR-SpecialCase-A2}
\end{subfigure}
\begin{subfigure}[b]{.49\linewidth}
\includegraphics[width=\linewidth]{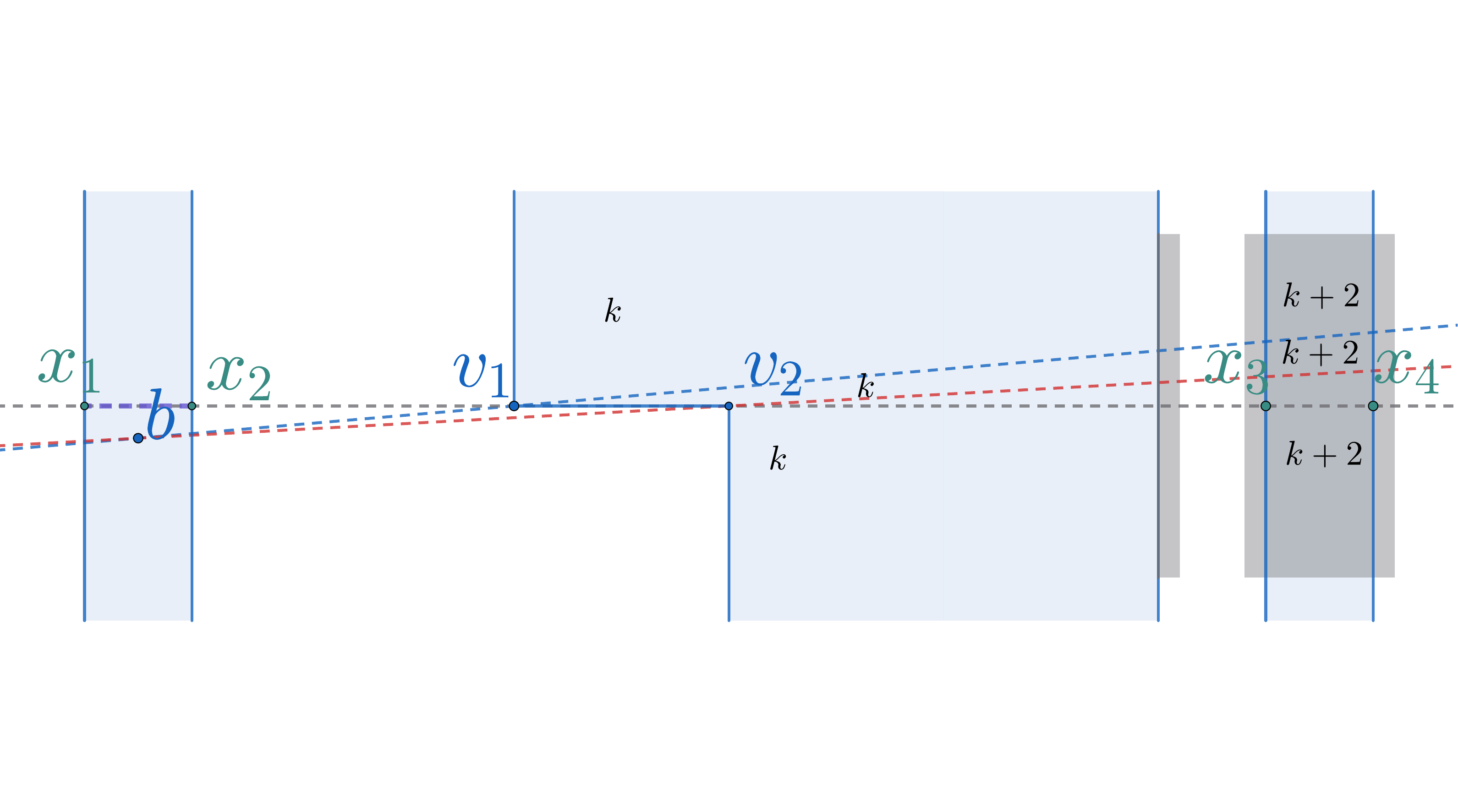}
\caption{Below $l_{g}$}\label{fig:CR-SpecialCase-B2}
\end{subfigure}

\caption{CR-SC; $Z = k - 1$, $W = k + 1$}
\label{fig:CR-SpecialCase-2}
\end{figure}

  \begin{figure}[H]
\centering
\begin{subfigure}[b]{.49\linewidth}
\includegraphics[width=\linewidth]{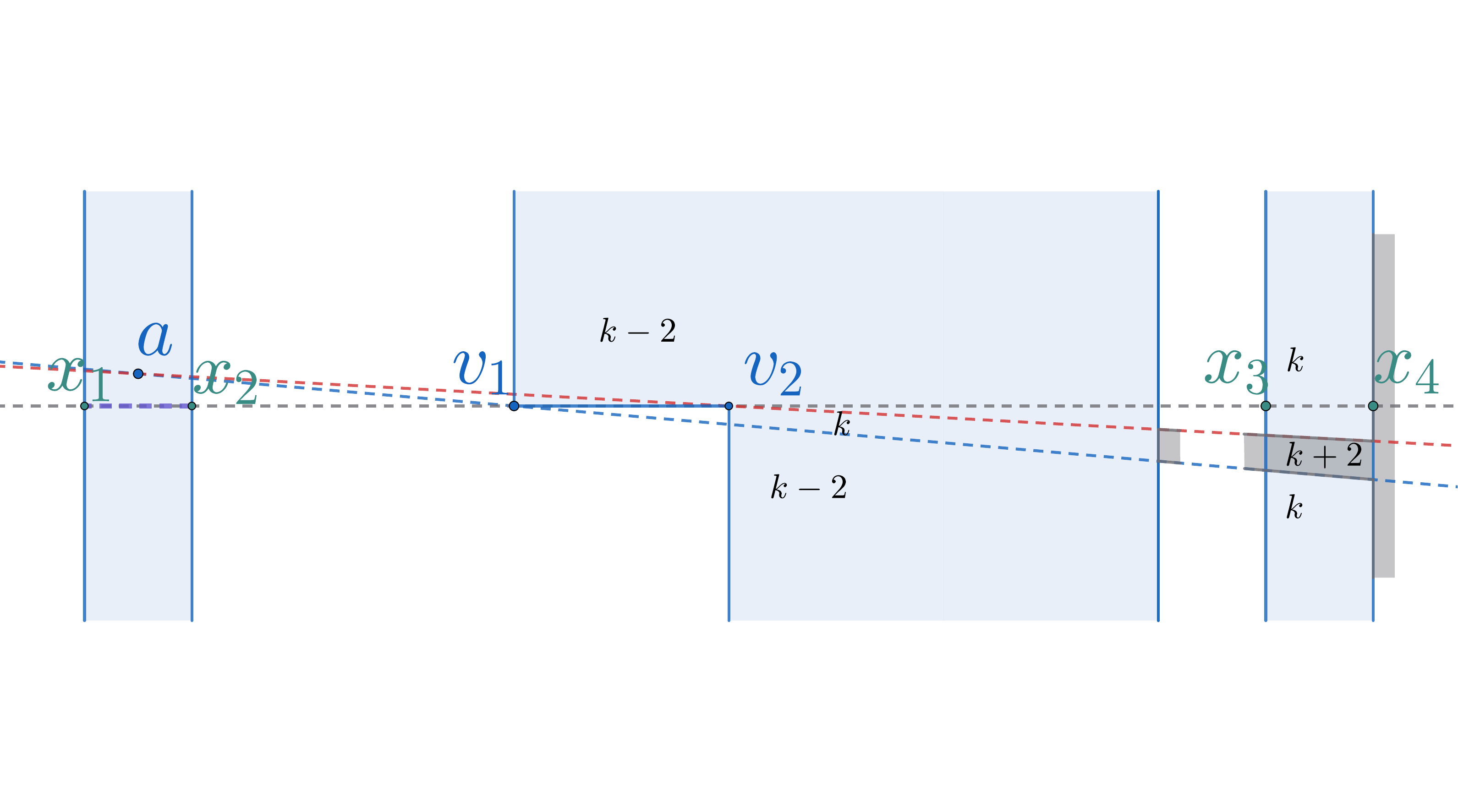}
\caption{Above $l_{g}$}\label{fig:CR-SpecialCase-A4}
\end{subfigure}
\begin{subfigure}[b]{.49\linewidth}
\includegraphics[width=\linewidth]{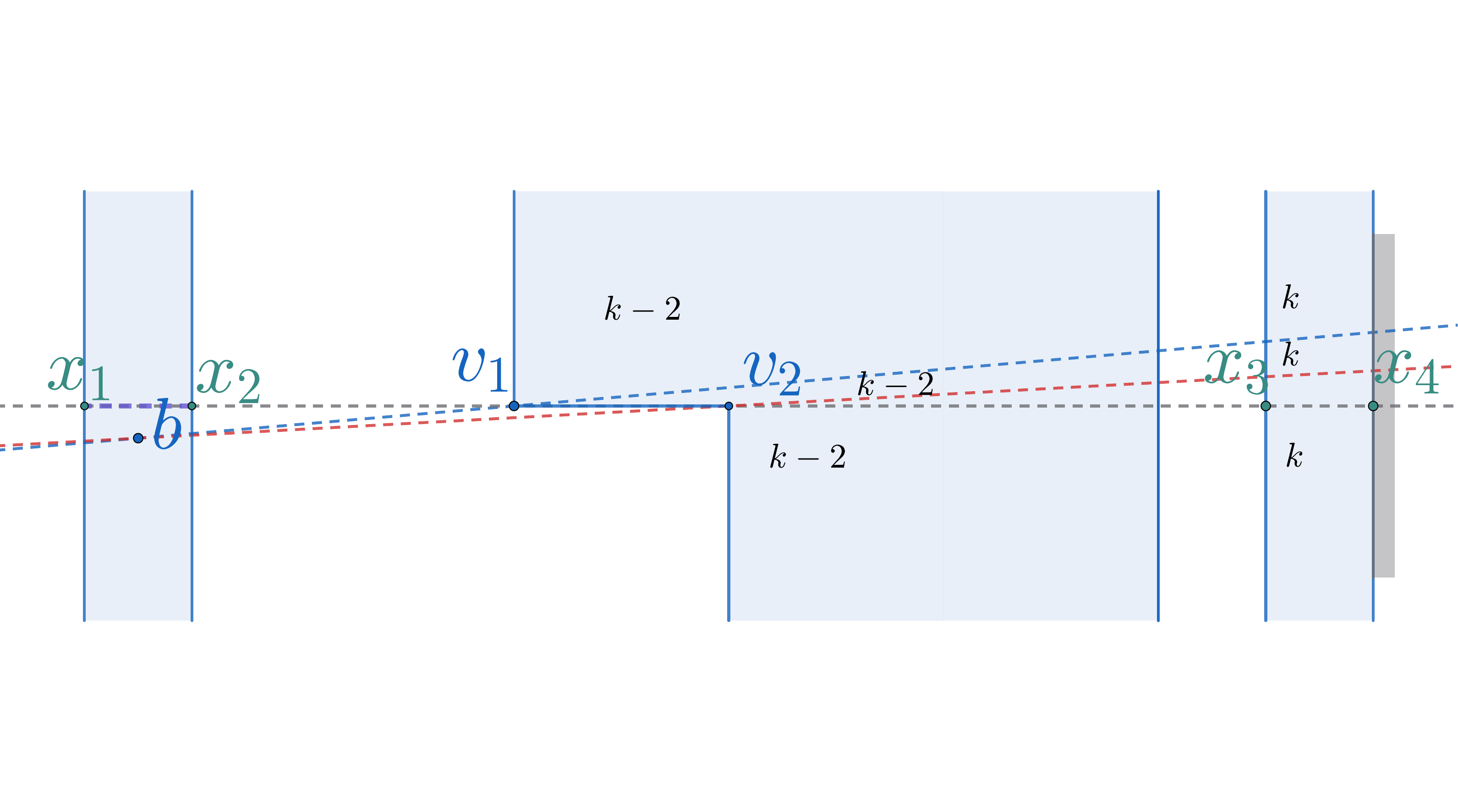}
\caption{Below $l_{g}$}\label{fig:CR-SpecialCase-B4}
\end{subfigure}

\caption{CR-SC; $Z = k - 3$, $W = k - 1$}
\label{fig:CR-SpecialCase-4}
\end{figure}

  \begin{figure}[H]
\centering
\begin{subfigure}[b]{.49\linewidth}
\includegraphics[width=\linewidth]{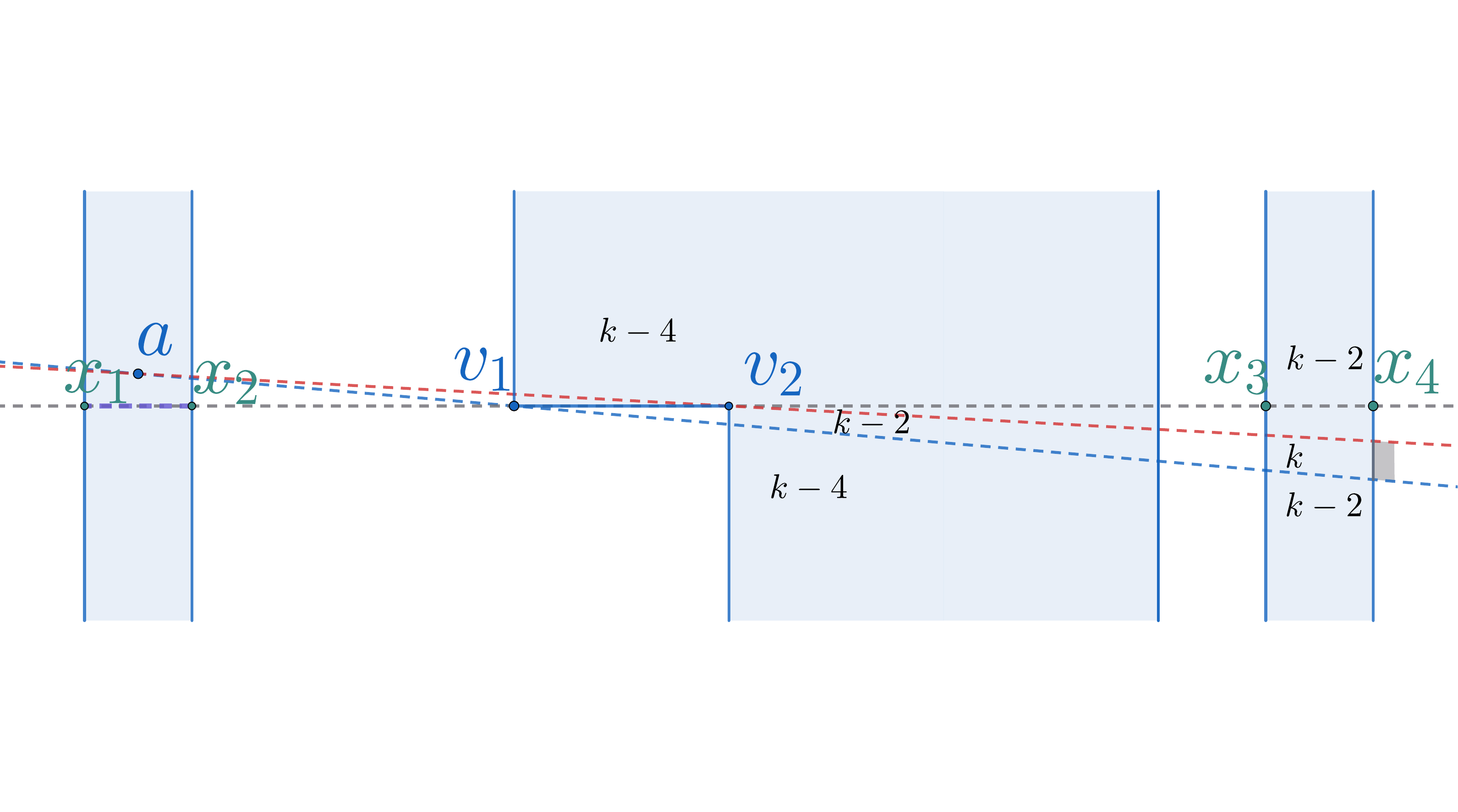}
\caption{Above $l_{g}$}\label{fig:CR-SpecialCase-A6}
\end{subfigure}
\begin{subfigure}[b]{.49\linewidth}
\includegraphics[width=\linewidth]{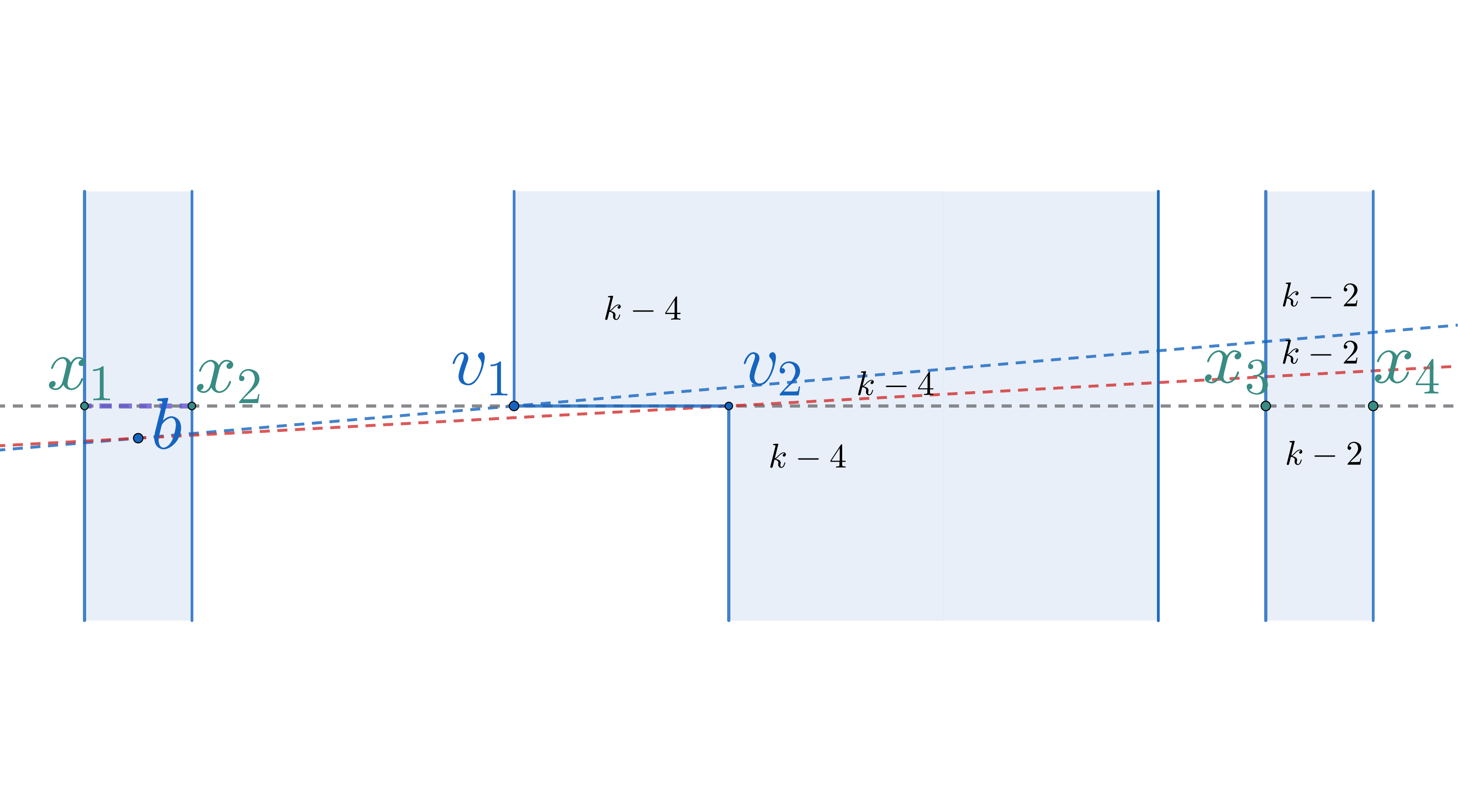}
\caption{Below $l_{g}$}\label{fig:CR-SpecialCase-B6}
\end{subfigure}

\caption{CR-SC; $Z = k - 5$, $W = k - 3$}
\label{fig:CR-SpecialCase-6}
\end{figure}

\end{proof}
\section{CC-SC}

\begin{lemma}
\label{lemma:CC-SC}
A partition line is needed in the following cases: 
\renewcommand{\labelitemi}{$\bullet$}
\label{appendix:RR-SC}
\begin{itemize}
    \item $Z = k - 1$ (Figure~\ref{fig:CC-SpecialCase-2})
    \item $W = k$ (Figure~\ref{fig:CC-SpecialCase-2})
    \item $W = k - 2$ (Figure~\ref{fig:CC-SpecialCase-4})
\end{itemize}

\end{lemma}
\begin{proof}
    See Figures~\ref{fig:CC-SpecialCase-2}-~\ref{fig:CC-SpecialCase-4}.
Note: 
For $Z \geq k + 1$, $v_{2}$ and its surroundings are completely invisible, For $W \geq k + 2$, $x_{3}x_{4}$ and its surroundings are completely invisible.

for $Z \leq k - 5$, $v_{2}$ and its surroundings are completely visible. For $W \leq k - 4$,  $x_{3}x_{4}$ and its surroundings are completely visible. 

  \begin{figure}[H]
\centering
\begin{subfigure}[b]{.49\linewidth}
\includegraphics[width=\linewidth]{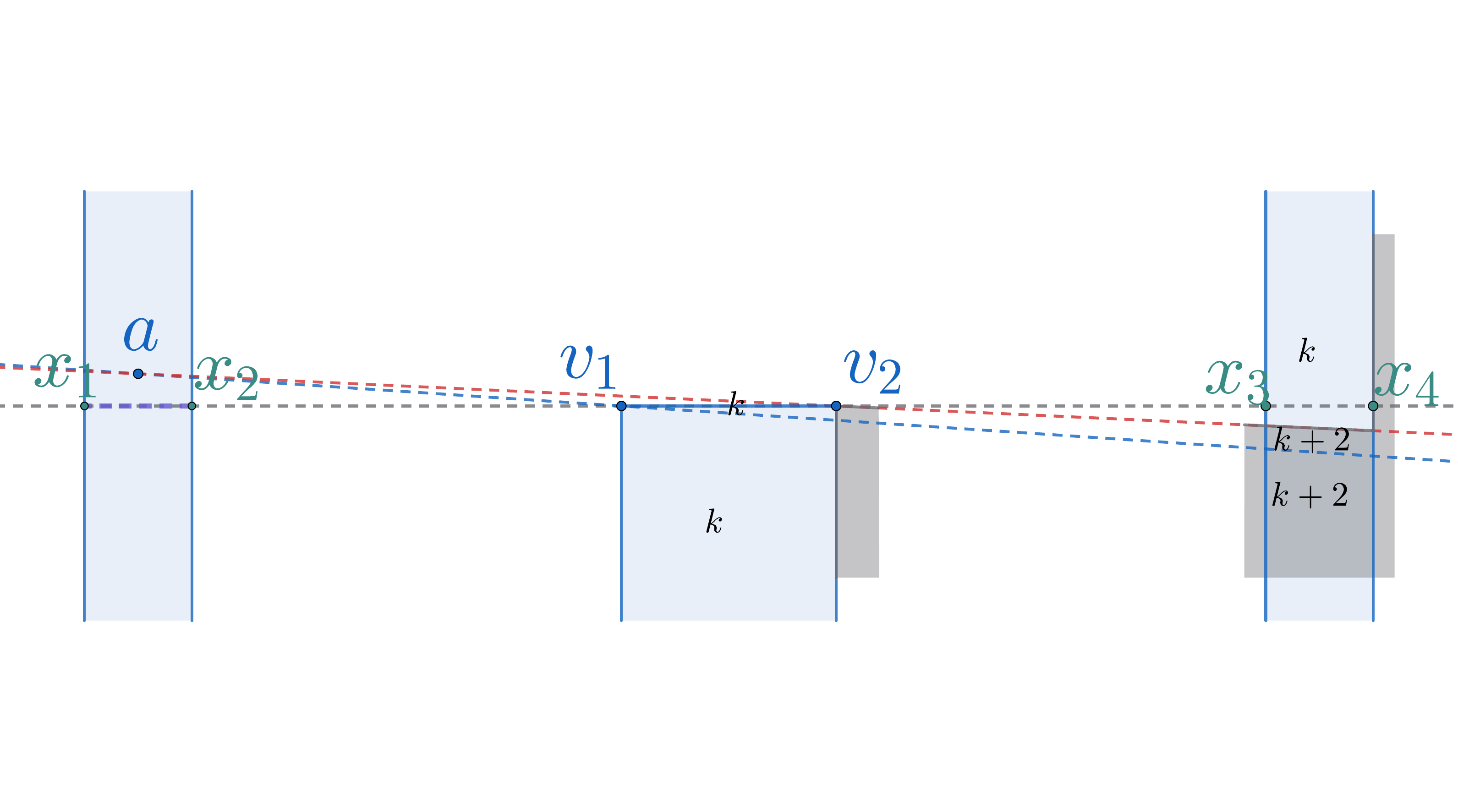}
\caption{Above $l_{g}$}\label{fig:CC-SpecialCase-A2}
\end{subfigure}
\begin{subfigure}[b]{.49\linewidth}
\includegraphics[width=\linewidth]{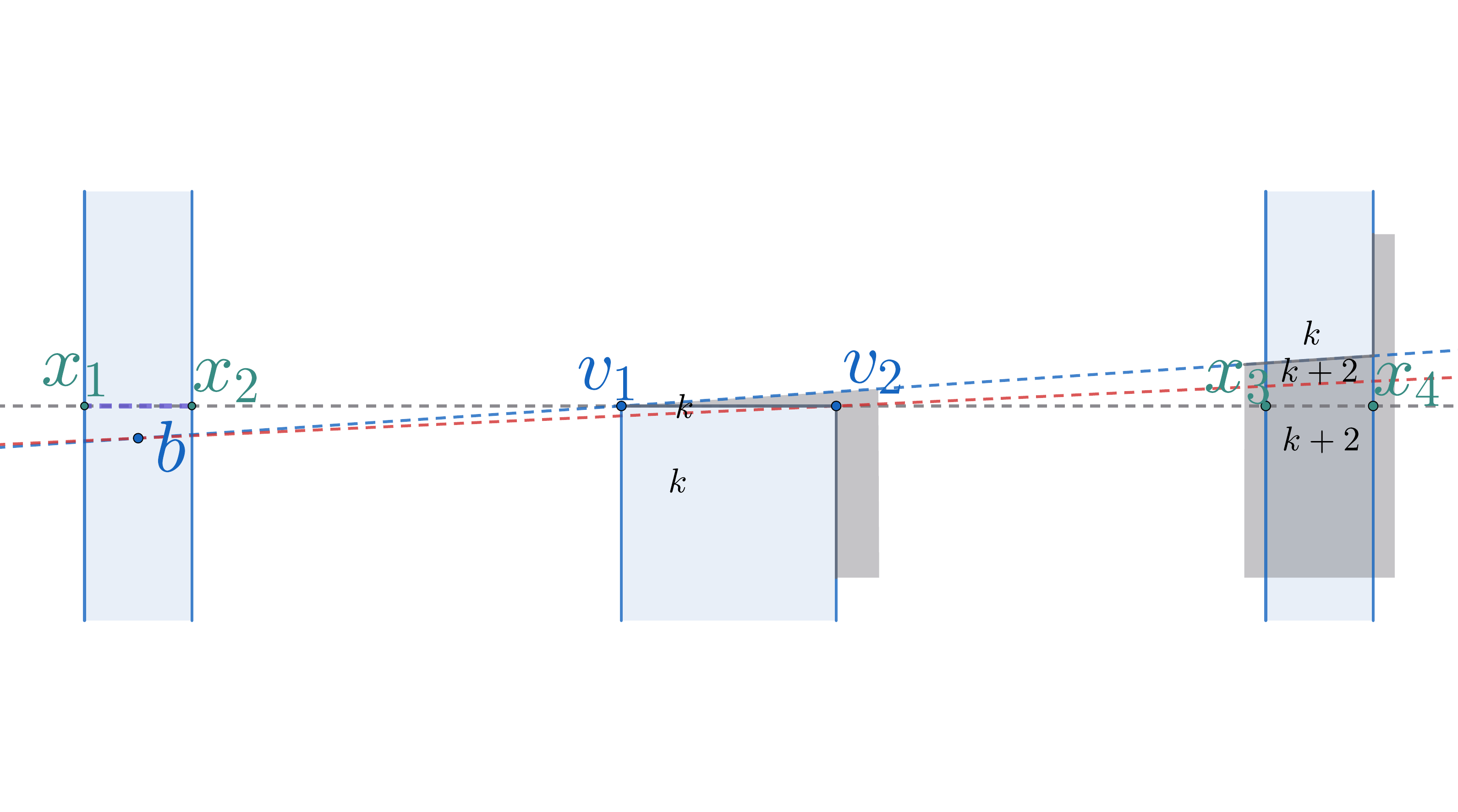}
\caption{Below $l_{g}$}\label{fig:CC-SpecialCase-B2}
\end{subfigure}

\caption{CC-SC; $Z = k - 1$, $W = k$}
\label{fig:CC-SpecialCase-2}
\end{figure}

  \begin{figure}[H]
\centering
\begin{subfigure}[b]{.49\linewidth}
\includegraphics[width=\linewidth]{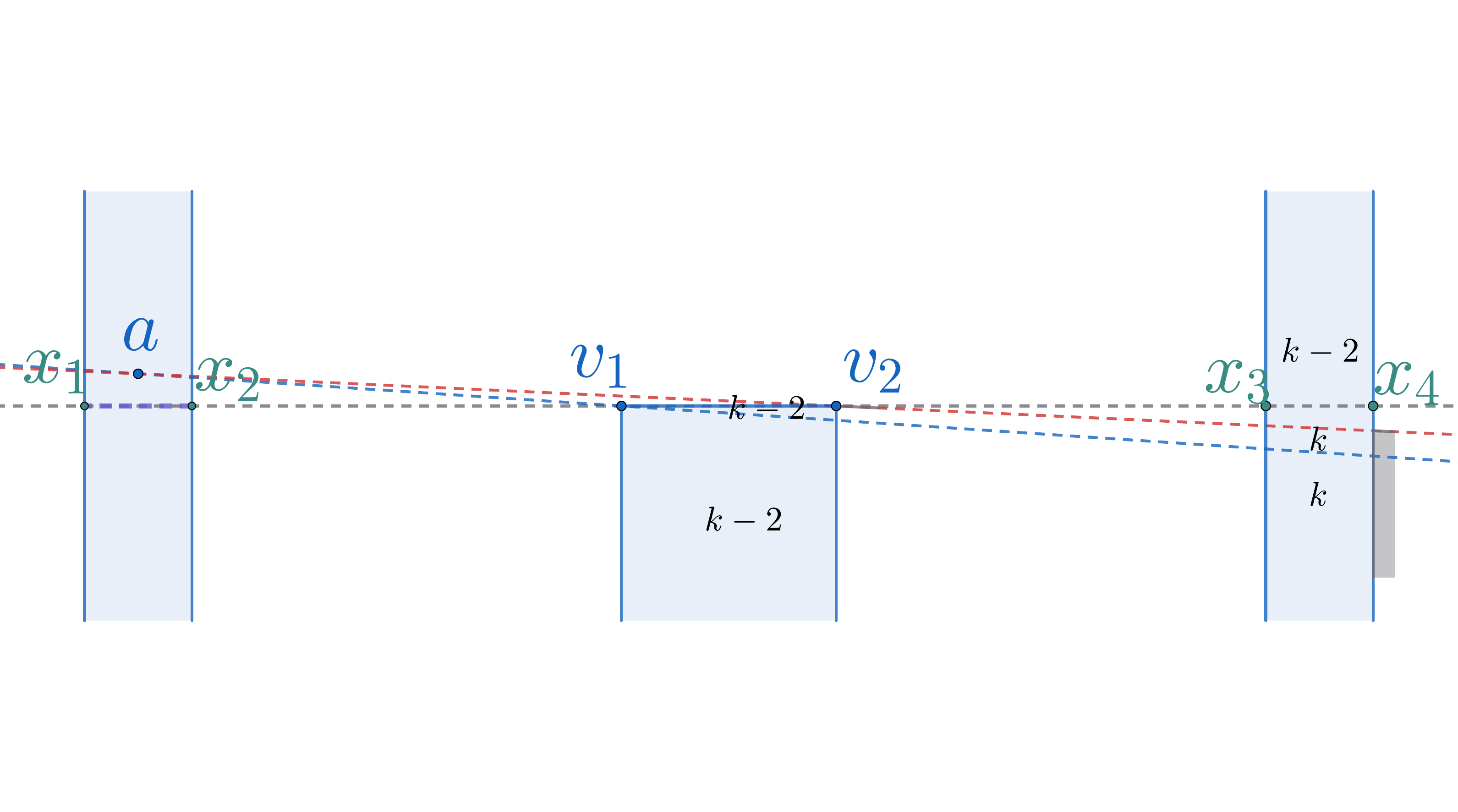}
\caption{Above $l_{g}$}\label{fig:CC-SpecialCase-A2}
\end{subfigure}
\begin{subfigure}[b]{.49\linewidth}
\includegraphics[width=\linewidth]{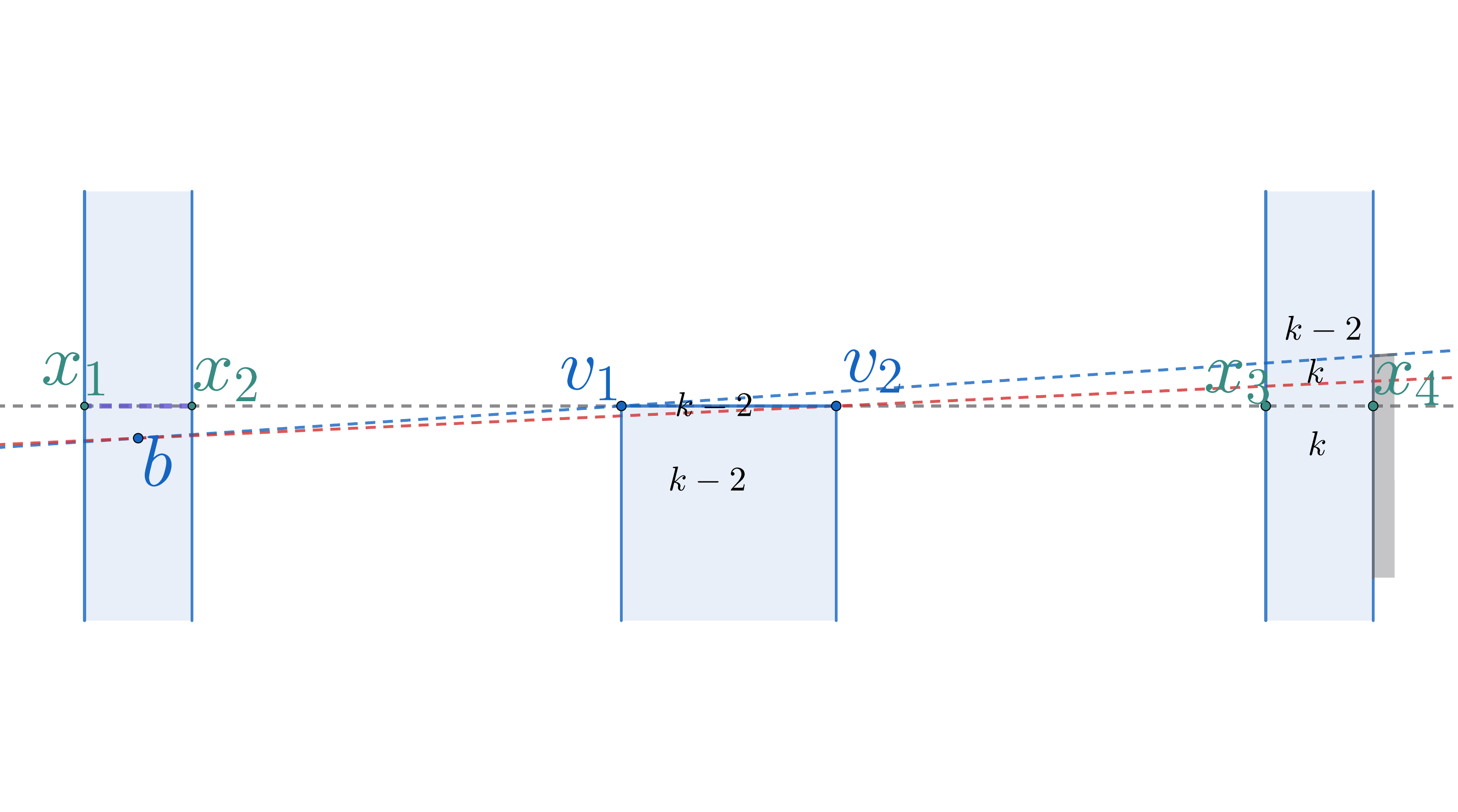}
\caption{Below $l_{g}$}\label{fig:CC-SpecialCase-B2}
\end{subfigure}

\caption{CC-SC; $Z = k - 3$, $W = k-2$}
\label{fig:CC-SpecialCase-4}
\end{figure}
\end{proof}

\newpage

\small
\begin{longtable}{llll}
\caption{Comprehensive Mapping of Cases, Subcases, Lemmas, and Figures} \label{tab:cases-comprehensive} \\
\toprule
\textbf{Case} & \textbf{Subcase} & \textbf{Lemma} & \textbf{Figure Reference} \\
\midrule
\endfirsthead

% This section defines what appears at the top of the next page(s)
\multicolumn{4}{c}{{\bfseries \tablename\ \thetable{} -- Continued from previous page}} \\
\midrule
\textbf{Case} & \textbf{Subcase} & \textbf{Lemma} & \textbf{Figure Reference} \\
\midrule
\endhead

% This section defines the footer text at the bottom of pages before the last one
\midrule
\multicolumn{4}{r}{{Continued on next page...}} \\
\endfoot

% This section defines the very end line of the entire table
\bottomrule
\endlastfoot

% --- TABLE DATA START ---

& $Z = k - 1$ &                                       & Figure~\ref{fig:CCS-generic-2}\\
                        & $W = k$     &                                       & Figure~\ref{fig:CCS-generic-2} \\
                        & $Z = k - 3$ &                                       & Figure~\ref{fig:CCS-generic-4}\\
                        & $W = k - 2$ &                                       & Figure~\ref{fig:CCS-generic-4}\\
\multirow{-5}{*}{CCS}   & $W = k - 4$ & \multirow{-5}{*}{Lemma~\ref{lemma:ccs}} & Figure~\ref{fig:CCS-generic-6} \\
\midrule
\rowcolor{lightgray}
                      & $Z = k - 1$ &                                         & Figure~\ref{fig:CCO-generic-2} \\
\rowcolor{lightgray}
                      & $W = k$     &                                         & Figure~\ref{fig:CCO-generic-2} \\
\rowcolor{lightgray}
                      & $Z = k - 3$ &                                         & Figure~\ref{fig:CCO-generic-4} \\
\rowcolor{lightgray}
                      & $W = k - 2$ &                                         & Figure~\ref{fig:CCO-generic-4} \\
\rowcolor{lightgray}
\multirow{-5}{*}{CCO} & $W = k - 4$ & \multirow{-5}{*}{Lemma~\ref{lemma:CCO}} & Figure~\ref{fig:CCO-generic-6}\\
\midrule
& $Z = k$     & \multirow{6}{*}{Lemma~\ref{lemma:CRS}} &  Figure~\ref{fig:CRO-generic2}\\
& $Z = k - 2$ & & Figure~\ref{fig:CRO-generic4} \\
& $W = k$     & & Figure~\ref{fig:CRO-generic4} \\
& $Z = k - 4$ & & Figure~\ref{fig:CRO-generic6} \\
& $W = k - 2$ & & Figure~\ref{fig:CRO-generic6} \\
\multirow{-6}{*}{CRS} & $W = k - 4$ & & Figure~\ref{fig:CRO-generic8} \\
\midrule
\rowcolor{lightgray}
                      & $Z = k$     &                                                            & Figure~\ref{fig:CRS-generic2} \\
\rowcolor{lightgray}
                      & $Z = k - 2$ &                                                            & Figure~\ref{fig:CRS-generic4} \\
\rowcolor{lightgray}
                      & $W = k$     &                                                            & Figure~\ref{fig:CRS-generic4} \\
\rowcolor{lightgray}
                      & $Z = k - 4$ &                                                            & Figure~\ref{fig:CRS-generic6} \\
\rowcolor{lightgray}
                      & $W = k - 2$ &                                                            & Figure~\ref{fig:CRS-generic6} \\
\rowcolor{lightgray}
\multirow{-6}{*}{CRO} & $W = k - 4$ & \multirow{-6}{*}{Lemma~\ref{lemma:CRO}} & Figure~\ref{fig:CRS-generic8} \\
\midrule
& $Z = k - 1$ & \multirow{5}{*}{Lemma~\ref{lemma:RCS}} & Figure~\ref{fig:RCS-generic-2} \\
& $W = k$     & & Figure~\ref{fig:RCS-generic-2} \\
& $Z = k - 3$ & & Figure~\ref{fig:RCS-generic-4} \\
& $W = k - 2$ & & Figure~\ref{fig:RCS-generic-4} \\
\multirow{-5}{*}{RCS} & $W = k - 4$ & & Figure~\ref{fig:RCS-generic-6} \\
\midrule
\rowcolor{lightgray}
                      & $Z = k - 1$ &                                                                   & Figure~\ref{fig:RCO-generic-2} \\
\rowcolor{lightgray}
                      & $W = k$     &                                                                   & Figure~\ref{fig:RCO-generic-2} \\
\rowcolor{lightgray}
                      & $Z = k - 3$ &                                                                   & Figure~\ref{fig:RCO-generic-4} \\
\rowcolor{lightgray}
                      & $W = k - 2$ &                                                                   & Figure~\ref{fig:RCO-generic-4} \\
\rowcolor{lightgray}
\multirow{-5}{*}{RCO} & $W = k - 4$ & \multirow{-5}{*}{Lemma~\ref{lemma:RCO}} & Figure~\ref{fig:RCO-generic-6} \\
\midrule
& $Z = k$     & \multirow{6}{*}{Lemma~\ref{lemma:RRS}} & Figure~\ref{fig:RRS-generic-0} \\
& $Z = k - 2$ & & Figure~\ref{fig:RRS-generic-2} \\
& $W = k$     & & Figure~\ref{fig:RRS-generic-2} \\
& $Z = k - 4$ & & Figure~\ref{fig:RRS-generic-4} \\
& $W = k - 2$ & & Figure~\ref{fig:RRS-generic-4} \\
\multirow{-6}{*}{RRS} & $W = k - 4$ & & Figure~\ref{fig:RRS-generic-6} \\
\midrule
\pagebreak
\rowcolor{lightgray}
                      & $Z = k$     &                                         & Figure~\ref{fig:RRO-generic-0} \\
\rowcolor{lightgray}
                      & $Z = k - 2$ &                                         & Figure~32 \\
\rowcolor{lightgray}
                      & $W = k$     &                                         & Figure~\ref{fig:RRO-generic-2} \\
\rowcolor{lightgray}
                      & $Z = k - 4$ &                                         & Figure~\ref{fig:RRO-generic-4} \\
\rowcolor{lightgray}
                      & $W = k - 2$ &                                         & Figure~\ref{fig:RRO-generic-4} \\
\rowcolor{lightgray}
\multirow{-6}{*}{RRO} & $W = k - 4$ & \multirow{-6}{*}{Lemma~\ref{lemma:RRO}} & Figure~\ref{fig:RRO-generic-6}\\

\midrule

& $Z = k$     & \multirow{4}{*}{Lemma~\ref{lemma:RC-SC}} & Figure~\ref{fig:RC-SpecialCase-0} \\
& $Z = k - 2$ & & Figure~\ref{fig:RC-SpecialCase-2} \\
& $W = k - 1$ & & Figure~\ref{fig:RC-SpecialCase-2} \\
\multirow{-4}{*}{RC-SC} & $W = k - 3$ & & Figure~\ref{fig:RC-SpecialCase-4} \\
\midrule
\rowcolor{lightgray}
                        & $Z = k$     &                                              & Figure~\ref{fig:RR-SpecialCase-0} \\
\rowcolor{lightgray}
                        & $Z = k - 2$ &                                              & Figure~\ref{fig:RR-SpecialCase-2} \\
\rowcolor{lightgray}
                        & $W = k$     &                                              & Figure~\ref{fig:RR-SpecialCase-2} \\
\rowcolor{lightgray}
\multirow{-4}{*}{RR-SC} & $W = k - 2$ & \multirow{-4}{*}{Lemma~\ref{lemma:RR-SC}} & Figure~\ref{fig:RR-SpecialCase-4} \\
\midrule
& $Z = k - 1$ & \multirow{4}{*}{Lemma \ref{lemma:CR-SC}} & Figure~\ref{fig:CR-SpecialCase-2} \\
& $Z = k - 3$ & & Figure~\ref{fig:CR-SpecialCase-4} \\
& $W = k - 1$ & & Figure~\ref{fig:CR-SpecialCase-4} \\
\multirow{-4}{*}{CR-SC} & $W = k - 3$ & & Figure~\ref{fig:CR-SpecialCase-6} \\
\midrule
\rowcolor{lightgray}
                        & $Z = k - 1$ &                                                                      & Figure~\ref{fig:CC-SpecialCase-2} \\
\rowcolor{lightgray}
                        & $W = k$     &                                                                      & Figure~\ref{fig:CC-SpecialCase-2} \\
\rowcolor{lightgray}
\multirow{-3}{*}{CC-SC} & $W = k - 2$ & \multirow{-3}{*}{Lemma~\ref{lemma:CC-SC}} & Figure~\ref{fig:CC-SpecialCase-4} \\\\

\end{longtable}

\end{document}